\documentclass[aps,prx,superscriptaddress,amsfonts,amsmath,amssymb,showpacs,floatfix,reprint,longbibliography]{revtex4-2}
\usepackage{algoloc}

\begin{document}
\title{\papertitle}
\author{Siddhant Midha}
\email{siddhantm@princeton.edu}
\affiliation{Princeton Quantum Initiative, Princeton University, Princeton, NJ 08544}

\author{Yifan F. Zhang}
\affiliation{Department of ECE, Princeton University, Princeton, NJ 08544}

\author{Daniel Malz}
\affiliation{Department of Physics, University of Basel, Klingelbergstrasse 82, CH-4056 Basel, Switzerland}

\author{Dmitry A. Abanin}
\affiliation{Princeton Quantum Initiative, Princeton University, Princeton, NJ 08544}
\affiliation{Department of Physics, Princeton University, Princeton, NJ 08544}
\affiliation{\'{E}cole Polytechnique F\'{e}d\'{e}rale de Lausanne (EPFL), 1015 Lausanne, Switzerland}

\author{Sarang Gopalakrishnan}
\affiliation{Princeton Quantum Initiative, Princeton University, Princeton, NJ 08544}
\affiliation{Department of ECE, Princeton University, Princeton, NJ 08544}

\begin{abstract}
    Belief propagation has recently emerged as a powerful framework for evaluating tensor networks in higher dimensions, combining computational efficiency with provable analytical guarantees.
    In this work, we develop the first end-to-end theory of tensor network belief propagation for a class of projected entangled pair states satisfying \emph{strong injectivity}. We show that when the injectivity parameter exceeds a constant threshold, BP fixed points can be found efficiently, and a cluster-corrected BP algorithm computes physical quantities to $1/\mathrm{poly}(N)$ error in $\mathrm{poly}(N)$ time for an $N$ qubit system.  We further identify a striking phenomenon we term \emph{algorithmic locality}: local perturbations of the tensor network affect the BP fixed point with an influence decaying rapidly with distance. 
    Moreover, through the cluster expansion, this locality extends to observables, implying that local expectation values can be approximated from local data with controlled accuracy. 
    We exhibit three regimes admitting strong injectivity: random Gaussian PEPS with large physical dimension, high-temperature Gibbs states of commuting two-body Hamiltonians, and fixed-depth quantum circuits subject to sufficiently strong depolarizing noise. Our bounds are worst-case sufficient conditions and do not identify sharp algorithmic transitions.
\end{abstract}

\maketitle

\section{Introduction}

Tensor network states provide a powerful framework for encoding and manipulating quantum many-body wave functions~\cite{white1992density,white1993density,schollwock2011density,orus2014practical,cirac2021matrix,verstraete2004matrix,vidal2007entanglement,verstraete2006faithful,verstraete2004renormalization,PerezGarcia2007,vidal2003efficient,Vidal2004}. In particular, projected entangled pair states (PEPS)~\cite{verstraete2004renormalization} generalize matrix product states~\cite{white1992density,white1993density,schollwock2011density,verstraete2006faithful} to higher-dimensional lattices, offering a natural description of ground states of local Hamiltonians with area-law entanglement~\cite{orus2014practical,cirac2021matrix}. While PEPS have been extensively studied numerically~\cite{jordan2008simulation,verstraete2004renormalization,orus2012exploring,orus2009simulation,verstraete2004renormalization, lubasch2014unifying, murg2007variational,levin2007tensor,evenbly2015tensor,corboz2016,Xie_2012_HOTRG} and have yielded significant insights into the structure of quantum phases, establishing their rigorous properties on general graphs remains challenging. In particular, PEPS defined on graphs with loops lack the simple iterative structure enjoyed by matrix product states~\cite{white1992density,white1993density,schollwock2011density,verstraete2006faithful} on one-dimensional chains. As a result, many analytical and computational tools available in one dimension do not directly generalize.

\begin{figure}[!t]\label{fig:epsilon_lppl}
  \centering

  \noindent\begin{minipage}{0.72\linewidth}
    \centering
    \llap{\textbf{(a)}\hspace{0.6em}}%
    \resizebox{0.8\linewidth}{!}{%
      \begin{tikzpicture}[
  baseline={([yshift=-0.65ex] current bounding box.center)},
  tensor/.style={
    draw, line width=0.3pt,
    fill=blue!20!tensorcolor!75!black,
    rounded corners=1pt, minimum size=2mm,
    inner sep=0pt, outer sep=0pt
  },
  tensorA/.style={
    draw, line width=0.35pt,
    fill=blue!30!tensorcolor!85!black,
    rounded corners=1pt, minimum size=2mm,
    inner sep=0pt, outer sep=0pt
  }
]

\def\dx{4mm}
\def\dy{4mm}
\def\ls{0.3}

\def\cx{2}   
\def\cy{1}   
\def\dmax{2.0616} 

\foreach \r in {0,...,2}{
  \foreach \c in {0,...,4}{
    \node[tensor] (T-\c-\r) at (\c*\dx, -\r*\dy) {};
  }
}

\node[tensorA] (T-\cx-\cy) at (\cx*\dx, -\cy*\dy) {};

\fill[black!90]
  ($(T-\cx-\cy)+({0.00*\ls},{0.35*\ls})$) --
  ($(T-\cx-\cy)+({0.18*\ls},{0.35*\ls})$) --
  ($(T-\cx-\cy)+({0.05*\ls},{0.05*\ls})$) --
  ($(T-\cx-\cy)+({0.22*\ls},{0.05*\ls})$) --
  ($(T-\cx-\cy)+({-0.05*\ls},{-0.35*\ls})$) --
  ($(T-\cx-\cy)+({0.02*\ls},{-0.05*\ls})$) --
  ($(T-\cx-\cy)+({-0.15*\ls},{-0.05*\ls})$) -- cycle;

\foreach \r in {0,...,2}{
  \foreach \c in {0,...,3}{
    \pgfmathtruncatemacro{\cp}{\c+1}

    \pgfmathsetmacro{\mx}{\c + 0.5}
    \pgfmathsetmacro{\my}{\r}

    \pgfmathsetmacro{\dist}{veclen(\mx-\cx,\my-\cy)}
    \pgfmathsetmacro{\t}{\dist/\dmax}

    \pgfmathtruncatemacro{\shade}{106 - 24*\t}
    \pgfmathsetmacro{\thick}{1.8 - 0.20*\dist}
    \pgfmathsetmacro{\opa}{1.0 - 0.55*\t}

    \draw[
      draw=black!\shade,
      line width=\thick pt,
      opacity=\opa
    ] (T-\c-\r) -- (T-\cp-\r);
  }
}

\foreach \c in {0,...,4}{
  \foreach \r in {0,...,1}{
    \pgfmathtruncatemacro{\rp}{\r+1}

    \pgfmathsetmacro{\mx}{\c}
    \pgfmathsetmacro{\my}{\r + 0.5}

    \pgfmathsetmacro{\dist}{veclen(\mx-\cx,\my-\cy)}
    \pgfmathsetmacro{\t}{\dist/\dmax}

    \pgfmathtruncatemacro{\shade}{106 - 24*\t}
    \pgfmathsetmacro{\thick}{1.8 - 0.40*\dist}
    \pgfmathsetmacro{\opa}{1.0 - 0.55*\t}

    \draw[
      draw=black!\shade,
      line width=\thick pt,
      opacity=\opa
    ] (T-\c-\r) -- (T-\c-\rp);
  }
}
\end{tikzpicture}%
    }
  \end{minipage}


  \noindent\begin{minipage}{0.72\linewidth}
    \centering
    \llap{\textbf{(b)}\hspace{0.6em}}%
    \includegraphics[width=\linewidth]{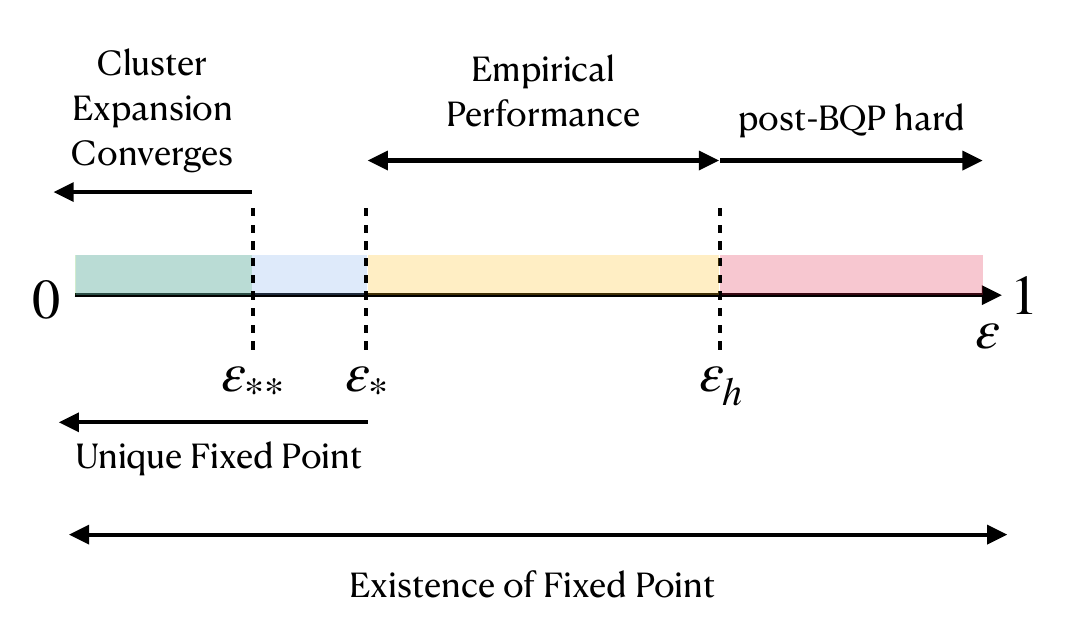}
  \end{minipage}

  \par
  \vspace{-4mm}
   \caption{(a) \emph{Algorithmic locality in tensor networks}: The effect of a perturbation at the center of the network on the fixed-point \emph{messages} living on edges of the graph decays exponentially with distance from perturbation. Loops (see \autoref{eq:loopexpmaintext}) and clusters (see \autoref{eq:clusterexpmaintext}) built out of the fixed-point messages inherit the locality subsequently. (b) \emph{Phase diagram of injective PEPS}: \autoref{thm:fixed-points} shows existence (for all $0 \leq \varepsilon < 1$) and uniqueness (for $\varepsilon < \varepsilon_{*} = \order{1/\Delta}$) of fixed points, where $\Delta$ is the degree of the graph. \autoref{thm:complexity_maintext} shows convergence of cluster expansion for $\varepsilon < \varepsilon_{**} = \order{\min\{1/D, (D/\Delta)^{\Delta/2}\}}$ leading to efficient classical simulation algorithms, with $D$ being the bond dimension of the PEPS. \autoref{thm:BPlocality_maintext} ties these results to prove algorithmic locality of tensor network belief propagation. The TN-BP algorithm performs well empirically in the intermediate regime. The hardness threshold for estimating local observables depicted $\varepsilon_h$ is from Ref.~\cite{harley2025computational}. } 
  \label{fig:fig_main}
\end{figure}
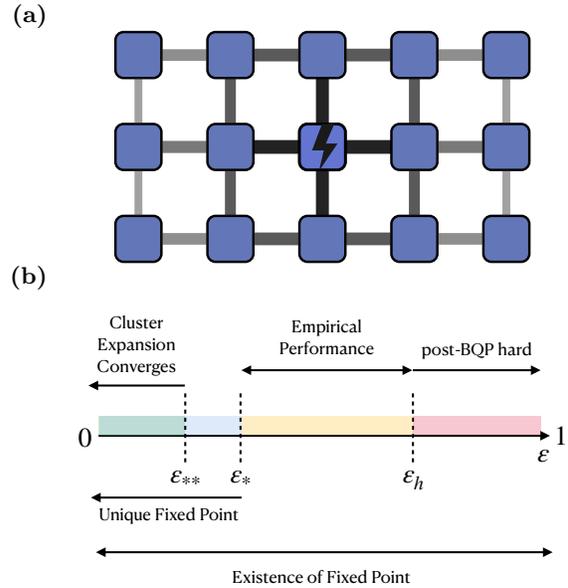

A framework for reasoning about PEPS on arbitrary graphs has recently emerged in the form of \emph{belief propagation} (BP)~\cite{robeva2016duality, BPsimpleupdate,sahu2022efficient, BP_gauging}, a message-passing algorithm originally developed for inference on graphical models~\cite{Gallager1962,pearl1988probabilistic,mceliece1998turbo,yedidia2003,Yedidia2005} deeply rooted in statistical physics~\cite{bethe1935statistical,Peierls1936,Thouless1987,Mezard1986cavity,Klein1979,Katsura1979,Nakanishi1981}. While classical belief propagation has long been placed on a rigorous footing in certain regimes~\cite{Weiss2000,Tatikonda2002,NIPS2002,Heskes2004,Mooij2005,chertkov2008exactness,Sudderth2008}, recent developments have extended this program to tensor networks through loop and cluster expansions~\cite{Evenbly2024,midha2025beyond,midha2026,gray2025,park2025}, which systematically account for the errors incurred by BP. In this picture, fixed points of the BP message-passing equations encode leading-order information about reduced density matrices of the PEPS on local regions~\cite{Evenbly2024,BP_gauging}, while the convergence properties of BP based expansions are closely tied to the decay of correlations in the underlying quantum state~\cite{midha2026}. A key advance enabling this perspective is the use of \emph{cluster expansions} from lattice statistical mechanics, which systematically decompose tensor network contractions into contributions from localized excitations or ``clusters'' built on top of the BP fixed points~\cite{midha2025beyond,midha2026}.

Despite its empirical success in quantum many-body applications~\cite{tindall2025dynamics,BP_gauging,rudolph2025simulatingsamplingquantumcircuits, park2025, liao2023simulation,sahu2022efficient}, the theoretical foundations of TN-BP remain incomplete. The cluster expansion technology only provides guarantees given the existence of a suitable BP fixed-point~\cite{midha2025beyond}. In practice, the BP message-passing equations used to compute fixed points of the tensor network are implemented as iterative updates, but their convergence is not guaranteed. Moreover, the cluster expansion framework yields controlled results when the contributions from \emph{loop tensors} decay sufficiently rapidly with loop size---a condition called ``decay of loops" that is expected to hold in physically relevant regimes (e.g., deep in a gapped phase)~\cite{midha2026}, but has not yet been established rigorously.
\subsection{Our results}
In this work we put tensor-network belief propagation (TN-BP) on fully rigorous ground by proving the following results. Working with PEPS on arbitrary graphs satisfying an injectivity condition, we prove that BP message-passing admits fixed points, and for a subclass of \emph{strongly injective} PEPS the fixed point is unique (\autoref{thm:fixed-points}). We further show that these PEPS satisfy a ``decay of loops" property controlling corrections to BP (\autoref{lem:decayofloops}). Using the cluster-expansion framework, we can then construct polynomial-time classical algorithms to compute properties such as local expectation values and connected correlation functions (\autoref{thm:complexity_maintext}).

Our convergence results, as a function of the injectivity strength $\varepsilon$ (formally defined later), are summarized in \autoref{fig:fig_main}(b). $\varepsilon = 1$ corresponds to non-injective PEPS and $\varepsilon = 0$ correspond to the ``maximally'' injective PEPS where all singular values are equal. While we show that a BP fixed point exists for all $\varepsilon  < 1$, it is known that contracting PEPS is $\mathsf{postBQP}$-hard~\cite{malz2025computational, harley2025computational} when $\varepsilon$ is above a constant threshold $\varepsilon_h$. Therefore, we do not expect BP to form a good approximation in this regime. On the other hand, we show two upper-bound thresholds $\varepsilon_{*}$ and $\varepsilon_{**}$: below $\varepsilon_{*}$ the fixed point is unique and can be efficiently found, and below $\varepsilon_{**} \le \varepsilon_{*}$ cluster-corrected BP can contract PEPS in polynomial-time, to inverse polynomial multiplicative error. \review{We note that the constants $\varepsilon_{*(*)}$ establish sufficient conditions for convergence, accounting for \emph{worst-case} instance. \autoref{fig:fig_main}(b) thus provides a phase diagram of rigorous algorithmic convergence.}

Our proof methodology reveals a striking phenomenon, which we call \emph{algorithmic locality}~(\autoref{thm:BPlocality_maintext}). Specifically, local perturbations of tensors affect BP messages only within a neighborhood, with influence that decays rapidly with distance \autoref{fig:fig_main}(a). This has both theoretical and practical consequences. On the theoretical side, algorithmic locality yields a clustering theorem analogous to exponential decay of correlations~\cite{hastings2006spectral}. On the practical side, it enables significant computational speedups when contracting multiple tensor networks that differ only locally: it suffices to update messages within the neighborhood of the perturbation.

{The injectivity thresholds can be viewed as local structural conditions on the tensor network allowing contractability on any bounded-degree graph.} We then study examples of quantum states that satisfy (strong) injectivity. We show this for (i) random PEPS with i.i.d. complex Gaussian entries with large local dimension (\autoref{thm:randommain}), (ii) Gibbs states of commuting local Hamiltonians at high temperature (\autoref{thm:gibbsmain}) and (iii) sufficiently noisy quantum circuits (\autoref{thm:noisymain}). 

To the best of our knowledge, this work provides the first rigorous proof, in the setting of quantum tensor network states, showing all of the following simultaneously: the BP fixed point can be found efficiently, the ``decay of loops" condition is satisfied, and the cluster expansion thus converges.  Furthermore, the graph-locality of our work makes our results directly relevant to quantum information tasks that are not constrained by geometric locality, such as simulating systems with $k$-local interactions~\cite{lloyd1996universal,Berry_2006,Low_2019,Childs_2009,Berry_2015,Childs_2011,Childs_2017,clinton2021hamiltonian} and decoding quantum low-density parity-check (LDPC) codes~\cite{leverrier2022quantumtannercodes, panteleev2022asymptoticallygoodquantumlocally, Fawzi_2018, Leverrier_2015, Tillich_2014, Roffe_2020, poulin2008iterativedecodingsparsequantum}.

\subsection{Related work}

Related versions of algorithmic locality have been studied in the PEPS literature~\cite{schwarz2017approximating,schuch2011classifying,anshu2016local,harley2025computational}. In particular, Ref.~\cite{schwarz2017approximating} provides quasipolynomial time algorithms for strongly injective PEPS in two-dimensions, while Ref.~\cite{anshu2016local} accesses local PEPS information from small patches by semidefinite bounds. In contrast to many existing works, our analysis relies only on \emph{graph locality} and does not assume any embedding in Euclidean space.

Prior rigorous results on tensor network BP~\cite{midha2025beyond,midha2026} are \emph{conditional}, requiring the existence of a suitable fixed point, and a decay of loops property around that fixed point. Using the injectivity parameter $\varepsilon$ as a sufficient condition, we show that these fixed-points exist, and are unique for strongly injective PEPS. We further prove the decay of loops condition, rendering the cluster expansion convergent. As a result, we lift the cluster-corrected BP algorithm to a rigorous end-to-end theory for strongly injective PEPS.

The computational complexity of PEPS has been investigated using random PEPS as a tractable playground. Herein, Ref.~\cite{GarciaRandomPEPS} used statistical mechanics arguments to study the typical complexity of random PEPS. Ref.~\cite{Evenbly2024} numerically demonstrated success of tensor network BP on random PEPS. Our results take a complementary approach of large local Hilbert space dimension, and unify these observations into a rigorous phase of algorithmic tractability of random PEPS using the lens of injectivity.

Classical simulability of noisy quantum circuits has been approached through several distinct mechanisms, include Pauli path truncations~\cite{Aharonov_2023}, percolation arguments~\cite{TrivediCiracNoise,nelson2026polynomialtimeclassicalsimulationnoisy}, and decay of conditional correlations~\cite{zhang2025classicallysamplingnoisyquantum}. Our noisy circuit result takes the circuit and maps it to an exact PEPS which is strongly injective above a constant noise threshold.

\section{Preliminaries}
\subsection{Projected entangled pair states}
A PEPS on a graph $G=(V,E)$ represents a quantum state $|\psi\rangle$ through a collection of local tensors $\{T_v\}_{v\in V}$, one for each vertex. Each tensor $T_v$ has $|\NN(v)|$ virtual legs (where $\NN(v)$ denotes the neighbors of $v$) and one physical leg. The virtual legs, each of dimension $D$ (the bond dimension), are contracted with neighboring tensors along the edges of $G$, while the physical legs, each of dimension $d$, correspond to the local Hilbert spaces of the quantum state. We denote the maximum degree by $\Delta := \max_v |\NN(v)|$ and assume throughout that both $\Delta$ and $D$ are $\order{1}$.

A central notion in our analysis is \emph{injectivity} of the local tensors. Viewing $T_v$ as a linear map from the virtual space to the physical space, $T_v: \bigotimes_{n\in \NN(v)}\HH_{(n,v)} \to \HH_v$ with $\HH_{(n,v)} \cong \mathbb{C}^D$ and $\HH_v \cong \mathbb{C}^d$, we say $T_v$ is injective if this map has trivial kernel. Injectivity requires $d \ge D^{|\NN(v)|}$, which generically holds after coarse-graining when $D = \order{1}$. The PEPS $\psi = \{T_v\}_v$ is injective if every local tensor is injective. For any injective PEPS $\psi$, there is a local, frustration-free parent Hamiltonian for which $\psi$ is the unique ground state~\cite{cirac2021matrix}.

To quantify a degree of injectivity, we introduce the notion of $\delta-$injectivity. For each tensor $T_v$, let $\{\lambda_1 \ge \lambda_2 \ge \cdots \ge \lambda_{D^{|\NN(v)|}}\}$ be the singular values of $T_v$ viewed as a matrix from virtual to physical space. 
We normalize so that $\lambda_1 = 1$ and define the \emph{injectivity parameter} $\delta$ by $\lambda_{\min} \ge \delta$. The tensor is $\delta$-injective if all singular values lie in $[\delta, 1]$. A PEPS is $\delta$-injective if every local tensor $T_v$ is $\delta$-injective. \review{Physically, strongly injective PEPS with $\delta$ close to one correspond to those with \emph{gapped} parent Hamiltonians~\cite{schuch2011classifying}.} Our main results establish rigorous thresholds on the injectivity parameter $\delta$ ensuring locality properties. We will state the thresholds in terms of the perturbative parameter $\varepsilon := 1-\delta^2$. 

\subsection{Belief Propagation}

The procedure of belief-propagation involves decomposing the operator space on each edge in terms of a rank-one BP subspace, and the remaining excitation. 
The rank-one subspace is the outer product of fixed-points on that edge. Schematically, the BP approximation is, 
\[
\begin{tikzpicture}[
  baseline={([yshift=-0.65ex] current bounding box.center)},
  tensor/.style={
    draw, line width=0.6pt,
    fill=blue!20!tensorcolor!75!black,
    rounded corners=1.2pt,
    inner sep=0pt,
    outer sep=0pt,
    minimum size=4mm
  },
  tensorA/.style={
    draw, line width=0.7pt,
    fill=blue!30!tensorcolor!85!black,
    rounded corners=1.2pt,
    inner sep=0pt,
    outer sep=0pt,
    minimum size=4mm
  }
]
  \def\dx{1.8}
  \def\dy{1.6}
  \def\a{0.33}

  \newcommand{\SqTN}[3]{%
    \coordinate (#1) at (#2,#3);
    \draw[thick, fill=tensorcolor, rounded corners=2pt]
      ($(#1)+(-\a,-\a)$) rectangle ($(#1)+(\a,\a)$);
    \node at (#1) {\scriptsize };
  }

  \SqTN{A}{0}{0}
  \SqTN{B}{\dx}{0}
  \SqTN{C}{0}{-\dy}
  \SqTN{D}{\dx}{-\dy}

  \coordinate (AE) at ($(A)+(\a,0)$);   \coordinate (BW) at ($(B)+(-\a,0)$);
  \coordinate (CE) at ($(C)+(\a,0)$);   \coordinate (DW) at ($(D)+(-\a,0)$);

  \coordinate (AS) at ($(A)+(0,-\a)$);  \coordinate (CN) at ($(C)+(0,\a)$);
  \coordinate (BS) at ($(B)+(0,-\a)$);  \coordinate (DN) at ($(D)+(0,\a)$);

  \draw[very thick] (AE) -- (BW);
  \draw[very thick] (CE) -- (DW);

  \draw[very thick] (AS) -- (CN);
  \draw[very thick] (BS) -- (DN);
\end{tikzpicture}
\;\approx\;
\begin{tikzpicture}[scale=0.95, baseline={([yshift=-0.65ex] current bounding box.center)}]
  \def\dx{1.8}
  \def\dy{1.6}
  \def\a{0.33}
  \def\r{0.1}
  \def\gap{0.4}

  \newcommand{\SqTN}[3]{%
    \coordinate (#1) at (#2,#3);
    \draw[thick, fill=tensorcolor, rounded corners=2pt]
      ($(#1)+(-\a,-\a)$) rectangle ($(#1)+(\a,\a)$);
    \node at (#1) {\scriptsize };
  }

  \SqTN{A}{0}{0}
  \SqTN{B}{\dx}{0}
  \SqTN{C}{0}{-\dy}
  \SqTN{D}{\dx}{-\dy}

  \coordinate (AE) at ($(A)+(\a,0)$);   \coordinate (BW) at ($(B)+(-\a,0)$);
  \coordinate (CE) at ($(C)+(\a,0)$);   \coordinate (DW) at ($(D)+(-\a,0)$);

  \coordinate (AS) at ($(A)+(0,-\a)$);  \coordinate (CN) at ($(C)+(0,\a)$);
  \coordinate (BS) at ($(B)+(0,-\a)$);  \coordinate (DN) at ($(D)+(0,\a)$);

  \BPEdge{AE}{BW}{\r}{\gap}
  \BPEdge{CE}{DW}{\r}{\gap}

  \BPEdge{AS}{CN}{\r}{\gap}
  \BPEdge{BS}{DN}{\r}{\gap}
\end{tikzpicture}.
\]
The \emph{fixed-points} of a PEPS are defined in terms of its norm network, $\ZZ = \inner{\psi}{\psi}$. For each edge $\vec{e} = (v,n)$ of the graph, the tensor $T_v \star \bar{T}_v$ of the norm network defines a superoperator on the virtual space from any set of legs to its complement \autoref{eq:inline_stacking_bp}(a). For each $n\in\NN(v)$ will be concerned with the superoperator from $\NN(v)/\{n\}$ to $n$. To a graph $G$, we associate the Cartesian-product space $\mathcal{K}_G$ consisting of positive, trace-one matrices assigned to each directed edge. An element $\boldsymbol{\mu} = \{\mu_{(v,w)}\}_{(v,w)} \in \mathcal{K}_G$ is said to satisfy the fixed-point equation for $\psi$ if, for every vertex $v \in V$ and each neighbor $n \in \mathcal{N}(v)$, the following holds:
\begin{equation}
        \Phi_{(v,n)}\left(\bigotimes_{n_i \in \NN(v)\setminus\{n\}}\mu_{(n_i,v)} \right)  \propto \mu_{(v,n)}
\end{equation}
where $\Phi_{(v,n)}$ denotes the superoperator at vertex $v$, directed from $\NN(v)/\{n\}$ to $n$, see \autoref{eq:inline_stacking_bp}(b).
\begin{equation}\label{eq:inline_stacking_bp}
    \begin{tikzpicture}[
scale=0.88,
baseline={([yshift=-0.6ex]current bounding box.center)},
peps/.style={
    draw,
    line width=0.9pt,
    fill=tensorcolor,
    rounded corners=1.5pt,
    minimum size=5mm
},
superop/.style={
    draw,
    line width=0.9pt,
    fill=blue!20!tensorcolor!75!black,
    rounded corners=2pt,
    minimum size=8mm
},
thinleg/.style={line width=1.1pt, draw=black!70},
thickleg/.style={line width=1.35pt, draw=black!70},
phys/.style={line width=1.3pt, draw=black!80},
endcirc/.style={draw=black!70, fill=white, line width=1.1pt}
]



\def\xb{0.18}
\def\yb{-0.30}
\draw[thinleg] (\xb-0.92,\yb) -- (\xb-0.28,\yb);
\draw[thinleg] (\xb,\yb+0.92) -- (\xb,\yb+0.28);
\draw[thinleg] (\xb,\yb-0.92) -- (\xb,\yb-0.28);
\draw[thinleg] (\xb+0.28,\yb) -- (\xb+0.92,\yb);
\node[peps] at (\xb,\yb) {};

\def\xa{-0.18}
\def\ya{0.30}
\draw[phys] (\xa,\ya) -- (\xb,\yb);

\draw[thinleg] (\xa-0.92,\ya) -- (\xa-0.28,\ya);
\draw[thinleg] (\xa,\ya+0.92) -- (\xa,\ya+0.28);
\draw[thinleg] (\xa,\ya-0.92) -- (\xa,\ya-0.28);
\draw[thinleg] (\xa+0.28,\ya) -- (\xa+0.92,\ya);
\node[peps] at (\xa,\ya) {};

\node at (1.38,0) {\Large $=$};

\def\xd{2.62}
\draw[thickleg] (\xd-0.92,0) -- (\xd-0.38,0);
\draw[thickleg] (\xd,0.92) -- (\xd,0.38);
\draw[thickleg] (\xd,-0.92) -- (\xd,-0.38);
\draw[thickleg] (\xd+0.38,0) -- (\xd+0.92,0);
\node[superop] at (\xd,0) {};

\node[font=\bfseries] at (1.30,-1.28) {(a)};


\def\xs{5.05}
\def\L{0.88}
\def\a{0.33}
\def\r{0.10}

\draw[thickleg] (\xs-\L,0) -- (\xs-\a,0);
\draw[thickleg] (\xs,\L) -- (\xs,\a);
\draw[thickleg] (\xs,-\L) -- (\xs,-\a);
\draw[thickleg] (\xs+\a,0) -- (\xs+\L,0);
\node[superop] at (\xs,0) {};

\draw[endcirc] (\xs-\L,0) circle (\r);
\draw[endcirc] (\xs,\L) circle (\r);
\draw[endcirc] (\xs,-\L) circle (\r);

\node at (\xs+1.4,0) {\Large $\propto$};

\def\xout{7.30}
\draw[thickleg] (\xout,0) -- (\xout+0.78,0);
\draw[endcirc] (\xout,0) circle (\r);

\node[font=\bfseries] at (6.28,-1.28) {(b)};

\end{tikzpicture}
\end{equation}

The set of virtual superoperators are used to define the \emph{message-passing} map of the PEPS, given component-wise as, 
\begin{equation}
    [\bm{F}(\bm{\mu})]_{\vec{e}} := \frac{f_{\vec{e}}(\bm{\mu})}{\Tr f_{\vec{e}}(\bm{\mu})}.
\end{equation}
where
\begin{equation}
    [\bm{f}(\bm{\mu})]_{\vec{e}} := f_{\vec{e}}(\bm{\mu}) := \Phi_{(v,n)}\!\left(
\bigotimes_{m \in \NN(v)\setminus\{n\}} \mu_{(m,v)}
\right).
\end{equation}
The fixed-point messages thus satisfy $\bm{F}(\bm{\mu}_\star) = \bm{\mu}_\star$. In general, there is no reason \textit{a priori} that a fixed-point exists. In fact, it is possible that the fixed-point exists but finding it is computationally prohibitive. Practitioners apply the map $\bm{F}(\bm{\mu})$ iteratively as follows, hoping it converges to a fixed point.
\begin{equation}\label{eq:iterative_update}
    \bm{\mu}^{(t+1)} := \bm{F}[\bm{\mu}^{(t)} ]
\end{equation}
This is not guaranteed to converge in general.

\subsection{Cluster Expansion}

We next study the error of BP contractions by systematically constructing corrections to BP and controlling their magnitude. Given a collection of fixed-points, the partition function of the norm-network can be expressed as the following \emph{loop expansion} \cite{midha2025beyond,Evenbly2024}, which operates over set of \emph{excitations} $\LL$ appearing as connected degree-two subgraphs of $G$, 
\begin{equation} \label{eq:loopexpmaintext}
  \ZZ = \ZZ_{\text{BP}}\left(1 + \sum_{\substack{\Gamma\subset\mathcal{L}\\\Gamma\text{ finite, disconnected}}}
  \prod_{\ell\in\Gamma} Z_\ell\right).
\end{equation}
Herein, $\ZZ_{\text{BP}} := \prod_{v\in V} Z^{(v)}$ is the BP approximation to the norm network, where $Z^{(v)} := \Tr\left[\mu_{(n,v)}  \Phi_{(v,n)}\!\left(
\bigotimes_{m \in \NN(v)\setminus\{n\}} \mu_{(m,v)}
\right)\right]$. The \emph{excitations} on top of the BP value are the \emph{loop tensors} $Z_\ell$ is defined as a tensor contraction on the subgraph $\ell$ [see SM]. It is known that the loop expansion does not converge quickly in general. Ref. \cite{midha2025beyond} shows that the loop expansion can be reorganized into a cluster expansion, which has provable convergence properties.
\begin{equation} \label{eq:clusterexpmaintext}
    \log{\ZZ} = \log{\ZZ_{\text{BP}}} + \sum_{\text{connected} \, \mathbf{W}} \phi(\mathbf{W}) Z_{\mathbf{W}},
\end{equation} 
where a \emph{cluster} is collection of loops with multiplicities, $\mathbf{W} = \{(\ell_1, \alpha_1), (\ell_2, \alpha_2), \ldots \}$, where each $\ell_i\in \LL$ is a loop and $\alpha_i \geq 1$ is its multiplicity in the cluster. The corresponding cluster evaluation is, $Z_{\mathbf{W}} = \prod_{i} Z_{\ell_i}^{\alpha_i}$ and the weight of a cluster is $|\mathbf W| := \sum_i \alpha_i |\ell_i|$ where $|\ell|$ for a loop $\ell\in\LL$ denotes the number of edges in $\ell$. The coefficient $\phi$ is the Ursell function, depending solely on the cluster geometry.

\section{Tensor Network BP for injective PEPS}
\subsection{Self-consistent BP fixed points}
Our first result shows that for injective PEPS, {at least one} fixed point exists, and for strongly injective PEPS, the fixed point is unique and the iterative application of $\bm{F}$ finds that fixed point.

\begin{maintextthm}[BP fixed points] 
    \label{thm:fixed-points}
  Given any PEPS with an injectivity parameter $\varepsilon$, the following holds:
    \begin{enumerate}
        \item[(i)] if $\varepsilon < 1$, that is, the PEPS is injective, then a BP fixed point exists.
        \item[(ii)] if $\varepsilon < \varepsilon_{*} := 1/{(2\Delta - 1)}$, then the BP fixed point is unique and can be efficiently found by applying $\bm{F}$ iteratively.
    \end{enumerate}
\end{maintextthm}

The proofs are detailed in the SM, and we present the main arguments here. For an $\varepsilon-$injective PEPS with $\varepsilon=0$, the virtual superoperator of each tensor is the fully depolarizing channel $\Phi_0$. We perturb away from the fully-depolarizing limit, and write $\Phi = \Phi_0 + \Delta\Phi$. One notes that $\|\Delta\Phi(X)\|_1 = \order{\varepsilon}$ for any $X$ with $\|X\|_1 \leq 1$~\footnote{Here, $\|\cdot\|_p$ denotes the Schatten $p-$norm.}. This ensures that $\Tr\Phi(X) \geq (1-\varepsilon) D$ for any unit-trace positive $X$. Hence, we have that $\Tr f_{\vec{e}}(\bm\mu) > 0$ for any $\bm\mu \in \KK_G$, establishing that the normalized message-passing map $\bm{F}$ is component-wise well defined, and maps $\KK_G$ to itself. Since $\KK_G$ is compact and convex, by Brouwer's fixed-point theorem~\cite{deimling2013nonlinear} $\bm{F}$ admits a fixed-point. This proves \autoref{thm:fixed-points}(i). 

To show uniqueness, we show that $\bm{F}$ is a Banach contraction on the complete space $(\KK_G, \|\cdot\|_{\text{max}})$ where $\|\bm{\mu}\|_{\text{max}} := \max_{\vec{e} \in \vec{E}}\|\mu_{\vec{e}}\|_1$. Using the fact that $\Phi_0$ is trivially contracting and that $\|\Delta\Phi(X)\|_1 = \order{\varepsilon}$, we show in the SM that $\|\Phi(X) - \Phi(Y)\|_1 = \|\Delta\Phi(X) - \Delta\Phi(Y)\|_1 \leq q_\Delta(\varepsilon)\|X-Y\|_1$. The condition $\varepsilon < \varepsilon_*$ suffices to ensure $q_\Delta(\varepsilon) < 1$, and therefore the global message passing map is a contraction. This proves \autoref{thm:fixed-points}(ii). \hfill \qed 

The Banach contraction implies an exponential convergence to the fixed point when applying $\bm{F}$ iteratively. Starting from an initial condition $\bm{\mu}^{(0)}$, \autoref{eq:iterative_update} converges to the fixed point up to some error $\epsilon$ in time $\order{{\log(1/\epsilon)}/{\log\left(\varepsilon_*/\varepsilon\right)}}$. The error-floor of BP-based algorithms is crucially dependent on the accuracy of the approximate fixed-points found. For inverse-polynomial accuracy, time $\order{\log N}$ suffices.


\begin{color}{black}

\subsection{Convergence of cluster expansion}
The main result of \cite{midha2025beyond} establishes that the cluster expansion converges exponentially fast in cluster order $m$ if the magnitude of each loop $\ell$ decays with the weight of the loop, 
\begin{equation}
    |Z_\ell| \leq e^{-c|\ell|}
\end{equation}
with a decay-rate above a constant threshold $c > c_0 = \order{\log\Delta}$. The \emph{decay of loops} condition not only gives an error bound on BP by the cluster contributions, but also allows for the systematic corrections of BP by correcting it with low-weight cluster contributions. Our next result establishes the decay of loop condition for PEPS when the injectivity parameter $\varepsilon$ is smaller than a second, smaller threshold $\varepsilon_{**} < \varepsilon_*$, leading to polynomial-time classical simulation.

\begin{maintextthm}[Polynomial-time classical simulation] \label{thm:complexity_maintext}
Given any PEPS with an injectivity parameter $\varepsilon < \varepsilon_{**}$ where  $\varepsilon_{**} = \order{\min\{1/D, (D/\Delta)^{\Delta/2}\}} < \varepsilon_{*}$, then decay of loops is satisfied and cluster expansion converges exponentially fast. This gives a polynomial-time algorithm to compute the following quantities.
\begin{enumerate}
    \item[(i)]  the norm $Z = \inner{\psi}{\psi}$ to $1/\poly(N)$ multiplicative error
    \item[(ii)]  local observables $\expval{O_A} = \expect{\psi}{O_A}{\psi}/\inner{\psi}{\psi}$ to $1/\poly(N)$ multiplicative error, given $\expval{O_A} \neq 0$
    \item[(iii)] correlation functions $\expval{O_AO_B} - \expval{O_A} \expval{O_B}$ to $\tilde{\mathcal{O}}(1/\poly(N))$ additive error
\end{enumerate}
where $A,B \subset V$ are disjoint local regions.
\end{maintextthm}

Here $\tilde{\mathcal{O}}$ suppresses polylogarithmic factors. Efficient classical simulation is following the line of results in Refs.~\cite{harley2025computational,schwarz2017approximating,schuch2011classifying}, and provides an alternate proof of efficient contraction algorithms for strongly injective PEPS. In particular, this result relies only on \emph{graph-locality}, and does not assume embedding in Euclidean space.

The core technical ingredient in the proof of \autoref{thm:complexity_maintext} is the following \autoref{lem:decayofloops}, which establishes loop decay for strongly injective PEPS. Once this decay is in place, the remainder of the theorem follows by direct application of existing results~\cite{midha2025beyond,midha2026}; for completeness, we provide the full argument in the SM.

\begin{maintextlem}\label{lem:decayofloops}
    Strongly injective PEPS with $\varepsilon < \varepsilon_{**}= \order{\min\{1/D, (D/\Delta)^{\Delta/2}\}}$ satisfy decay of loops.
\end{maintextlem}

We can show decay of loops by bounding the ``building-block" of a loop excitation, which appears as the following tensor:
\begin{center}
  $\Big\|$
\begin{tikzpicture}[
  scale=0.6,
  baseline={([yshift=-0.65ex] current bounding box.center)},
  tensor/.style={draw, line width=0.9pt, fill=tensorcolor, rounded corners=2pt, minimum size=8mm},
  redbond/.style={line width=1.6pt, draw=red},
  blackbond/.style={line width=1.6pt, draw=black}
]
  \def\dx{2.5}
  \def\dy{2.5}

  \node[tensor, fill=blue!20!tensorcolor!75!black] (D) at (0,-\dy) {};
  \draw[blackbond] (D.east) -- ++(0.6,0);
  \draw[redbond] (D.west) -- ++(-0.6,0);
  \draw[blackbond] (D.south) -- ++(0,-0.6);
  \draw[blackbond] (D.north) -- ++(0,0.6);
\end{tikzpicture} $\Big\|_2  \leq  \eta(\varepsilon, D, \Delta) = \order{\varepsilon/\sqrt{D^\Delta}}$
\end{center}
where the colored-edge denotes a projection onto the excited subspace. This is shown by arguing that for each edge, $\|\mu_{\vec{e}} - \mathbbm{1}/D\|_1 = \order{\varepsilon}$, which shows that contracting with the excitation cancels out the depolarizing part $\Phi_0$ of $\Phi$, allowing us to use the perturbative bound on $\|\Delta\Phi(X)\|_1 = \order{\varepsilon}$. This condition is also sufficient to show the decay of an vertex with multiple excitations. Furthermore, each loop with $|\ell|$-edges is constructed from at least $2|\ell|/\Delta$ vertices. By `cutting' the loop through a series of Cauchy-Schwarz inequalities, in the SM we show that,
\begin{equation}
    |Z_\ell| = \order{\eta^{2|\ell|/\Delta}}
\end{equation}
Thus, to ensure that $|Z_\ell| \leq e^{-c|\ell|} < e^{-c_0|\ell|}$ with $c_0 = \order{\log\Delta}$ it is sufficient to ensure that $\varepsilon < \varepsilon_{**} = \order{\min\{1/D, (D/\Delta)^{\Delta/2}\}}$. Here, the $1/D$ factor ensures meaningful normalization of the tensors with the BP approximation (see SM for details). \qed 



Decay of loops has previously been used as a sufficient condition to establish efficient and controlled tensor network contraction. \autoref{lem:decayofloops} allows us to directly invoke the main results of~\cite{midha2025beyond}: the cluster expansion converges exponentially fast, and $\log \mathcal{Z}$ can be computed to $\order{1/\poly(N)}$ accuracy by summing clusters of size at most $\order{\log N}$, yielding a $\poly(N)$-time algorithm and establishing part (i) of \autoref{thm:complexity_maintext}. Parts (ii–iii) then follow by applying the results of~\cite{midha2026} to local observables and correlation functions.

\end{color}

\section{Algorithmic Locality}

We now proceed to our main result showing the locality of the TN-BP algorithm. The starting point is showing stability of the message-passing procedure to local perturbations of the tensor network: the influence of a perturbation at a local region on the fixed-point messages decays exponentially fast away from the region. Combined with the convergence of the cluster expansion, which lets us correct for the errors made by BP, we establish locality of the entire algorithm.

\begin{maintextthm}[Algorithmic locality] \label{thm:BPlocality_maintext}
Consider perturbing a strongly injective PEPS (satisfying $\varepsilon < \varepsilon_{**}$) in a local region $A$: $T_A \to T_A'$ such that $\|T_A - T'_A\|_\infty \; =  \order{\varepsilon_* - \varepsilon}$. Let $\bm{\mu}_\star$ and $\bm{\mu}'_\star$ be the fixed points of the original and perturbed PEPS. 
\begin{enumerate}
\item[(1)] \review{\emph{Locality of BP.} For any $\vec{e}=(v,n)$ with $d(v,A) = r$, we have
\begin{equation}
    \|\mu'_{\star,\vec{e}} - \mu_{\star,\vec{e}}\|_1 = \order{e^{-r/\xi_*}}
\end{equation}
Where $1/\xi_* = \order{\log{\varepsilon_*/\varepsilon}}$}
\item[(2)] \review{\emph{Locality of observables.} For any observable $O_B$ (with $\|O_B\|_\infty = 1$) supported in region $B$ that is separated from $A$ by graph distance $R$, we have
\begin{equation}
    \left|\expval{O_B}' - \expval{O_B} \right| = \order{ e^{-R/\xi_{**}}}
\end{equation}
 where $1/\xi_{**} = \order{\min\{1/\xi_*, c-c_0\}}$; $\expval{O_B}$ and $\expval{O_B}'$ are expectation values of $O_B$ on the original and the perturbed PEPS, respectively.}
\end{enumerate}

\end{maintextthm}

The proof can be found in the SM. The main arguments are as follows. We first argue locality of fixed-points, which is then coupled with the decay of loops condition to show locality of cluster corrections, establishing the result. As a result, this also shows that accounting for the changes caused by weak local perturbations requires only local recomputation. Crucially, the proof is constructive, following the steps taken in practical computations within the TN-BP algorithm.

By Weyl's perturbation theory of singular values~\cite{weyl1912asymptotische}, the perturbed tensor network also satisfies strong injectivity with $\varepsilon < \varepsilon_*$, and admits a unique BP fixed point by \autoref{thm:fixed-points}. Now, we start from the fixed-points of the unperturbed networks $\bm{\mu}_{\star}$ and run message-passing on the perturbed state with the map $\bm{F}'$ resulting in $\bm{\mu}^{(t)} = (\bm{F}')^t[\bm{\mu}_{\star}]$. By \autoref{thm:fixed-points}, this iterative process will converge to the perturbed fixed-points $\bm{\mu}^{(t)} \to \bm{\mu}_\star'$ as $t\to\infty$. Combining this convergence with the light-cone inherent in the message passing dynamics, we have, 
\begin{equation}  
\lVert \mu_{\star,\vec{e}} - \mu'_{\star,\vec{e}} \rVert_1
\le
\underbrace{\lVert \mu_{\vec{e}}^{(t)} - \mu'_{\star,\vec{e}} \rVert_1}_{\text{Convergence}}
+
\underbrace{\lVert \mu_{\star,\vec{e}} - \mu_{\vec{e}}^{(t)} \rVert_1}_{\text{Lightcone}}\end{equation}
The convergence error decreases exponentially with time as $\sim e^{-r/\xi_{*}}$, and the lightcone term is zero for $t < r$ and starts increasing thereafter, owing to the strict lightcone in the message-passing dynamics [see SM for a proof]. The optimum bound is thus obtained at $t=r-1$, leading to 
\begin{equation}
    \|\mu_{\star,\vec{e}} - \mu'_{\star,e}\|_1 = \order{e^{-r/\xi_{*}}}
\end{equation}
establishing locality of the BP fixed-points.

We now show the locality of entire the cluster-corrected TN-BP algorithm. We will use the cluster expansion for local observables from Ref.~\cite{midha2026},
\begin{equation} \label{eq:clusterexpobs}
\expval{O_B} = \expval{O_B}_{\text{BP}} \cdot \exp\left(\sum_{\substack{\text{connected} \, \mathbf{W} \\ \text{supp}(\mathbf{W}) \cap B \neq\emptyset }} \phi(\mathbf{W}) Z^{O_B}_\mathbf{W}\right),
\end{equation}
which expresses local expectation values as the BP approximation dressed by intersecting cluster corrections. The exact form of the cluster correction $Z_\mathbf{W}^{O_B}$ can be found in the SM.
The decay of loops condition, assured from \autoref{lem:decayofloops} given $\varepsilon < \varepsilon_{**}$, ensures that the cluster expansion converges. The cluster expansion depends on both the BP expectation value $\expval{O_B}_{\text{BP}}$ and clusters intersecting the region. By locality of fixed-points, we first have that, 
\begin{equation}\label{eq:O_BPlocality}
    \left|\langle O_B \rangle_{\text{BP}}' - \langle O_B \rangle_{\text{BP}}\right| = \order{   e^{-R/\xi_*}}.
\end{equation} 
This follows by noting that the BP approximation depends solely on the messages at the boundary $\partial B$ and the local tensors in $B$. Next, the contribution from the clusters intersecting with $B$ in \autoref{eq:clusterexpobs} is dealt with as follows. We divide the clusters of concern into near and far clusters, $\mathcal{W} = \mathcal{W}_{\text{near}} \cup \mathcal{W}_{\text{far}}$, where clusters in $\mathcal{W}_{\text{near}}$ are distance at most a cutoff $R_{\text{th}}$ away from $B$. Since clusters are also a function of the message tensors involved in the clusters, by carefully accounting for the effect of locality of BP in the cluster contributions, we establish locality of the near clusters, 
\begin{equation}\label{eq:near_cluster_locality}
    \left|\sum_{\mathbf{W} \in \mathcal{W}_{\text{near}}} \phi(\mathbf{W}) (Z^{O_B'}_{\mathbf{W}} - Z^{O_B}_{\mathbf{W}})\right|  = \order{\Delta e^{-R/\xi_*}}
\end{equation}
Finally, the far clusters in $\mathcal{W}_{\text{far}}$ can be accounted for by noting that by convergence of the cluster expansion, contribution of large clusters decays with size, leading to 
\begin{equation}\label{eq:far_cluster_decay}
    \left|\sum_{\mathbf{W} \in \mathcal{W}_{\text{far}}} \phi(\mathbf{W}) (Z^{O_B'}_{\mathbf{W}} - Z^{O_B}_{\mathbf{W}})\right| = \order{|B| \cdot e^{-d\cdot R_{\text{th}}}}
\end{equation}
for $d = c-c_0 = \order{1}$. Combining the bounds from \autoref{eq:O_BPlocality}, \autoref{eq:near_cluster_locality} and \autoref{eq:far_cluster_decay}, we establish locality of the local expectation value, that is,
\begin{equation}
    \Big| \expval{O_B} - \expval{O_B}'\Big| = \order{e^{-R/\xi_{**}}},
\end{equation}
for a correlation length $1/\xi_{**} = \min\{1/\xi_*, c-c_0\}$, provided the cutoff scale satisfies $R_{\text{th}}/R = \Theta(1)$. In principle, this ratio can be optimized to obtain the tightest bound. In practice, within the TN-BP algorithm, it can be tuned to achieve a desired target accuracy. This shows that the effect of the perturbation can be captured through local cluster recomputation, with the required locality scale determined by the desired accuracy. In practice, when evaluating quantities across tensor networks that differ only locally, this leads to significant computational speedups.

\section{Injective PEPS in the wild}\label{sec:examples}
We now proceed to discuss physical examples of states satisfying the (strong) injectivity criteria discussed above. We will do this in three regimes: (i) random PEPS, (ii) Gibbs states of commuting local Hamiltonians, and (iii) noisy quantum circuits. In each case, we will show sufficient thresholds for local-dimension/temperature/noise that guarantees $\varepsilon_0-$injectivity. As a result, one can substitute $\varepsilon_0:= \varepsilon_{**}$, and the previous results of classical simulability~\autoref{thm:complexity_maintext} and locality~\autoref{thm:BPlocality_maintext} apply straightforwardly.

All of the results of this section will be stated with regards to a bounded-degree graph $G=(V,E)$ with max-degree $\Delta$, and a desired strong-injectivity threshold $\varepsilon_0$. For a graph $G$, a \emph{matching} $\pi$ is a subset of edges $\pi = {e_1,e_2,\dots }\subseteq E$ such that the edges are disjoint: $e_i \cap e_j = \emptyset$ for $i\neq j$. 
\subsection{Random PEPS}
Recently, there have been many attempts to reason about the computational complexity of random PEPS---where each entry is sampled independently as a standard complex Gaussian. In particular Ref~\cite{GarciaRandomPEPS} studied the typical complexity of random PEPS with statistical mechanics methods. Furthermore, \cite{Evenbly2024} numerically demonstrated the effectiveness of tensor network BP on random PEPS. Here, we establish a rigorous regime in which random Gaussian PEPS satisfy the strong-injectivity conditions required for efficient contraction.

For each vertex $v\in V$, sample each entry of
\[
G_v : \bigotimes_{e\ni v}\C^D \to \C^d
\]
as a standard complex Gaussian, independently. Then, we show that, for large local Hilbert space dimension $d$, this ensemble of PEPS is strongly injective with high probability. 

\begin{maintextthm}\label{thm:randommain}
  For a given injectivity threshold $0<\varepsilon_0<1$ and failure probability $0<\eta<1$, if the local dimension satisfies
  \[
        d \geq \frac{64}{\varepsilon_0^2} \left(2 D^\Delta \log 9 + \log \frac{2N}{\eta}\right)
  \]
  then the resulting random PEPS is $\varepsilon_0-$injective with probability $\geq 1-\eta$.
\end{maintextthm}

This identifies a concrete, provably tractable regime for random PEPS, complementing earlier analytical and numerical investigations.

\subsection{Gibbs states of commuting local Hamiltonians}
Tensor network representations of Gibbs states of local Hamiltonians have been extensively studied. Particularly, Ref.~\cite{MolnarPEPS} uses a cluster expansion to show efficient PEPS exist for general local Hamiltonians. For commuting Hamiltonians, the Gibbs operator factorizes into local terms, allowing exact tensor-network representations with constant bond dimension~\cite{AlhambraReview}.

Consider two-local Hamiltonians defined on the graph $G$: 
\[
H = \sum_{\{v,w\}\in E}h_{vw}
\]
with an interaction scale $J$ such that $\|h_{vw}\|_{\infty} \leq J$. Furthermore, $\supp(h_{vw}) \subseteq \{v,w\}$, and all terms commute pairwise.

We show that, given an injectivity requirement $\varepsilon_0$, the Gibbs state of such Hamiltonians admit PEPS purifications which are $\varepsilon_0$-injective.
\begin{maintextthm}\label{thm:gibbsmain}
    Given injectivity threshold $\varepsilon_0$, there exists a constant inverse temperature $\beta_c > 0$ depending only on $d, \Delta, J, \varepsilon_0$ such that for $\beta < \beta_c$, a purification $\Psi_\beta$ has a PEPS representation with bond dimension $D=2d^2$ which is $\varepsilon_0-$injective.
\end{maintextthm}

The central ingredient of the proof is the fact that at small $\beta$, the Gibbs factors can be decomposed for any $\{u,v\}\in E$,
\[
e^{-\beta h_{uv}} \propto \sum_{i=1}^D F_{u,i} \otimes F_{v,i}
\]
for $D=2d^2$. Furthermore, at small $\beta$, the $F$ matrices are positive, and satisfy,
\[
\|F - I\|_\infty \leq O(\sqrt{\beta J}).
\]
As a result, $\beta < \beta_c(d, \Delta, J, \varepsilon_0)$ guarantees the PEPS is $\varepsilon_0-$injective.

\subsection{Noisy quantum circuits}
The computational complexity of simulating noisy quantum circuits has attracted considerable attention, with recent work investigating the relation to tensor-networks~\cite{malz2025computational,harley2025computational}.  Here, we identify a regime in which noisy quantum circuits admit exact, strongly injective PEPS purifications. Given the connectivity graph $G$, define two operations: 
\begin{enumerate}
    \item[(1)] Noise channel $\mathcal D_p$, 
    \[
    \mathcal D_p := \bigotimes_{v\in V}\mathcal D^{(v)}_p 
    \] 
    with 
    \[
    \mathcal D^{(v)}_p (\rho):= (1-p)\rho + p \frac{I_2}{2}
    \]
    \item[(2)] Given a matching $\pi$, a unitary layer
    \[
    \mathcal U(\pi) := \bigotimes_{\{v,w\}\in \pi} \mathcal U_{vw}
    \]
    where $\mathcal U_{vw}$ is a two qubit unitary acting on vertices $v,w$ and identity elsewhere.
\end{enumerate}
We consider circuits of the form,
\[
\rho_t := \left(\mathcal D_p \circ \mathcal U(\pi_t)\right) \circ \dots \circ (\mathcal D_p \circ \mathcal U(\pi_1) )(\rho_0)
\]
Then, we show that, given an injectivity threshold $\varepsilon_0$, for any constant depth circuit there is a constant bond-dimension PEPS which is $\varepsilon_0-$injective above a noise threshold $p > p_c$.
\begin{maintextthm}\label{thm:noisymain}
    For every fixed depth $t$, injectivity threshold $\varepsilon_0 \in (0,1)$ there exists a constant threshold $p_c< 1$ depending only on $t,\varepsilon_0$ such that for $p>p_c$ the output state admits an exact PEPS purification $\Psi_\rho$ on a subgraph of $G^2$ with a bond dimension $D \leq 4^t$ which is $\varepsilon_0-$injective.
\end{maintextthm}
The central ingredient is the Kraus decomposition of a noisy quantum channel,
\[
\mathcal N(\rho) = \sum_\alpha K_\alpha \rho K_\alpha^\dagger 
\]
By retaining the Kraus indices as local purifying registers, the noisy circuit can be represented as a pure tensor network: 
\[
\begin{tikzpicture}[
    baseline={(current bounding box.center)},
    x=1cm,y=1cm,
    >=Latex,
    every node/.style={font=\scriptsize},
    wire/.style={line width=0.45pt},
    dwire/.style={line width=0.45pt, double, double distance=1.1pt},
    gate/.style={draw, rounded corners=1pt, minimum width=1.45cm, minimum height=0.42cm, line width=0.45pt},
    noise/.style={draw, rounded corners=1pt, minimum width=0.42cm, minimum height=0.42cm, line width=0.45pt},
    kraus/.style={draw, circle, minimum size=0.36cm, inner sep=0pt, line width=0.45pt},
]

\definecolor{pastelblue}{RGB}{191,219,254}
\definecolor{pastelgreen}{RGB}{187,247,208}
\definecolor{pastelpink}{RGB}{251,207,232}

\begin{scope}

\draw[dwire] (0,0) -- (0,2.25);
\draw[dwire] (1.2,0) -- (1.2,2.25);

\draw[wire,fill=white] (0,0) circle (0.06);
\draw[wire,fill=white] (1.2,0) circle (0.06);

\node[gate, fill=pastelblue] (U1) at (0.6,0.78) {$U$};

\node[noise, fill=pastelgreen] (N1) at (0,1.63) {$\mathcal N$};
\node[noise, fill=pastelgreen] (N2) at (1.2,1.63) {$\mathcal N$};

\node at (0.6,-0.38) {$\rho$};

\end{scope}

\draw[->, line width=0.55pt] (2.0,1.0) -- (2.8,1.0);

\begin{scope}[xshift=3.7cm]

\draw[wire] (0,0) -- (0,0.62);
\draw[wire] (1.2,0) -- (1.2,0.62);

\draw[wire,fill=white] (0,0) circle (0.06);
\draw[wire,fill=white] (1.2,0) circle (0.06);

\node[gate, fill=pastelblue] (U2) at (0.6,0.82) {$U$};

\draw[wire] (0,1.03) -- (0,1.42);
\draw[wire] (1.2,1.03) -- (1.2,1.42);

\node[kraus, fill=pastelpink] (K1) at (0,1.63) {$K$};
\node[kraus, fill=pastelpink] (K2) at (1.2,1.63) {$K$};

\draw[wire] (0,1.81) -- (0,2.25);
\draw[wire] (1.2,1.81) -- (1.2,2.25);
\node at (0.18,2.08) {$o$};
\node at (1.38,2.08) {$o$};

\draw[wire] (K1.west) -- ++(-0.42,0);
\draw[wire] (K1.west) ++(-0.42,0) -- ++(0,0.42);
\draw[wire] (K2.east) -- ++(0.42,0);
\draw[wire] (K2.east) ++(0.42,0) -- ++(0,0.42);

\node at (-0.3,1.42) {$\alpha$};
\node at (1.5,1.42) {$\alpha$};

\node at (0.6,-0.38) {$|\Psi\rangle$};

\end{scope}

\end{tikzpicture}
\]

For the depolarizing channel, the corresponding local Kraus tensor has squared singular-value ratio
\[
r(p)=\frac{p}{4-3p}.
\]
After contracting the temporal direction for constructing the PEPS after $t-$layers, one obtains, 
\[
\varepsilon(T_v) \leq 1-r(p)^{2t} = O(t(1-p))
\]
As a result, sufficiently strong noise guarantees the desired amount of strong injectivity.
\section{Discussion}
This work establishes a fully rigorous, end-to-end treatment of tensor network belief propagation, placing it on firm theoretical footing for quantum many-body applications. \review{These results provide an analytical framework for understanding when information in a loopy tensor network is effectively local, and how that locality can be exploited algorithmically.} Given a quantum state expressed as a PEPS, we showed that injectivity is sufficient to guarantee the existence of BP fixed-points, and \emph{strong} injectivity further guarantees uniqueness.

Building on recent developments in cluster-expansion techniques correcting excitations on top of BP, we further showed that the decay of loops condition is satisfied for strongly injective PEPS. As a result, we obtain polynomial-time classical simulation algorithms for computing properties of PEPS, continuing the line of work in Refs.~\cite{harley2025computational,schwarz2017approximating,schuch2011classifying}. Importantly, our results rely only on \emph{graph locality}, without assuming any embedding in Euclidean space. This makes our results directly relevant to quantum information tasks unconstrained by geometric locality, e.g., simulation of many-body physics on sparse graphs, or decoding quantum LDPC codes.

As our main result, we establish \emph{algorithmic locality} of tensor network belief propagation: local perturbations of tensors induce only local changes in the BP fixed point, allowing recomputation using strictly local data. This locality directly translates into efficient evaluation of local expectation values through the cluster expansion technique. Importantly, the proof is constructive and closely follows the actual execution of the algorithm---rather than solely relying on abstract arguments, it explicitly tracks how information propagates under message passing. In this sense, the analysis not only certifies correctness but also mirrors the way a practitioner would implement BP. Theoretically, this yields a clustering result similar to decay of correlations~\cite{hastings2006spectral}. This enables significant computational speedups when contracting multiple tensor networks which differ only locally, a situation frequently encountered in tensor network simulations in many-body physics.

\review{The thresholds on injectivity should be understood as structural, worst-case sufficient conditions under which tensor networks are provably contractible on any bounded-degree graph. They define a rigorous phase diagram of algorithmic convergence, while the underlying locality of TN-BP is expected to persist substantially beyond this certified regime in typical instances.} We provide concrete examples of physical states in different regimes which can be analyzed using injectivity as a common lens. In particular, we show that (i) random PEPS, (ii) high temperature commuting Gibbs states, and (iii) sufficiently noisy quantum circuits are all strongly injective in their appropriate regimes. As a result, our results of efficient classical simulation and algorithmic locality apply to all of these states.

\section{Acknowledgements}
We thank Grace Sommers for helpful comments on the manuscript and collaboration on related projects.
S.M. acknowledges helpful discussions at the International Centre for Theoretical Sciences (ICTS) Bangalore during the program ``Generalised Symmetries and Anomalies in Quantum Phases of Matter 2026" (code: ICTS/GSYQM2026/01). This work was partially supported by a Brown Investigator Award (S.M. and D.A.). Y.F.Z. and S.G. acknowledge support from NSF QuSEC-TAQS OSI 2326767.

\emph{AI Statement---}In the v2, the strong injectivity results of \autoref{sec:examples} were proved using GPT 6.

\bibliography{refs}

\newcommand{\supplementheading}{\section*{Supplementary Material}}
\providecommand{\supplementheading}{}
\startsupplement

\supplementheading
\listofsmentries

\smsection{Preliminaries}
\smsubsection{Notations}
\textbf{Spaces} We will denote by $\HH$ finite dimensional Hilbert spaces. The space of bounded linear operators on $\HH$ will be denoted as $\BB(\HH)$. The set of positive operators on $\HH$ is denoted $\pos(\HH) \subset \BB(\HH)$. Further, we define the set of positive, trace normalized operators on $\HH$,
\begin{equation}
  \mathcal K(\HH):=\{X \in \BB(\HH): X\succeq 0, \Tr X=1\}.
\end{equation}
Then, $\mathcal K(\HH)$ is compact and convex.

\textbf{Norms} For a vector $x\in\ \C^m$ and any $p\in[1,\infty)$ we define the Schatten norms,
\begin{equation}
    \|x\|_p  := \left(\sum_{i=1}^m |x_i|^p\right)^{1/p}
\end{equation}
For $p=\infty$ we have, 
\begin{equation}
    \|x\|_\infty := \max_{1 \leq i \leq m} |x_i|
\end{equation}
For an operator $A \in \BB(\HH)$, we define the operator Schatten norm in terms of its singular values $\sigma(A)$ as, 
\begin{equation}
    \|A\|_p := \|\sigma(A)\|_p
\end{equation}
for any $p \in [1,\infty]$. We will frequently use the following facts:
\begin{enumerate}
    \item For any $1 \leq p \leq q \leq \infty$ and any operator $A$, we have, 
    \begin{equation}
        \|A\|_p \geq \|A\|_q
    \end{equation}
    In particular, 
    \begin{equation}
        \|A\|_1 \geq \|A\|_2 \geq \|A\|_\infty
    \end{equation}
    \item For any $1 \leq p \leq q \leq \infty$ and any non-zero operator $A$, 
    \begin{equation}
        \|A\|_p \leq \text{rank}(A)^{\frac{1}{p} - \frac{1}{q}} \|A\|_q
    \end{equation}
    In particular, 
    \begin{equation}
        \|A\|_2 \leq \sqrt{\text{rank}(A)} \|A\|_\infty
    \end{equation}
\end{enumerate} 
We also note that, the $\infty-$norm obeys,
\begin{equation}
    \|A\|_\infty = \sup_{x \neq 0}\frac{\|Ax\|_2}{\|x\|_2} = \sigma_{\text{max}}(A)
\end{equation}
we will denote $\|\cdot\|_\infty$ by just $\|\cdot\|$ for convenience. Let $T_{i_1,\dots,i_q}$ be an order-$q$ tensor with Frobenius norm
\begin{equation}
\|T\|_F^2 := \sum_{i_1,\dots,i_q} |T_{i_1,\dots,i_q}|^2 .
\end{equation}
For any bipartition of the tensor legs $A \subseteq \{1,\dots,q\}$ and
$B = A^c$, reshape the tensor into a linear map
\[
T : \mathcal H_A \to \mathcal H_B
\]
by grouping the indices in $A$ and $B$ into multi-indices
$a=(i_{a_1},\dots,i_{a_{|A|}})$ and $b=(i_{b_1},\dots,i_{b_{|B|}})$,
yielding matrix elements $T_{a,b}$. Then the Schatten-$2$ norm of this
operator coincides with the Frobenius norm of the tensor,
\begin{equation}
\|T\|_{2} = \|T\|_F .
\end{equation}

We will also use the H\"{o}lder's inequality for Schatten norms, where for each $(p,q)$ satisfying $1/p + 1/q = 1$ we have,
\begin{equation}
    |\tr(A^\dagger B)| \leq \|A\|_p \|B\|_q
\end{equation}
\textbf{Graphs} For a graph $G=(V,E)$, we define some useful notation. For any $v\in V$, let $\NN(v)$ denote its neigbors in $G$. For any edge $e = \{v,w\} \in E$, we will refer to it's directed versions $\vec{e} = (v,w)$ and $\overleftarrow{e} = (w,v)$. For vertices $v,w \in V$, let $d(v,w)$ denote the graph distance (shortest path length) between $v$ and $w$. For subsets $A,B \subseteq V$, define
\[
d(A,B) := \min_{v\in A, w\in B} d(v,w).
\]
For an edge $e=\{v,w\}$, we define, for any $u \in V$,
$$
d(e,u) := \min\{d(v,u), d(w,u)\}
$$
and similarly, for any $A \subset V$, 
$$
d(e,A) := \min_{u\in A}d(e,u).
$$

\textbf{Common symbols} We will use the following symbols throughout this work:
\begin{enumerate}
    \item $D$. Bond dimension 
    \item $\Delta$. Graph Degree 
    \item $\delta$. Injectivity parameter 
    \item $\varepsilon:=1-\delta^2$. Deviation from maximal injectivity.
\end{enumerate}

\smsubsection{(Strongly) injective tensors}
We work throughout in finite dimension.
Let $\{\HH_i\}_{i=1}^l$ be finite-dimensional Hilbert spaces (the $l$ virtual legs),
and let $\HH_{\mathrm{phys}}$ be the physical Hilbert space.
A single-site PEPS tensor with $l$ virtual legs is a linear map
\begin{equation}
   T : \bigotimes_{i=1}^l \HH_i \longrightarrow \HH_{\mathrm{phys}}.
\end{equation}

\begin{defn}[Injective tensor]
  The tensor $T$ is called \emph{injective} if the induced map
  $T: \bigotimes_{i=1}^l \HH_i \longrightarrow \HH_{\mathrm{phys}}$
  from virtual boundary to physical bulk is injective. 
\end{defn}

We perform a singular value decomposition
\begin{equation}
  T = V \Sigma U^\dagger \in \C^{n_p \times n_v},
\end{equation}
where $n_p \geq n_v$ by injectivity, where $U \in \C^{n_v \times n_v}$ and $V:\in \C^{n_p \times n_v}$  satisfy $U^\dagger U = \mathbbm{1}_{n_v}$ (unitarity) and $V^\dagger V = \mathbbm{1}_{n_v}$ (isometry), and $\Sigma=\mathrm{diag}(\lambda_e) \in \C^{n_v \times n_v}$ collects the singular values.

Normalize tensor such that the singular values satisfy,
\begin{equation}
  \lambda_{\max}:=\max_e \lambda_e = 1,
\qquad
\lambda_{\min}:=\min_e \lambda_e = \delta \in (0,1].
\end{equation}
We call $\delta$ the \emph{injectivity parameter}. Equivalently, this constraints the the \emph{condition number}, defined as, 
\begin{equation}
  \kappa(T) := \frac{\max_e \lambda_e}{\min_e \lambda_e}
\end{equation} 
to be $\kappa(T) = \delta^{-1}$.

\begin{defn}[$\delta-$injective tensor]
   We say that $T$ is $\delta-$injective for $\delta\in (0,1]$ if the condition number of the tensor satisfies $\kappa(T) \leq \delta^{-1}$.
\end{defn}

By definition of injectivity, we have $\delta > 0$. When $\delta=1$ all singular values are unity and $T = VU^\dagger$ itself is an isometry.

\smsubsection{Induced virtual superoperators}

The associated \emph{double-layer} (norm) tensor is obtained by contracting the
physical leg of $T$ with its adjoint. Equivalently, for any bipartition of the virtual legs, 
$L \sqcup L^c$, the norm tensor defines a superoperator,
$$
  \Phi_{L\to L^c} : \mathcal B\!\left(\bigotimes_{i\in L}\HH_i\right)
  \longrightarrow
  \mathcal B\!\left(\bigotimes_{j\in L^c}\HH_j\right),
  \qquad
  \Phi_{L\to L^c}(X)=\sum_a K_{a,L\to L^c}\,X\,K_{a,L\to L^c}^\dagger ,
$$
where $\{K_{a,L\to L^c}\}$ are Kraus operators determined by $T$.

$$
\begin{tikzpicture}[
  scale=0.9,
  baseline={([yshift=-0.5ex] current bounding box.center)},
  peps/.style={
    draw,
    line width=0.8pt,
    fill=tensorcolor,
    rounded corners=1.5pt,
    minimum size=6mm,
    font=\large
  },
  conj/.style={
    draw,
    line width=0.8pt,
    fill=tensorcolor, 
    rounded corners=1.5pt,
    minimum size=6mm,
    font=\large
  },
  superop/.style={
    draw,
    line width=0.9pt,
    fill=blue!20!tensorcolor!75!black, 
    rounded corners=2pt,
    minimum size=9mm,
    font=\Large\bfseries
  },
  thinleg/.style={line width=0.7pt, draw=black!70, ->, >=stealth},
  thickleg/.style={line width=1.8pt, draw=black!70, ->, >=stealth},
  phys/.style={line width=1.5pt, draw=black!80}
]

\def\xb{0.25} \def\yb{-0.35}
\draw[thinleg] (\xb-1.2, \yb) -- (\xb-0.3, \yb);
\draw[thinleg] (\xb, \yb+1.2) -- (\xb, \yb+0.3);
\draw[thinleg] (\xb, \yb-1.2) -- (\xb, \yb-0.3);
\draw[thinleg] (\xb+0.3, \yb) -- (\xb+1.2, \yb);
\node[conj] at (\xb, \yb) {};

\def\xa{-0.25} \def\ya{0.35}
\draw[phys] (\xa, \ya) -- (\xb, \yb);

\draw[thinleg] (\xa-1.2, \ya) -- (\xa-0.3, \ya);
\draw[thinleg] (\xa, \ya+1.2) -- (\xa, \ya+0.3);
\draw[thinleg] (\xa, \ya-1.2) -- (\xa, \ya-0.3);
\draw[thinleg] (\xa+0.3, \ya) -- (\xa+1.2, \ya);
\node[peps] at (\xa, \ya) {};

\node[font=\bfseries] at (0, -1.8) {(a)};

\node at (2.2, 0) {\Large $=$};

\def\xc{4.4}
\draw[thinleg] (\xc-1.2, 0.1) -- (\xc-0.45, 0.1);
\draw[thinleg] (\xc-1.2, -0.1) -- (\xc-0.45, -0.1);
\draw[thinleg] (\xc-0.1, 1.2) -- (\xc-0.1, 0.45);
\draw[thinleg] (\xc+0.1, 1.2) -- (\xc+0.1, 0.45);
\draw[thinleg] (\xc-0.1, -1.2) -- (\xc-0.1, -0.45);
\draw[thinleg] (\xc+0.1, -1.2) -- (\xc+0.1, -0.45);
\draw[thinleg] (\xc+0.45, 0.1) -- (\xc+1.2, 0.1);
\draw[thinleg] (\xc+0.45, -0.1) -- (\xc+1.2, -0.1);

\node[superop] at (\xc, 0) {$\Phi$};
\node[font=\bfseries] at (\xc, -1.8) {(b)};

\node at (6.6, 0) {\Large $\equiv$};

\def\xd{8.8}
\draw[thickleg] (\xd-1.2, 0) -- (\xd-0.45, 0);
\draw[thickleg] (\xd, 1.2) -- (\xd, 0.45);
\draw[thickleg] (\xd, -1.2) -- (\xd, -0.45);
\draw[thickleg] (\xd+0.45, 0) -- (\xd+1.2, 0);

\node[superop] at (\xd, 0) {$\Phi$};
\node[font=\bfseries] at (\xd, -1.8) {(c)};

\end{tikzpicture}
$$

If $T$ is injective, the virtual superoperator can now be cast in the following form, for some disjoint pairing of legs $L \sqcup L^c$, 
\begin{equation}
\label{eq:Phi_def}
\Phi_{L\to L^c}(X)
:=\sum_a \lambda_a^2\,K_{a,L\to L^c}\,X\,K^\dagger_{a,L\to L^c},
\end{equation}
for $X\in\mathcal B(\HH_L)$, where $\{\lambda_a\}$ are the singular values of $T$. We consider $T$ normalized such that $\max_a\lambda_a = 1$. Here $\{K_{a,L\to L^c}\}_a$ are Kraus operators (determined by $U$ and the bipartition) satisfying the \emph{bistochastic} identities
\begin{equation}
\label{eq:Kraus_bistochastic}
\sum_a K^\dagger_{a,L\to L^c}K_{a,L\to L^c} = \mathbbm \dim(\HH_{L^c}) 1_L,
\qquad
\sum_a K_{a,L\to L^c}K^\dagger_{a,L\to L^c} = \mathbbm\dim(\HH_{L}) 1_{L^c},
\end{equation}
which follow from the fact that $U$ is a unitary. To see this, note that,
$$\sum_a U_{i_1\dots i_\ell}^a \bar{U}_{i_1'\dots i_\ell'}^a = \prod_{j} \delta_{i_j i_j'}$$
Now, the Kraus operators are defined as, 
$$K_{a, L\to L^c} := [U^a]_{(i_L)\to (i_{L^c})}$$
which satisfy completeness up to normalization. For instance, 
$$\sum_a K^\dagger_{a,L\to L^c}K_{a,L\to L^c}  = \sum_{a, i_{L^c}}  [\bar{U}^a]_{(i_L'), (i_{L^c})}  [U^a]_{(i_L), (i_{L^c})} = \sum_{i_{L^c}} \delta_{i_L, i'_L} = \dim(\HH_{L^c}) \mathbbm{1}_L$$

Moreover, when $T$ is $\delta-$injective with $\delta = 1$ (i.e., an isometry), the virtual superoperator is just a fully-depolarizing map (up to normalization).

\smsection{Projected entangled pair states}
\smsubsection{Definitions}
Consider a PEPS defined on a graph $G = (V,E)$ with vertex set $V$ and edge set $E$. Each vertex $v\in V$ is associated with a local tensor $T_v$. 

\begin{defn}[PEPS]
A projected entangled pair state (PEPS) on a graph $G = (V,E)$ with uniform virtual dimension $D$ and physical dimension $d_p$ is specified for each $v\in V$ by a tensor 
$$
T_v : \bigotimes_{n\in\NN(v)}\HH_{(n,v)} \to \HH_{\emph{phys}} \cong \C^{d_p},
$$
where $\NN(v)$ denotes the set of neighbors of $v$, and each virtual space $\HH_{(n,v)} \cong \C^D$.
\end{defn}

We will consider graphs of uniform degree for notational simplicity, with $\NN(v) = \Delta$ for all $v \in V$.  

\begin{defn}[Injective PEPS]
We say a PEPS is $\delta$-injective if for every vertex $v\in V$ the corresponding tensor $T_v$ is $\delta-$injective.
\end{defn}

\smsubsection{Belief Propagation}
For each directed edge $\vec{e}=(v,w)$, we define a \emph{bond-space} $\HH_{\vec{e}} \cong \C^D$. A message along edge $\vec{e}=(v,w)$ is a positive operator $\mu_{\vec{e}} \in \pos(\HH_{\vec{e}})$ representing the message from $v$ to $w$. Since messages in opposite directions may differ, we distinguish $\mu_{(v,w)}$ from $\mu_{(w,v)}$.

Let $\vec{E}$ denote the set of all directed edges. The collection of messages on all edges is represented as,

$$
\bm{\mu} = (\mu_{\vec{e}})_{\vec{e}\in\vec{E}} \in \pos(G).
$$

where we define $\pos(G)$ to be the Cartesian-product space of positive matrices on each directed edge. Similarly, we denote with $\KK_G$ the product space of trace-one normalized positive matrices at each edge.

For each vertex $v\in V$ and each neighbor $n\in\NN(v)$, the tensor $T_v$ induces a virtual superoperator
$$
\Phi_{(v,n)}: \bigotimes_{m\in\NN(v)\setminus\{n\}} \pos(\HH_{(m,v)}) \to \pos(\HH_{(v,n)})
$$
that takes the incoming messages from all neighbors of $v$ except $n$ and produces the outgoing message from $v$ to $n$. In Kraus form,
$$
\Phi_{(v,n)}\!\left(
\bigotimes_{m \in \NN(v)\setminus\{n\}} \mu_{(m,v)}
\right)
= \sum_\alpha \lambda_\alpha^{(v,n)}\,
K^{(v,n)}_\alpha\,
\left(\bigotimes_{m \in \NN(v)\setminus\{n\}} \mu_{(m,v)}\right)
(K_\alpha^{(v,n)})^\dagger,
$$
where $\{K^{(v,n)}_\alpha\}_\alpha$ are the Kraus operators corresponding to the bipartition that isolates leg $n$ at vertex $v$, and $\lambda_\alpha^{(v,n)} >0$ are the squared singular values of $T_v$ in this bipartition. This is shown schematically as follows,

$$
\begin{tikzpicture}[
  scale=0.9,
  baseline={([yshift=-0.65ex] current bounding box.center)},
  tensor/.style={
    draw,
    line width=0.9pt,
    fill=blue!20!tensorcolor!75!black, 
    rounded corners=2pt,
    minimum size=9mm,
    font=\Large\bfseries,
    text=black  
  },
  message/.style={
    draw,
    line width=0.9pt,
    fill=white, 
    rounded corners=4pt,
    inner sep=4pt,
    font=\large
  },
  msg/.style={line width=1.5pt, draw=black!70, ->, >=stealth}
]

\def\dx{1.8}
\def\dy{1.6}

\node[tensor] (v) at (0, 0) {$\Phi_{(v,n_4)}$};

\node[message] (m1) at (0, \dy) {$\mu_{(n_1,v)}$};
\node[message] (m2) at (0, -\dy) {$\mu_{(n_2,v)}$};
\node[message] (m3) at (-2.4, 0) {$\mu_{(n_3,v)}$};

\coordinate (out_v) at (\dx-0.2, 0);

\draw[msg] (m1.south) -- (v.north);
\draw[msg] (m2.north) -- (v.south);
\draw[msg] (m3.east) -- (v.west);
\draw[msg] (v.east) -- (out_v);

\node at (\dx + 0.6, 0) {\Large $\propto$};

\node[message] (m4) at (\dx + 2.0, 0) {$\mu_{(v,n_4)}$};
\coordinate (out_final) at (\dx + 3.8, 0);
\draw[msg] (m4.east) -- (out_final);

\end{tikzpicture}
$$
With the virtual superoperator, we define the message-passing maps of the PEPS.

\begin{defn}[Message-passing maps of PEPS]
The \emph{unnormalized message-passing map} $\bm{f}: \pos(G) \to \pos(G)$ is defined component-wise: for each directed edge $\vec{e}=(v,n)\in\vec{E}$,
$$
[\bm{f}(\bm{\mu})]_{\vec{e}} := f_{\vec{e}}(\bm{\mu}) := \Phi_{(v,n)}\!\left(
\bigotimes_{m \in \NN(v)\setminus\{n\}} \mu_{(m,v)}
\right).
$$
The \emph{normalized message-passing map} $\bm{F}: \mathcal K_G \to \mathcal K_G$ is defined by
$$
[\bm{F}(\bm{\mu})]_{\vec{e}} := \frac{f_{\vec{e}}(\bm{\mu})}{\Tr f_{\vec{e}}(\bm{\mu})}.
$$
\end{defn}
The \emph{fixed-points} of a PEPS will be fixed-points of the message passing equations, as defined in the following.
\begin{defn}[Fixed point of PEPS]
A message configuration $\bm{\mu}_\star = (\mu_{\vec{e}})_{\vec{e}\in\vec{E}} \in \pos(G)$ is a \emph{fixed point} of the PEPS if for each directed edge $\vec{e} \in \vec{E}$, there exists $\lambda_{\vec{e}} > 0$ such that
$$
f_{\vec{e}}(\bm{\mu}_\star) = \lambda_{\vec{e}} \mu_{\vec{e}}.
$$
When all messages are normalized to unit trace (i.e., $\bm{\mu}_\star \in \mathcal K_G$), a fixed point satisfies $\bm{F}(\bm{\mu}_\star) = \bm{\mu}_\star$.
\end{defn}

\smsubsection{Loop expansion}
The contraction of the tensor network state, 
\[\ZZ = \inner{\psi}{\psi}\]
can be formally described as a Taylor series in terms of `loops' on the network as described below. Herein, the zeroth order term is the BP approximation determined solely by the tensors and the fixed points. 
\begin{defn}[BP Approximation] \label{def:BPnormalized}
    For each $v\in V$, the local contribution to the partition function is then given as,
\begin{equation}\label{eq:tensor_z}
        Z^{(v)} := \left[\underset{n \in \NN(v)}{\bigotimes}\frac{\mu_{(n,v)}}{\sqrt{I_{vn}}}\right] \star T_v,
\end{equation}
where $I_{vn} := \Tr(\mu_{(v,n)}\mu_{(n,v)})$.
\end{defn}
Now, we begin to define the excitations. The essential ingredient is the projector onto the non-BP subspace. Since the message tensors can be different in each direction, this is generally a non-orthogonal (w.r.t the Hilbert-Schmidt product) projector.

\begin{defn}[Excitation projector] \label{def:excitation_projector}
  For unit-trace operators $X, X' \in \KK(\HH)$ with $\Tr(X)=\Tr(X')=1$, define the excitation projector
\begin{equation}
    \Pi_{X,X'}^\perp(Y):=Y-\frac{\Tr(XY)}{\Tr(XX')} X'.
\end{equation}
\end{defn}
The excitations appear in \emph{loops}, which are subgraphs with the following properties. 
\begin{defn}[Loop subgraph]
A \emph{loop} is a connected subgraph $\ell = (W,F)$ where $W$ is the vertex set and $F$ is the edge set, and every vertex has degree at least $2$. We call $|W|$ the \emph{weight} of the loop.
\end{defn}
We denote the set of connected loops on the graph as $\LL$. Two loops can be either compatible or incompatible, as defined below. 
\begin{defn}[Compatible loops]\label{def:incompatibility}
Two loops \(\ell, \ell' \in \LL \) are said to be \emph{compatible}, written \(\ell \sim \ell'\), if they do not overlap; that is, they share no vertex or edge in the underlying graph. A family \(\Gamma \subset \LL\) of loops is called \emph{compatible} if every pair of distinct loops in \(\Gamma\) is compatible.
\end{defn}

Given a loop $\ell$, the correction value $Z_\ell$ is defined by the following tensor contraction. 

\begin{defn}[Loop activity]
Given a loop $\ell = (W,F)$, the \emph{loop activity} $Z_\ell$ is defined as the contraction of the PEPS tensor network where:
\begin{enumerate}
\item On each edge $e \in F$ (edges inside the loop), insert the antiprojector $\Pi_{\mu_{\star,\vec{e}},\mu_{\star,\overleftarrow{e}}}^\perp$ corresponding to the fixed point on that edge;
\item On each edge $e \notin F$ (edges outside the loop), insert the fixed point projector $\mu_{\star,\vec{e}}\otimes \mu_{\star,\overleftarrow{e}}$.
\end{enumerate}
\end{defn}

With this, we have the following series expansion for the tensor network contraction, from \cite{midha2025beyond, Evenbly2024}. The tensor network is first normalized with respect to the BP approximation value, such that the BP approximation of the normalized object is unity. We will later show that given strong injectivity, this normalization is always meaningful. This leads to the following expansion.

\begin{prop}[Loop expansion] \label{prop:loopseriesformal}
The tensor network contraction admits the expansion
\begin{equation}\label{eq:loop_expansion_reorganized}
  \ZZ = Z_{BP}\left(1 + \sum_{\substack{\Gamma\subset\mathcal{L}\\\Gamma\emph{ finite, compatible}}}
  \prod_{l\in\Gamma} Z_l\right)
\end{equation}
where the sum runs over all finite sets \(\Gamma\) of mutually compatible loops.
\end{prop}

Owing to combinatorial growth of disconnected loops present in this expansion, this does not converge in general. We use techniques from cluster expansions to mitigate this problem, which is now reviewed.

Moreover, the loop tensors are always defined on the normalized tensor network, whose BP expectation value is unity. This ensures that the loop tensors are constructed from strictly local data. This locality extends to clusters, which we define now.

\smsubsection{Cluster expansion techniques}

\begin{defn}[Clusters]
Given the set of allowed excitations $\LL$, a \emph{cluster} is a collection of tuples 
\[
\mathbf{W} = \{(\ell_1, \alpha_1), (\ell_2, \alpha_2), \ldots, (\ell_k, \alpha_k)\}
\]
where each $\ell_i\in \LL$ is a loop and $\alpha_i \in \N$ is the multiplicity of loop $\ell_i$ in the cluster. The total number of loops in the cluster is $n_{\mathbf{W}} = \sum_{i=1}^k \alpha_i$.
\end{defn}

We define the cluster weight $|\mathbf{W}| = \sum_{i} \alpha_i |l_i|$. We also denote $\mathbf{W}! = \prod_i \alpha_i !$. The support of a cluster is $\supp(Z_\mathbf{W}) = \cup_{l\in\mathbf{W}} \supp(l)$. We denote the correction of the cluster $Z_{\mathbf{W}}$ as the product of the loop corrections raised to their respective multiplicities:

\begin{defn}[Cluster correction]
For a cluster $\mathbf{W} = \{(\ell_1, \alpha_1), \ldots, (\ell_k, \alpha_k)\}$, the cluster correction is defined as
\[
Z_{\mathbf{W}} = \prod_{i=1}^k Z_{\ell_i}^{\alpha_i}.
\]
\end{defn}
We call a cluster \(\mathbf{W}\) \emph{connected} if the interaction graph \(G_\mathbf{W}\) is connected, meaning there is a path between any two vertices in the interaction graph. 

\begin{defn}[Interaction graph]
Given a cluster $\mathbf{W} = \{(\ell_1, \alpha_1), \ldots, (\ell_k, \alpha_k)\}$, the \emph{interaction graph} $G_\mathbf{W} = (V_\mathbf{W}, E_\mathbf{W})$ has $|V_\mathbf{W}| = \sum_{i=1}^k \alpha_i$ vertices, with loop $\ell_i$ corresponding to $\alpha_i$ vertices. There is an edge $(\ell,\ell') \in E_\mathbf{W}$ if the loops $\ell$ and $\ell'$ are incompatible ($\ell \not\sim \ell'$, i.e., they share edges), or if they are identical ($\ell=\ell'$).
\end{defn}

We define the following properties of a cluster.

\begin{defn}[Properties of a cluster]
    For a cluster $\mathbf{W}= \{(l_1, \alpha_1), (l_2, \alpha_2), \ldots\}$, we define the following:
    \begin{enumerate}
        \item[(i)] The \emph{cluster weight} (or ``\emph{order}") is defined as 
        \begin{equation}
            |\mathbf{W}| = \sum_{i} \alpha_i |l_i|,
        \end{equation}
        where the number of edges in excitation $l$ be denoted as $|l|$.
        \item[(ii)] The \emph{support} of a cluster is,
        \begin{equation}
        \emph{supp}(Z_\mathbf{W}) = \cup_{l\in\mathbf{W}}\emph{supp}(l)
        \end{equation}
        \item[(iii)] The \emph{correction} of the cluster $Z_{\mathbf{W}}$ as the product of the loop corrections raised to their respective multiplicities:
        \begin{equation}
  Z_{\mathbf{W}} = \prod_iZ_{l_i}^{\alpha_i}.
\end{equation}
\end{enumerate}
\end{defn}

The following two results are from Ref.~\cite{midha2025beyond}, modifying the loop expansion into a cluster expansion in terms of the connected clusters on the tensor network.

\begin{prop}[Connected clusters only]\label{prop:clusterexpformal}
The free energy can be expressed as
\begin{equation}
  \log \ZZ = \log Z_{BP} +  \sum_{\emph{connected} \, \mathbf{W}} \phi(\mathbf{W}) Z_{\mathbf{W}},
\end{equation}
where the sum runs over all connected clusters \(\mathbf{W}\). The coefficient \(\phi(\mathbf{W})\) is given by
\begin{equation}
  \phi(\mathbf{W}) = \frac{1}{\mathbf{W}!} \sum_{\substack{C \in G_\mathbf{W} \\ C \emph{ connected}}} \sum_{(i,j) \in C} (-1)^{|E(C)|}
\end{equation}
\end{prop}
The cluster expansion renders into an efficient algorithm given decay of loops.
\begin{theorem}[Convergence of cluster expansion] \label{thm:clusterconvformal}
    Assume there exists a constant \(c> c_0 :=\log(2e\Delta) + \frac{1}{2} \) such that
\begin{equation}
  |Z_l|\le e^{-c |l|}
\end{equation} 
then, the series for \(\log Z\) converges absolutely. Moreover, the error in truncating the series at order \(m\), denoted $F_m$, 

\begin{equation}
  F_m = \log Z_{BP} +  \sum_{\substack{\emph{connected} \, \mathbf{W} \\ |\mathbf W| \leq m}} \phi(\mathbf{W}) Z_{\mathbf{W}},
\end{equation}
is bounded by
\begin{equation}
  \left| \log Z - {F}_m \right| \le N e^{-d(m+1)}
\end{equation}
where $d = c - c_0$, $\Delta$ is the degree of the graph, and $N$ is the number of vertices.
\end{theorem}
The main technical argument in the proof of Theorem~\ref{thm:clusterconvformal} is the following tail-bound on the decay of large clusters.
\begin{lemma}[Tail bound for large clusters]
\label{lem:cluster_tail}
Consider a PEPS satisfying the loop decay condition $|Z_\ell| \le e^{-c|\ell|}$ for $c > c_0 = \log(2e\Delta) + 1/2$. For any region $A \subset V$ and cutoff $m \in \N$, the contribution from all connected clusters with weight higher than $m$ supported on $A$ satisfies:
\begin{equation}
\sum_{\substack{\emph{connected} \, \mathbf{W} \\ \emph{supp}(\mathbf{W}) \cap A \neq\emptyset \\ |\mathbf{W}| > m}} |\phi(\mathbf{W}) Z_{\mathbf{W}}| \leq \order{|A| e^{-(c-c_0)(m+1)}}.
\end{equation}
\end{lemma}
We will use this in the later proofs that follow.
\subsection{Cluster expansion for observables and correlation functions}
Before introducing these expansions, we define a key notation. As before, $\LL$ denotes the set of connected loops. Now, for any subset of vertices $A\subset V$, we denote $\LL_A$ to be the set of subgraphs which are min-degree two in $V/A$ but can be degree one in $A$. With abuse of notation, we still refer to `loops' in $\LL_A$. That is, $\LL_A$ is the set of strings that loop in $V/A$ but can terminate in $A$. For our purposes, $A$ will always be some local region satisfying $|A|=\order{1}$.

The following are from Ref.~\cite{midha2026}, providing a cluster expansion for local observables and correlation functions, with guarantees of convergence and clustering of correlations.  We first review the multiplicative versions of the local observable expansions. 

\begin{prop}[Cluster expansion for observables]
\label{prop:localobsexpansionformal}
Let $|\psi\rangle$ be a PEPS state on $G=(V,E)$. For a local observable $O_A$ with support on region $A\subset V$ and $\expval{O_A} \neq 0$, the expectation value admits a cluster expansion
\begin{equation} 
\expval{O_A} = \expval{O_A}_{\text{BP}} \cdot \exp\left(\sum_{\substack{\emph{connected} \, \mathbf{W} \\ (\mathbf{W}) \cap A \neq\emptyset }} \phi(\mathbf{W}) Z^{O_A}_\mathbf{W}\right),
\end{equation}
where $\expval{O_A}_{\text{BP}}$ is the BP approximation (computed using fixed point messages $\bm{\mu}_\star$), the sum is over connected clusters whose support overlaps $A$, and $\phi(\mathbf{W})$ are combinatorial coefficients depending on the interaction graph structure. We have $Z^{O_A}_\mathbf{W} := (Z^{A}_{\mathbf{W}} - Z_{\mathbf{W}})$ where $Z_\mathbf{W}$ and $Z^{A}_{\mathbf{W}}$ denote the evaluation of the clusters on tensor networks $\TT = \inner{\psi}{\psi}$ and $\TT^A = \expect{\psi}{O_A}{\psi}$ respectively after suitable BP normalization.
\end{prop}

Given decay of loops (strings in $\LL_A$), truncating this series at a finite cluster order $m$ enjoys the following convergence guarantees.

\begin{theorem}\label{thm:localobsalgorithmformal}
    Given a local region $A \subset V$, assume decay of loops in $\LL_A$ on both $\ZZ=\inner{\psi}{\psi}$ and $\ZZ^A = \expect{\psi}{O_A}{\psi}$ for a local observable with $\expval{O_A}\neq 0$. Then, truncating the cluster expansion for local expectation values  
    \begin{equation}\label{eq:OA-cc} 
    \expval{O_A}_m = \expval{O_A}_{\emph{BP}} \cdot \exp{\sum_{\substack{\emph{conn} \, \mathbf{W} \leftarrow \LL_A \\ 
    \emph{supp}(\mathbf{W}) \cap A \neq\emptyset \\ |\mathbf W| \leq m}} \phi_{\mathbf{W}} (Z^A_{\mathbf{W}} - Z_{\mathbf{W}})}
\end{equation}
    leads to a relative error $\delta_m = |\expval{O_A} - \expval{O_A}_m|/|\expval{O_A}|$ bounded by 
    \begin{equation}
        \delta_m \leq \order{|A| e^{-(c-c_0)(m+1)} }
    \end{equation}
    where $d=c-c_0 = \order{1}$.
\end{theorem}

We also have the following additive version of the expansion for local observables. 

\begin{prop}\label{prop:localobsfreeformal}
Let $\psi$ be a state represented as a PEPS on $G=(V,E)$. For $A\subset V$, let $\ZZ_{\lambda}^A$ denote the TN for $\expect{\psi}{\lambda O_A}{\psi}$ with $\emph{supp}(O_A) \subseteq A$, while $\ZZ^A$, $\ZZ$ denote the TNs for $\expect{\psi}{O_A}{\psi}$ and $\inner{\psi}{\psi}$, respectively. Given a cluster $\mathbf{W}=\{(l_1,\alpha_1),(l_2,\alpha_2),...\}$, denote the cluster correction on $\ZZ_\lambda^A$, $\ZZ^A$, and $\ZZ$, by $\loopZ{\mathbf{W}}, Z_{\mathbf{W}}^A, Z_{\mathbf{W}}$, respectively. Also denote $\mathbf{W}_A$ the subset of loops within $\mathbf{W}$ which intersect $A$. Then
\begin{equation}
    \expval{O}_A = \partial_\lambda \log\ZZ_{\lambda}^A\mid_{\lambda = 0} \, = \expval{O_A}_{BP} + \sum_{\substack{\mathrm{conn} \, \mathbf{W} \leftarrow \LL_A \\ 
    \mathrm{supp}(\mathbf{W}) \cap A \neq\emptyset}} \phi_{\mathbf{W}} Z_{\mathbf{W}} \sum_{l \in \mathbf{W}_A} \alpha_l \left[\frac{Z_l^A}{Z_l} - \expval{O_A}_{BP} \right].
\end{equation}
\end{prop}

Given decay of loops, this also enjoys similar guarantees of convergence when $A$ is a local region.

\begin{theorem}\label{thm:localobsalgorithmformal-derivative}
    Given a local region $A \subset V$, assume decay of loops in $\LL_A$ on both $\ZZ=\inner{\psi}{\psi}$ and $\ZZ^A = \expect{\psi}{O_A}{\psi}$. Then, truncating the cluster expansion (derivative version) for local expectation values (\autoref{prop:localobsfreeformal}) of an observable with with $\|O_A\|=1$ in a region $A\subset V$ leads to an additive error bounded by 
    \begin{equation}
        |\expval{O_A} - \expval{O_A}_m| \leq \order{m |A|e^{-d(m+1)}},
    \end{equation}
    where $d=c-c_0$.
\end{theorem}

The additive version of the expansion can straightforwardly be generalized to one for connected correlation functions, as summarized in the following.

\begin{prop}\label{prop:correlatorexpansionformalderivative}
Let $\psi$ be a state represented as a PEPS on $G=(V,E)$. Given local disjoint regions $A_i \subset V$, and observables $O_{A_i}$ satisfying $\emph{supp}(O_{A_i}) \subseteq A_i$, we have the following cluster expansion for the $p-$point connected correlation functions,

\begin{equation}\label{eq:conn-corr-formal}
   \expval{O_{A_1}\dots O_{A_p}}_c = \sum_{\substack{\emph{conn.}\, \mathbf{W} \leftarrow \LL_{\mathbf{A}} \\
   \emph{supp}(\mathbf{W}) \cap A_i \neq \emptyset}} \phi_{\mathbf{W}} \partial_{\bm \lambda} Z^{\mathbf{A}}_{\mathbf{W},\bm{\lambda}}\Big|_{\bm{\lambda}=0}
\end{equation}
where we denote $\bm{\lambda} := (\lambda_1,\dots,\lambda_p)$, $\partial_{\bm{\lambda}} := \prod_i \partial_{\lambda_i}$, $\mathbf{A} = (A_1,\dots,A_p)$ and the cluster evaluation $Z^{\mathbf{A}}_{\mathbf{W},\bm{\lambda}}$ is done over the network $\ZZ_{\bm{\lambda}}^\mathbf{A} := \expect{\psi}{\exp(\sum_i O_{A_i}\lambda_i)}{\psi}$. The clusters $\mathbf{W}$ are constructed over $\LL_\mathbf{A}$, the set of connected sub-graphs required to be degree $\geq 2$ everywhere except in all regions $A_i$.

\end{prop}

This shows that the cluster expansion for $p-$point connected correlation functions must involve clusters intersecting all the local regions of concern.

Finally, we recall from \cite{midha2026} decay of loops in $\LL_{AB}$ for two local disjoint regions $A$ and $B$ suffices to show clustering of connected correlators $\expval{O_AO_B}_c$ for bounded local operators $O_{A(B)}$.

\begin{theorem}\label{thm:correlatorboundformal}
Decay of loops in $\LL_{AB}$ on the networks $\ZZ^{X}$ for $X\in \{\emptyset,A,B,AB\}$ ensures that bipartite connected correlators $\expval{O_AO_B}_c := \expval{O_AO_B} - \expval{O_A} \expval{O_B}$ for bounded local observables $O_{A(B)}$ with $\emph{supp}(O_{A(B)})\subseteq A(B)$ and $\|O_{A(B)}\|=1$ cluster, 
\[
\Bigl|\expval{O_AO_B}_c\Bigr| \leq \order{e^{-d(A,B) / \xi} }
\]
for a finite correlation length $\xi \le \order{(1/(c-c_0))}$ and $d(A,B)$ being the graph distance.
\end{theorem}

We combine these results to show that decay of loops renders the cluster expansion into an efficient algorithm for computing properties of PEPS. Before we proceed, it is good to note the following: 

\begin{enumerate}
    \item[(i)] Decay of loops is a bulk property. That is, for a loop $\ell$ if $Z_\ell$ (evaluated on $\ZZ$) decays, then so does $Z^A_\ell$ (evaluated on $\ZZ^A$) for some local region $A$, since changing the contracting by some local amount can not alter the asymptotic decay.
    \item[(ii)] However, one does not straightforwardly imply the other (for instance, the loop activity of $\ell \in \LL_A / \LL$ vanishes on $\ZZ$ but could be non-zero on $\ZZ^A$). Nonetheless, the proof technique we use to show decay of loops [in \autoref{prop:loop_decay}] proceeds via showing a decay of the building block of the loop excitations [in \autoref{prop:excitation_decay} and \autoref{lem:excitationprojectorbound}] takes care of this subtle detail.
\end{enumerate}

\begin{theorem}\label{thm:efficiencyviadecay}
For a PEPS satisfying the conditions of \autoref{thm:localobsalgorithmformal} and \autoref{thm:correlatorboundformal}, that is, satisfying decay of loops, we have that,
\begin{enumerate}
    \item[(i)] the norm $Z = \inner{\psi}{\psi}$ can be computed to $1/\poly(N)$ multiplicative error in $\poly(N)$ time
    \item[(ii)] local observables $\expval{O_A} = \expect{\psi}{O_A}{\psi}/\inner{\psi}{\psi}$ with $\expval{O_A} \neq 0$ can be computed to multiplicative error $\epsilon$ in $\poly(1/\epsilon)$ time
    \item[(iii)] $2$-point correlation functions  can be computed to additive error $\tilde{\mathcal{O}}(1/\poly(N))$ in $\poly(N)$ time
\end{enumerate}
\end{theorem}
\begin{proof}
    By \autoref{thm:clusterconvformal}, we have that the $m-$order approximation to the free-energy satisfies
    \begin{equation}
  \left| \log \ZZ - \tilde{F}_m \right| \le N e^{-d(m+1)}
\end{equation}
Hence, to ensure $\left| \log Z - \tilde{F}_m \right| < \order{1/\poly(N)}$ we require $m = \Omega(\log{N})$. The complexity of computing a cluster of order $m$ is $e^{\order{m}}$, which is $\order{\poly(N)}$ for $m = \Omega(\log{N})$. Finally, note that additive error guarantee for $\log\ZZ$ translates to a multiplicative error guarantee for $\ZZ$. This proves (i). 

Now, to bound the error in computing local expectation values, we use \autoref{thm:localobsalgorithmformal}. We have for a bounded $O_B$ with $\expval{O_B} \neq 0$,
\begin{equation}
    \Big| \frac{\expval{O_B} - \expval{O_B}_m}{\expval{O_B}}\Big| \leq \order{e^{-d(m+1)}} \leq \epsilon
\end{equation}
which can be ensured given $m= \order{\log{1/\epsilon}}$, again computable in $\poly{1/\epsilon}$ time. This shows (ii).

Lastly, we compute correlation functions. Let $A$ and $B$ be the disjoint local regions of concern, with local operators $O_{A(B)}$ with $\|O_{A(B)}\| =1$. If $d(A,B) = \Omega(\log{N})$ then we have that $|\expval{O_AO_B}_c| \leq 1/\poly(N)$ by \autoref{thm:correlatorboundformal} and we trivially output $0$. For $d(A,B) = \order{\log{N}}$, we first note by \autoref{prop:correlatorexpansionformalderivative} that only clusters intersecting all the local regions appear in the expansion. 
Second, we consider truncating the additive version of the correlator expansion of \autoref{prop:correlatorexpansionformalderivative} at cluster weight $m = d(A,B) + \order{\log{1/\epsilon}}$, leading to an error $\order{m^2 (|A|+|B|) e^{-dm}}$, which is (note that $(|A|+|B|) = \order{1}$,
\begin{equation}
    \order{d(A,B)^2 \left(\log{\frac{1}{\epsilon}}\right)^2 \cdot e^{-d[d(A,B) + \log{1/\epsilon}]}} = \order{\log{N}\cdot \frac{1}{\poly(N)}} = \tilde{\mathcal{O}}(1/\poly(N)) 
\end{equation}

since we already have $d(A,B) = \order{\log N}$ and we choose $\epsilon = 1/\poly(N)$. This can be computed at a cost of $e^{d(A,B)} \poly{1/\epsilon}$, which is $\poly(N,1/\epsilon) = \poly(N)$.
\end{proof}

\newpage 

\smsection{BP Fixed points of PEPS}

\smsubsection{Existence of fixed points}

We first show that, the virtual operator of an injective tensor does not map density matrices to zero.



\begin{prop}
\label{prop:injective_nonvanishing}
Assume the tensor $T$ is $\delta-$injective with $\delta > 0$. For any leg $i\in[\Delta]$, the Kraus
form for the map $\Phi_{[\Delta]_{\neq i}\to i}$ can be chosen as
\begin{equation}
  \Phi_{L_{\neq i}\to i}(Z)=\sum_a \lambda_a^2\,K_a Z K_a^\dagger,
\qquad
\sum_a K_a^\dagger K_a = D \mathbbm 1_{L_{\neq i}}.
\end{equation}
Then for all $X\in\mathcal  \KK(\otimes_{j\neq i}\HH_j)$,
\begin{equation}
\label{eq:trace_lower_injective}
\Tr f(X) \ \ge\ D \delta^2 \ >\ 0.
\end{equation}
In particular $f(X)\neq 0$ for all 
$X\in\mathcal K(\HH)$, and the normalized map
\begin{equation}
\label{eq:F_def_injective}
F:\mathcal \KK(\otimes_{j\neq i}\HH_j)\to\mathcal K(\HH_i),
\qquad
F(X):=\frac{f(X)}{\Tr f(X)},
\end{equation}
is well-defined and continuous.
\end{prop}

\begin{proof}
Using the Kraus form and cyclicity of the trace,
\begin{equation}
  \Tr f(X)
=
\Tr\!\left[\sum_a \lambda_a^2 K_a XK_a^\dagger\right]
=
\Tr\!\left[\Big(\sum_a \lambda_a^2 K_a^\dagger K_a\Big)\,X\right].
\end{equation}
Since $\lambda_a^2\ge \delta^2 > 0$ and $\sum_a K_a^\dagger K_a= D\mathbbm 1$,
$$
\sum_a \lambda_a^2 K_a^\dagger K_a \ \succeq\ \delta^2 \sum_a K_a^\dagger K_a
= \delta^2\,D\,\mathbbm 1,
$$
hence
\begin{equation}
  \Tr f(X)\ge \delta^2 D \Tr(X)=\delta^2D.
\end{equation}
This proves \eqref{eq:trace_lower_injective}. Continuity of $F$ is immediate from
continuity of $f$ and the uniform lower bound on the denominator.
\end{proof}

Using \autoref{prop:injective_nonvanishing}, we can now show that the normalized message-passing map is well-defined and continuous.

\begin{prop}
\label{prop:nonTI_nonvanishing}
If the PEPS is $\delta$-injective with $\delta > 0$, then for all $\bm{\mu} \in \mathcal K_G$ and all $\vec{e}\in\vec{E}$,
$$
\Tr f_{\vec{e}}(\bm{\mu}) \ge D\,\delta^2 > 0.
$$
In particular, $\bm{F}$ is well-defined and continuous on $\mathcal K_G$.
\end{prop}

\begin{proof}
Fix a directed edge $\vec{e}=(v,n)\in\vec{E}$. Using the Kraus representation and the trace lower bound from injectivity (as in \autoref{prop:injective_nonvanishing}),
$$
\Tr f_{\vec{e}}(\bm{\mu}) \ge \delta^2 \sum_\alpha \Tr\!\left[(K_\alpha^{(v,n)})^\dagger K_\alpha^{(v,n)} \bigotimes_{m\in\NN(v)\setminus\{n\}} \mu_{(m,v)}\right]
= \delta^2 D \prod_{m\in\NN(v)\setminus\{n\}} \Tr \mu_{(m,v)}.
$$
Since $\mu_{(m,v)} \in \mathcal K(\HH_{(m,v)})$, each $\Tr \mu_{(m,v)} = 1$, yielding $\Tr f_{\vec{e}}(\bm{\mu}) \ge D\,\delta^2$. The trace-normalized global map $\bm{F}: \mathcal K_G \to \mathcal K_G$ is defined by
$$
[\bm{F}(\bm{\mu})]_{\vec{e}} := \frac{f_{\vec{e}}(\bm{\mu})}{\Tr f_{\vec{e}}(\bm{\mu})}.
$$
Since for each edge $\Tr f_{\vec{e}}(\bm{\mu}) \ge D\,\delta^2$, this is well-defined and continuous.

\end{proof}

Now, we recall Brouwer's fixed-point theorem.
\begin{theorem}[Brouwer \cite{deimling2013nonlinear}]
  A continuous function on a nonempty finite-dimensional compact, convex set has a fixed point.
\end{theorem}

Since $\KK_G$ is a compact and convex finite-dimensional set, and $\bm{F}$ is a continuous function mapping $\KK_G$ to itself, we thus get existence of fixed points. 

\begin{theorem}[Existence for injective PEPS]
\label{thm:nonTI_existence}
If the PEPS is $\delta$-injective with $\delta > 0$, then there exists a fixed point $\bm{\mu}_\star \in \mathcal K_G$ such that $\bm{F}(\bm{\mu}_\star) = \bm{\mu}_\star$.
\end{theorem}

\begin{proof}
By \autoref{prop:nonTI_nonvanishing}, $\bm{F}$ is a continuous map from the compact convex set $\mathcal K_G$ to itself. Brouwer's fixed-point theorem guarantees the existence of a fixed point.
\end{proof}


\smsubsection{Uniqueness of fixed points}

Let $\Phi$ be the virtual superoperator of a $\delta-$injective tensor. We will aim to perturb away from the $\delta=1$ limit by expressing $\Phi = \Phi_0 + \Delta\Phi$. Where, $\Phi_0$ denotes the virtual superoperator in the maximally-injective case ($\delta=1$, all singular values $\lambda_e \equiv 1$). Then
\begin{equation}
\label{eq:Phi0_depolarizing}
\Phi_0\!\left(\bigotimes_{j \neq l} \mu_j\right) = \prod_{j \neq l} \Tr(\mu_j)\,\mathbbm 1 = \mathbbm 1,
\qquad\text{for all }\mu_j \in \KK(\HH).
\end{equation}
For general $\delta \in (0,1]$, decompose $\Phi = \Phi_0 + \Delta\Phi$ where
$$
\Delta\Phi := \Phi - \Phi_0 = \sum_e (\lambda_a^2 - 1) K_a \,(\cdot)\, K_a^\dagger
$$
captures the deviation from isometricity. We will use this form to show that the message-passing has a unique fixed point when this deviation is controlled.

\begin{lemma}\label{lem:depolperturbation}
  For the non-depolarizing perturbation $\Delta\Phi$ in the virtual superoperator of a $\delta-$injective tensor for the leg bipartition $[\Delta] = L\cup L^c$, we have that 
 \begin{equation}
  \|\Delta\Phi(A)\|_1 \leq \dim(\HH_{L^c} )(1-\delta^2)\|A\|_1
 \end{equation}
 for any $A \in \KK(\HH_L)$.
\end{lemma}
\begin{proof}
  We have, 
  \begin{equation}
    \Delta\Phi := \Phi - \Phi_0 = \sum_a (\lambda_a^2 - 1) K_a \,(\cdot)\, K_a^\dagger
  \end{equation}
  Let $E:=\sum_a (\lambda_a^2-1)\,K_a^\dagger K_a$ and recall that $\sum_a K^\dagger_a K_a = \dim(\HH_{L^c} )\mathbbm{1}_{L}$. Since $\lambda_a^2\in[\delta^2,1]$ we get,
\begin{equation}
\label{eq:E_bound_uniqueness}
\|E\|_\infty \le  \dim(\HH_{L^c} )(1-\delta^2).
\end{equation}
For Hermitian $A \in \KK(\HH_L)$, by H\"{o}lder's inequality,
$$
\|\Delta\Phi(A)\|_1 \le \|E\|_\infty\,\|A\|_1
\le \dim(\HH_{L^c} )(1-\delta^2)\,\|A\|_1.
$$

\end{proof}

For each edge $\vec{e}=(v,n)$, we decompose $\Phi_{(v,n)} = \Phi_0^{(v,n)} + \Delta\Phi^{(v,n)}$, where $\Phi_0^{(v,n)}$ corresponds to the isometric case ($\delta=1$) and $\Delta\Phi^{(v,n)}$ captures the deviation from isometricity. We finally combine these estimates to establish uniqueness of the fixed-point given $\varepsilon < \varepsilon_* = 1/(2\Delta -1)$ (recall the definition of $\varepsilon := 1-\delta^2$).

\begin{prop}[Banach Contraction]
\label{prop:banach_contr}
Let $\Delta \ge 2$ and suppose the PEPS is $\delta$-injective with $\delta\in(0,1]$.
If
$$
\delta^2>\frac{2(\Delta-1)}{2(\Delta-1)+1}\Leftrightarrow \quad \boxed{\varepsilon < \varepsilon_* := \frac{1}{2\Delta - 1}},
$$
then the normalized map $\bm{F}: \KK_G \to \KK_G$ is a Banach contraction on $(\KK_G, \|\cdot\|_{\max})$ with contraction constant
$$
q_\Delta(\delta) := \frac{2(\Delta-1)(1-\delta^2)}{\delta^2} < 1.
$$
Consequently, $\bm{F}$ has a unique fixed point $\bm{\mu}_\star \in \KK_G$.
\end{prop}

\begin{proof}
Fix two message configurations $\bm{\mu}, \bm{\nu} \in \KK_G$. We wish to estimate
$$
\|\bm{F}(\bm{\mu}) - \bm{F}(\bm{\nu})\|_{\max}
= \max_{\vec{e}\in\vec{E}} \left\|\frac{f_{\vec{e}}(\bm{\mu})}{\Tr f_{\vec{e}}(\bm{\mu})} - \frac{f_{\vec{e}}(\bm{\nu})}{\Tr f_{\vec{e}}(\bm{\nu})}\right\|_1.
$$
For each fixed edge $\vec{e}=(v,n)$, consider the contraction on $\Phi_{(v,n)}$, we aim to establish,
$$
\left\|\frac{f_{\vec{e}}(\bm{\mu})}{\Tr f_{\vec{e}}(\bm{\mu})} - \frac{f_{\vec{e}}(\bm{\nu})}{\Tr f_{\vec{e}}(\bm{\nu})}\right\|_1
\le q_{\Delta}(\delta) \max_{m\in\NN(v)\setminus\{n\}} \|\mu_{(m,v)} - \nu_{(m,v)}\|_1,
$$
where
$$
q_\Delta(\delta) := \frac{2(\Delta-1)\,(1-\delta^2)}{\delta^2}.
$$
Since the incoming messages for edge $e=(v,n)$ are $\{\mu_{(m,v)}: m\in\NN(v)\setminus\{n\}\}$, and each of these is a component of $\bm{\mu}$, we have
$$
\max_{m\in\NN(v)\setminus\{n\}} \|\mu_{(m,v)} - \nu_{(m,v)}\|_1 \le \|\bm{\mu} - \bm{\nu}\|_{\max}.
$$

The difference of the unnormalized component maps is
$$
f_{\vec{e}}(\bm{\mu}) - f_{\vec{e}}(\bm{\nu})
= \Phi_{(v,n)}\!\left(\bigotimes_{m\in\NN(v)\setminus\{n\}} \mu_{(m,v)} - \bigotimes_{m\in\NN(v)\setminus\{n\}} \nu_{(m,v)}\right).
$$
Since $\Phi_0$ is the depolarizing channel and all messages have unit trace, we have
$$
\Phi_0\!\left(\bigotimes_{m\in\NN(v)\setminus\{n\}} \mu_{(m,v)} - \bigotimes_{m\in\NN(v)\setminus\{n\}} \nu_{(m,v)}\right) = \mathbbm{1} - \mathbbm{1} = 0.
$$
Therefore,
$$
f_{\vec{e}}(\bm{\mu}) - f_{\vec{e}}(\bm{\nu})
= \Delta\Phi\!\left(\bigotimes_{m\in\NN(v)\setminus\{n\}} \mu_{(m,v)} - \bigotimes_{m\in\NN(v)\setminus\{n\}} \nu_{(m,v)}\right).
$$

Using the telescoping sum,
$$
\bigotimes_{m\in \NN(v)\setminus\{n\}} \mu_{(m,v)}
-
\bigotimes_{m\in \NN(v)\setminus\{n\}} \nu_{(m,v)}
=
\sum_{k=1}^{\Delta-1}
\left(
  \bigotimes_{i<k} \nu_{(m_i,v)}
\right)
\otimes
\bigl(\mu_{(m_k,v)}-\nu_{(m_k,v)}\bigr)
\otimes
\left(
  \bigotimes_{i>k} \mu_{(m_i,v)}
\right).
$$

and applying \autoref{lem:depolperturbation} with an application of triangle inequality, we obtain
\begin{align}
\|f_{\vec{e}}(\bm{\mu}) - f_{\vec{e}}(\bm{\nu})\|_1
&\le D(1-\delta^2)
   \sum_{m\in \NN(v)\setminus\{n\}}
   \|\mu_{(m,v)} - \nu_{(m,v)}\|_1 \\
&\le (\Delta-1)D(1-\delta^2)
   \|\bm{\mu} - \bm{\nu}\|_{\max}.
\end{align}
Now, for positive operators $A, B$ with traces $a := \Tr A$ and $b := \Tr B$, we have the normalization inequality
$$
\Big\|\frac{A}{a}-\frac{B}{b}\Big\|_1
\le \frac{2}{\min(a,b)}\,\|A-B\|_1.
$$
(Proof: By triangle inequality and reverse triangle inequality, $\big\|\frac{A}{a}-\frac{B}{b}\big\|_1 = \frac{1}{ab}\|bA - aB\|_1 \le \frac{1}{a}\|A-B\|_1 + \frac{|a-b|}{ab}\|B\|_1 \le \frac{2}{a}\|A-B\|_1$.)

By \autoref{prop:injective_nonvanishing}, $\Tr f_{\vec{e}}(\bm{\mu}) \ge D\delta^2$ for all $\bm{\mu} \in \KK_G$. Therefore,
$$
\left\|[\bm{F}(\bm{\mu})]_{\vec{e}} - [\bm{F}(\bm{\nu})]_{\vec{e}}\right\|_1
= \left\|\frac{f_{\vec{e}}(\bm{\mu})}{\Tr f_{\vec{e}}(\bm{\mu})} - \frac{f_{\vec{e}}(\bm{\nu})}{\Tr f_{\vec{e}}(\bm{\nu})}\right\|_1
\le \frac{2}{D\delta^2} \|f_{\vec{e}}(\bm{\mu}) - f_{\vec{e}}(\bm{\nu})\|_1
\le \frac{2(\Delta-1)(1-\delta^2)}{\delta^2} \|\bm{\mu} - \bm{\nu}\|_{\max}.
$$

Taking the maximum over all legs $l \in [\Delta]$,
$$
\|\bm{F}(\bm{\mu}) - \bm{F}(\bm{\nu})\|_{\max}
\le q_\Delta(\delta) \|\bm{\mu} - \bm{\nu}\|_{\max},
\qquad
q_\Delta(\delta) := \frac{2(\Delta-1)(1-\delta^2)}{\delta^2}.
$$

When $\delta^2 > \frac{2(\Delta-1)}{2(\Delta-1)+1}$, we have $q_\Delta(\delta) < 1$. 
\end{proof}

\begin{theorem}[Uniqueness and convergence to the fixed point]
\label{thm:nonTI_uniqueness}
Under the conditions of \autoref{prop:banach_contr}, an injective PEPS with $\varepsilon < \varepsilon_*$ has a unique fixed point $\bm \mu_\star \in \mathcal K_G$. Moreover, to achieve accuracy $\|\bm \mu^{(t)} - \bm \mu_\star\|_{\max} \le \epsilon$, it suffices to take
\begin{equation}
t = \mathcal O\!\left(\frac{\log(1/\epsilon)}{\log(\varepsilon_*/\varepsilon)}\right).
\end{equation}
\end{theorem}

\begin{proof}
Consider the map $\bm F : \KK_G \to \KK_G$ acting on the complete metric space $(\KK_G, \|\cdot\|_{\max})$. For $\varepsilon < \varepsilon_*$, the map $\bm F$ is a contraction with contraction factor $q_\Delta(\varepsilon) < 1$, i.e.,
\begin{equation}
\|\bm F(\bm \mu) - \bm F(\bm \nu)\|_{\max}
\le q_\Delta(\varepsilon)\, \|\bm \mu - \bm \nu\|_{\max}
\quad \forall \bm \mu, \bm \nu \in \KK_G,
\end{equation}
by \autoref{prop:banach_contr}. Since $(\KK_G, \|\cdot\|_{\max})$ is complete, the Banach fixed-point theorem implies the existence and uniqueness of a fixed point $\bm \mu_\star \in \KK_G$.

Applying the contraction bound with $\bm \nu = \bm \mu_\star$ and iterating, we obtain
\begin{equation}
\|\bm \mu^{(t)} - \bm \mu_\star\|_{\max}
\le q_\Delta(\varepsilon)^t \, \|\bm \mu^{(0)} - \bm \mu_\star\|_{\max}.
\end{equation}
This shows exponential convergence to the fixed point. Since
\[
q_\Delta(\varepsilon)=\mathcal O\!\left(\frac{\varepsilon}{\varepsilon_*}\right),
\]
the error after $t$ iterations behaves, up to constants, as
\begin{equation}
\|\bm \mu^{(t)} - \bm \mu_\star\|_{\max}
\lesssim
\left(\frac{\varepsilon}{\varepsilon_*}\right)^t.
\end{equation}
Therefore, to achieve fixed-point error at most $\epsilon$, it suffices to choose $t$ such that
\[
\left(\frac{\varepsilon}{\varepsilon_*}\right)^t \lesssim \epsilon.
\]
Taking logarithms yields
\begin{equation}
t = \mathcal O\!\left(\frac{\log(1/\epsilon)}{\log(\varepsilon_*/\varepsilon)}\right).
\end{equation}
\end{proof}


\newpage 
\smsection{Loop Corrections}

\smsubsection{Decay of single excitation}

We now establish bounds on the decay of excitations from the fixed point found in \autoref{thm:nonTI_uniqueness}. The key idea is to show that any fixed point must lie close to the normalized identity, and consequently the orthogonal projection away from the fixed point is well-approximated by the corresponding projection away from the identity. This ensures that the `depolarizing' part of the virtual superoperator does not contribute significantly to the loop tensors. This is indeed a loose bound, and in practice one expects that loops will decay faster than the bounds proven hereon.

Starting from this section, we will assume that $\Delta$ and $D$ are all $\order{1}$ constants. From there, we will not keep track of the coefficients of $\order{\varepsilon^2}$ terms, which may be complicated functions of $(D,\Delta)$ but the constant is always $\order{1}$. Later when we prove the convergence of any Taylor expansions, we can always choose a sufficiently small $\varepsilon$ to account for the constant.

We first establish that the fixed-points of an injective tensor are $\order{\varepsilon}$ close to the maximally-mixed state in trace-distance.

\begin{prop}[Fixed point proximity to identity]
\label{prop:fp_near_identity}
Let $\bm{\mu}_\star$ be the fixed point from \autoref{thm:nonTI_uniqueness}. Then for each edge $\vec{e} \in \vec{E}$,
$$
\|\mu_{\star,\vec{e}} - I_1\|_1 \le \frac{2(1-\delta^2)}{\delta^2} = 2\varepsilon + \order{\varepsilon^2}.
$$
where $I_1 = \mathbbm{1}/D$.
\end{prop}

\begin{proof}
Fix any vertex $v\in V$. Fix an outgoing directed edge from $v$, call it $\vec{e}$. Since $\mu_{\star,\vec{e}}$ is a component of the fixed point, we have
$$
\mu_{\star,\vec{e}} = \frac{f_{\vec{e}}(\bm{\mu}_\star)}{\Tr f_{\vec{e}}(\bm{\mu}_\star)}.
$$

Write
$$
f_{\vec{e}}(\bm{\mu}_\star)=\mathbbm 1 + R,
\qquad
R:=f_{\vec{e}}(\bm{\mu}_\star)-\mathbbm 1.
$$
We know that (i) $\|R\|_1 \le D(1-\delta^2)$ [by \autoref{lem:depolperturbation}.] and (ii) $\Tr f_{\vec{e}}(\bm{\mu}_\star) \ge \delta^2 D$ [by \autoref{prop:injective_nonvanishing}]

Now write
$$
\mu_{\star,\vec{e}} - I_1
=
\frac{f_{\vec{e}}(\bm{\mu}_\star)}{\Tr f_{\vec{e}}(\bm{\mu}_\star)} - \frac{\mathbbm 1}{D}
=
\frac{\mathbbm 1 + R}{\Tr f_{\vec{e}}(\bm{\mu}_\star)} - \frac{\mathbbm 1}{D}.
$$
Applying the triangle inequality,
\begin{align*}
\|\mu_{\star,\vec{e}} - I_1\|_1
&\le
\Big\|\frac{R}{\Tr f_{\vec{e}}(\bm{\mu}_\star)}\Big\|_1
+
\Big\|\mathbbm 1\Big(\frac{1}{\Tr f_{\vec{e}}(\bm{\mu}_\star)}-\frac{1}{D}\Big)\Big\|_1 \\
&=
\frac{\|R\|_1}{\Tr f_{\vec{e}}(\bm{\mu}_\star)}
+
\|\mathbbm 1\|_1\,
\frac{\big|\Tr f_{\vec{e}}(\bm{\mu}_\star)-D\big|}{\Tr f_{\vec{e}}(\bm{\mu}_\star)\cdot D}.
\end{align*}
Using $\big|\Tr f_{\vec{e}}(\bm{\mu}_\star)-D\big|\le \big|\Tr(R)\big| \leq \|R\|_1$ and $\|\mathbbm 1\|_1=D$, we obtain
$$
\|\mu_{\star,\vec{e}} - I_1\|_1
\le
\frac{2\|R\|_1}{\Tr f_{\vec{e}}(\bm{\mu}_\star)}
\le
\frac{2\,D(1-\delta^2)}{\delta^2 D}
=
\frac{2(1-\delta^2)}{\delta^2}.
$$
where we used $\|R\|_1 \leq D(1-\delta^2)$ again. Thus, 
$$
\|\mu_{\star,\vec{e}} - I_1\|_1 \leq 2\varepsilon + \order{\varepsilon^2}
$$
\end{proof}

Next up, we show that given strong injectivity, the messages are full-rank, which is a technical step we need for ensuring that the BP normalization is well-conditioned.

\begin{lemma}[Full rank messages]\label{lem:fullrank}
Let \(\mu \succeq 0\) be an operator on a \(D\)-dimensional Hilbert space with $\Tr(\mu)=1$
and let $I_1 := \frac{I}{D}$
denote the trace-normalized identity. If
$$\|\mu - I_1\|_1 < \frac{2}{D},
$$
then $\mu$ is full rank.
\end{lemma}

\begin{proof}
Let \(\lambda_1,\dots,\lambda_D\) be the eigenvalues of \(\mu\). Since \(\mu \succeq 0\) and \(\Tr(\mu)=1\), we have
\[
\lambda_i \ge 0,
\qquad
\sum_{i=1}^D \lambda_i = 1.
\]
Since \(I_1 = I/D\), its eigenvalues are all equal to \(1/D\). Since $[\mu, I_1] = 0$, we have,
\[
\|\mu-I_1\|_1
=
\sum_{i=1}^D \left|\lambda_i-\frac1D\right|.
\]

We prove the contrapositive. Suppose that \(\mu\) is not full rank. Then at least one eigenvalue vanishes, say \(\lambda_j=0\). Define
\[
x_i := \lambda_i-\frac1D.
\]
Then
\[
\sum_{i=1}^D x_i
=
\sum_{i=1}^D \lambda_i - 1
=
0.
\]
Hence the total positive part of the \(x_i\)'s equals the total absolute value of the negative part:
\[
\sum_{i:x_i>0} x_i
=
\sum_{i:x_i<0} (-x_i).
\]
Therefore,
\[
\sum_{i=1}^D |x_i|
=
2\sum_{i:x_i<0} (-x_i).
\]
Since \(x_j = -1/D\), we obtain
\[
\sum_{i:x_i<0} (-x_i) \ge \frac1D,
\]
and thus
\[
\|\mu-I_1\|_1
=
\sum_{i=1}^D \left|\lambda_i-\frac1D\right|
=
\sum_{i=1}^D |x_i|
\ge
\frac{2}{D}.
\]
This proves the contrapositive: if \(\mu\) is not full rank, then
\[
\|\mu-I_1\|_1 \ge \frac{2}{D}.
\]
Equivalently, if
\[
\|\mu-I_1\|_1 < \frac{2}{D},
\]
then \(\mu\) is full rank.
\end{proof}

Now, we argue that the BP message normalization $I_{vw}$ is perturbed away from the $\varepsilon=0$ case by $O(\varepsilon^2)$, given strong-injectivity. As a reminder, we defined the excitation projector given matrices $X, X'\in\KK(\HH)$ as 
$$
\Pi_{X,X'}^\perp(Y):=Y-\frac{\Tr(XY)}{\Tr(XX')} X'.
$$

\begin{lemma}
 \label{lem:Ivwconditioned}
    For any edge $e = \{v,w\} \in E$, we have the normalization of the excitation projector (as in \autoref{def:excitation_projector}) $I_{vw} := \Tr(\mu_{(v,w),\star}\mu_{(w,v),\star})$ obeys,
    \begin{equation}
      \big|I_{vw} - \frac{1}{D}\big| \leq 4\frac{(1-\delta^2)^2}{\delta^4} = \order{\varepsilon^2}
    \end{equation}
    In addition, $I_{vw}$ remains positive at sufficiently small $\varepsilon$.
    \begin{equation}
      \varepsilon < \frac{1}{1+2\sqrt{D}} \implies I_{vw} > 0
    \end{equation}
\end{lemma}
\begin{proof}
  Since $\Tr(\mu_{\star,\vec{e}})=\Tr(\mu_{\star,\overleftarrow{e}})=\Tr(I_1)=1$, we have 
\begin{align}
  \Tr(\mu_{\star,\vec{e}}\mu_{\star,\overleftarrow{e}})  
  &= \Tr\!\Big([\mu_{\star,\vec{e}}-I_1 + I_1]\,[\mu_{\star,\overleftarrow{e}}-I_1 + I_1]\Big) \\ 
  &=  \Tr(I_1^2) + \Tr\!\Big(I_1[\mu_{\star,\vec{e}}-I_1]\Big) + \Tr\!\Big(I_1[\mu_{\star,\overleftarrow{e}} - I_1]\Big) 
      + \Tr\!\Big([\mu_{\star,\vec{e}} - I_1][\mu_{\star,\overleftarrow{e}} - I_1]\Big) \\ 
  &= \Tr(I_1^2)+ \Tr\!\Big([\mu_{\star,\vec{e}} - I_1][\mu_{\star,\overleftarrow{e}} - I_1]\Big) \end{align}
Now, we have $\Tr(I_1^2)= 1/D$ and,
$$
\Big| \Tr\!\Big([\mu_{\star,\vec{e}} - I_1][\mu_{\star,\overleftarrow{e}} - I_1]\Big) \Big| \leq \|\mu_{\star,\vec{e}} - I_1\|_1 \|\mu_{\star,\overleftarrow{e}} - I_1\|_\infty\leq \|\mu_{\star,\vec{e}} - I_1\|_1\|\mu_{\star,\overleftarrow{e}} - I_1\|_1
$$
where, we use Hölder's inequality $|\Tr(AB)|\le \|A\|_1\|B\|_\infty$ and $\|B\|_\infty \le \|B\|_1$. Further using Proposition~\ref{prop:fp_near_identity}, we get, 

\begin{equation}
  \Big| \Tr(\mu_{\star,\vec{e}}\mu_{\star,\overleftarrow{e}})  - \frac{1}{D} \Big| \leq 4\frac{(1-\delta^2)^2}{\delta^4}
\end{equation}
 For small enough $1-\delta^2$, we are guaranteed that $\Tr(\mu_{\star,\vec{e}}\mu_{\star,\overleftarrow{e}}) > 0$, and the normalization in the excitation projector is meaningful. Substituting $\varepsilon = 1 - \delta^2$, the right-hand side becomes $4\frac{\varepsilon^2}{(1-\varepsilon)^2}$. To ensure that the trace is strictly positive, we require the lower bound of this absolute value inequality to be greater than zero:

\begin{equation}
    \frac{1}{D} - 4\frac{\varepsilon^2}{(1-\varepsilon)^2} > 0
\end{equation}

Rearranging the terms and taking the square root of both sides, which is valid since $1-\varepsilon > 0$, we obtain:

\begin{equation}
    \frac{1}{\sqrt{D}} > \frac{2\varepsilon}{1-\varepsilon} \implies 1 - \varepsilon > 2\varepsilon\sqrt{D}
\end{equation}

Isolating $\varepsilon$, we find the strict positivity condition:

\begin{equation}
    \varepsilon < \frac{1}{1 + 2\sqrt{D}}
\end{equation}

Thus, as long as $\varepsilon$ satisfies this bound, we are guaranteed that $\Tr(\mu_{\star,\vec{e}}\mu_{\star,\overleftarrow{e}}) > 0$, and the normalization in the excitation projector is physically meaningful.

\end{proof}
Now, we show that the BP fixed point approximation is well conditioned (away from zero). This is important, because the cluster expansion requires normalizing the tensor network with the fixed point value. We obtain that for sufficiently small $\varepsilon < \order{1/D}$ ensures well conditioned normalization.

\begin{prop}[Bound on BP approximation]\label{prop:bpapproximationbound}
    The local BP approximation at vertex $v$ with degree $\Delta$,
    \begin{equation}
        Z^{(v)} := \left[\bigotimes_{n \in \NN(v)}\frac{\mu_{(n,v)}}{\sqrt{I_{nv}}}\right] \star T_v,
    \end{equation}
    for a $\delta$-injective PEPS obeys, 
    \begin{equation}
        \left| Z^{(v)} - D^{\Delta/2} \right| \leq D^{(\Delta/2)+1}\varepsilon + \order{\varepsilon^2}
    \end{equation}
    where $\varepsilon := 1-\delta^2$. In particular, we have that, 
    \begin{equation}
     \varepsilon < \min\{\frac{1}{D}, \frac{1}{1+2\sqrt{D}}\} \implies Z^{(v)} > 0
    \end{equation} 
\end{prop}

\begin{proof}
    We start by splitting $Z^{(v)}$ into its numerator $N^{(v)}$ and denominator $D^{(v)}$ components. By definition, the numerator $N^{(v)}$ evaluates to:

    \[
    N^{(v)} :=
        \left[\bigotimes_{n \in \NN(v)}{\mu_{(n,v)}}\right] \star T_v = 1 + \Tr\left[\mu_{(n_*,v)} \Delta\Phi_{(v,n_*)}\left(\bigotimes_{m\in\NN(v)\setminus\{n_*\}} \mu_{(m,v)}\right)\right]
    \]
    We can bound the trace term using Hölder's inequality for Schatten norms ($\|AB\|_1 \leq \|A\|_\infty \|B\|_1$) and the fact that the infinity norm is upper bounded by the 1-norm ($\|\mu_{(n_*,v)}\|_\infty \leq \|\mu_{(n_*,v)}\|_1 = 1$):
    \[
        \Big|\Tr\left[\mu_{(n_*,v)} \Delta\Phi_{(v,n_*)}\left(\bigotimes_{m\in\NN(v)\setminus\{n_*\}} \mu_{(m,v)}\right)\right]\Big| \leq 1 \cdot \left\|\Delta\Phi_{(v,n_*)}\left(\bigotimes_{m\in\NN(v)\setminus\{n_*\}} \mu_{(m,v)}\right)\right\|_1
    \]
    By \autoref{lem:depolperturbation}, substituting $\dim(\HH_{L^c}) = D$ and $1-\delta^2 = \varepsilon$, we know that for any operator $A$:
    \[
        \|\Delta\Phi(A)\|_1 \leq D\varepsilon \|A\|_1
    \]
    Applying this to our state, and using the fact that the trace norm of a product of valid normalized messages is $1$ (i.e., $\|\bigotimes \mu_{(m,v)}\|_1 = 1$), the deviation of the numerator from $1$, $\delta N = N^{(v)}-1$ is bounded by:
    \[
        |\delta N| \leq D\varepsilon
    \]
    For the denominator $D^{(v)}$, we are given the bound on the local normalizations $I_{nv}$ by \autoref{lem:Ivwconditioned}. Substituting $\delta^4 = (1-\varepsilon)^2$:
    \[
        \Big| I_{nv}  - \frac{1}{D} \Big| \leq 4\frac{\varepsilon^2}{(1-\varepsilon)^2} = \order{\varepsilon^2}
    \]
    This implies we can write $I_{nv} = \frac{1}{D} + \delta I_{nv}$ where $|\delta I_{nv}| = \order{\varepsilon^2}$. Expanding the product of the denominators over the $\Delta$ neighbors yields:
    \[
    [D^{v}]^{-1} = 
        \prod_{n \in \NN(v)} I_{nv}^{-1/2} = D^{\Delta/2} \prod_{n \in \NN(v)} (1 + D \delta I_{nv})^{-1/2} = D^{\Delta/2}(1 + \order{\varepsilon^2})
    \]
    Assembling the numerator and the denominator, we get:
    \[
        Z^{(v)} = D^{\Delta/2}(1 + \delta N)(1 + \order{\varepsilon^2}) = D^{\Delta/2}(1 + \delta N + \order{\varepsilon^2})
    \]
    Subtracting the unperturbed expectation $D^{\Delta/2}$ and applying our bound for $|\delta N|$ yields the desired result:
    \[
        \left| Z^{(v)} - D^{\Delta/2} \right| \leq D^{(\Delta/2)+1}\varepsilon + \order{\varepsilon^2}
    \]
    Now $\varepsilon < 1/D$, guarantees that the numerator is positive (this can also be seen by the fact that messages must be full rank by \autoref{lem:fullrank}). From \autoref{lem:Ivwconditioned}, we have that $\varepsilon < 1/(1+2\sqrt{D})$ guarantess that the denominator is positive. Thus, we have, 
    \begin{equation}
      \varepsilon < \min\{\frac{1}{D}, \frac{1}{1+2\sqrt{D}}\} \implies Z^{(v)} > 0
    \end{equation}
\end{proof}

One also notes that, positivity of the numerator of $Z^{(v)}$,
$$ N^{(v)} := \left[\bigotimes_{n \in \NN(v)}{\mu_{(n,v)}}\right] \star T_v =\Tr\left[\mu_{(n_*,v)} \Phi_{(v,n_*)}\left(\bigotimes_{m\in\NN(v)\setminus\{n_*\}} \mu_{(m,v)}\right) = \Tr\left[\mu_{(n_*,v)}\sum_a \lambda_a^2 K_a \mu_{m\neq n_*}K_a^\dagger\right]\right] $$ where we denoted $\mu_{m\neq n_*} := \left(\bigotimes_{m\in\NN(v)\setminus\{n_*\}} \mu_{(m,v)}\right)$ can also be obtained by a full-rank condition on the messages. To see this, we first recall that $\varepsilon < 1/D$ ensures that the message tensors are full rank. Second, for each term in the sum, note, $\mu_{m\neq n_*} \succ 0 \implies  K_a \mu_{m\neq n_*}K_a^\dagger \succ 0$ for $K_a \neq 0$, and thus
$$
\Tr[\mu_{(n_*,v)} K_a \mu_{m\neq n_*}K_a^\dagger] > 0
$$
since $\Tr(AB) > 0$ for two positive operators $A,B\succ 0$.

Now, we will use the fact that the fixed-point messages are $\order{\varepsilon}$ close to the maximally mixed state to argue closeness of the projector to the orthogonal-projector onto the space of trace-less matrices. 

\begin{lemma}[Projection difference bound]
\label{lem:projection_difference}
Given a fixed point $\bm{\mu}_\star$, denote its component on both directions of an edge $e$ as $\mu_{\star,\vec{e}}, \mu_{\star,\overleftarrow{e}}$. Then, for any $Y$ with $\|Y\|_1\le 1$,
$$
\|\Pi_{\mu_{\star,\vec{e}},\mu_{\star,\overleftarrow{e}}}^\perp(Y)-\Pi_{I_1,I_1}^\perp(Y)\|_1 \le g(\varepsilon, D, \Delta):= 2(D+1)\varepsilon + \order{\varepsilon^2}
$$
as long as 
$$ 
\varepsilon<\frac{1}{1+2\sqrt{D}}
$$

\end{lemma}

\begin{proof}
We will use the following:   
\begin{enumerate}
  \item $\Big| I_e  - \frac{1}{D} \Big| \leq 4\frac{(1-\delta^2)^2}{\delta^4}$ [by \autoref{lem:Ivwconditioned}] 
  \item $\|\mu_{\star,\vec{e}} - I_1\|_1 \le \frac{2(1-\delta^2)}{\delta^2}$ [by \autoref{prop:fp_near_identity}]
\end{enumerate}

 We can write
\begin{align*}
\Pi_{\mu_{\star,\vec{e}},\mu_{\star,\overleftarrow{e}}}^\perp(Y)-\Pi_{I_1,I_1}^\perp(Y)
&=\Big(Y-\frac{\Tr(\mu_{\star,\vec{e}}Y)}{I_e} \mu_{\star,\overleftarrow{e}}\Big)
 -\Big(Y-\frac{\Tr(I_1 Y)}{\Tr(I_1^2)} I_1\Big) \\
&=\Tr(Y) I_1-\frac{\Tr(\mu_{\star,\vec{e}}Y)}{I_e} \mu_{\star,\overleftarrow{e}},
\end{align*}

adding and subtracting $\mu_{\star,\overleftarrow{e}}\Tr(Y)$ gives, 

\begin{align*}
  &\Tr(YI_1) -\frac{\Tr(\mu_{\star,\vec{e}}Y)}{I_e} \mu_{\star,\overleftarrow{e}} \\ 
  &\quad= \Tr(Y)[I_1 - \mu_{\star,\overleftarrow{e}}] + \mu_{\star,\overleftarrow{e}}\left(\Tr(Y) - \frac{\Tr(\mu_{\star,\vec{e}}Y))}{I_e}\right) \\
 &\quad= \Tr(Y)[I_1 - \mu_{\star,\overleftarrow{e}}] + \mu_{\star,\overleftarrow{e}} \left(\frac{\Tr(Y[I_e \mathbbm{1} - \mu_{\star,\vec{e}}])}{I_e}\right)  \\ 
 &\quad= \Tr(Y)[I_1 - \mu_{\star,\overleftarrow{e}}] + \mu_{\star,\overleftarrow{e}} \left(\frac{\Tr(Y[I_e \mathbbm{1} - I_1 + I_1 - \mu_{\star,\vec{e}}])}{I_e}\right) 
\end{align*}
We perform $|\Tr(Y[I_e \mathbbm{1} - I_1 + I_1 - \mu_{\star,\vec{e}}])| \leq |\Tr(Y[I_e \mathbbm{1} - I_1])| + |\Tr(Y[I_1 - \mu_{\star,\vec{e}}])| \leq \|Y\|_1 \|(I_e - 1/D) \mathbbm{1}\|_\infty + \|Y\|_1 \|[I_1 - \mu_{\star,\vec{e}}]\|_\infty \leq |(I_e - 1/D)| +\|[I_1 - \mu_{\star,\vec{e}}]\|_1 $ through a sequence of H\"{o}lder and Schatten norm inequalities. Thus, 
\begin{align*}
  \Big\| \Tr(YI_1) -\frac{\Tr(\mu_{\star,\vec{e}}Y)}{I_e} \mu_{\star,\overleftarrow{e}} \Big\|_1 \leq \|I_1 - \mu_{\star,\overleftarrow{e}}\|_1 + \frac{1}{1/D - 4\frac{(1-\delta^2)^2}{\delta^4}} \left(4\frac{(1-\delta^2)^2}{\delta^4} + \|I_1 - \mu_{\star,\vec{e}}\|_1\right)
\end{align*}

Plugging in the bounds from our previous propositions, we obtain:
\begin{align*}
  \Big\| \Pi_{\mu_{\star,\vec{e}},\mu_{\star,\overleftarrow{e}}}^\perp(Y)-\Pi_{I_1,I_1}^\perp(Y) \Big\|_1 &\leq \frac{2(1-\delta^2)}{\delta^2} + \frac{4\frac{(1-\delta^2)^2}{\delta^4} + \frac{2(1-\delta^2)}{\delta^2}}{1/D - 4\frac{(1-\delta^2)^2}{\delta^4}}.
\end{align*}

Substituting the perturbation parameter $\varepsilon = 1-\delta^2$, which implies $\frac{1-\delta^2}{\delta^2} = \frac{\varepsilon}{1-\varepsilon}$, yields:
\begin{align*}
  \Big\| \Pi_{\mu_{\star,\vec{e}},\mu_{\star,\overleftarrow{e}}}^\perp(Y)-\Pi_{I_1,I_1}^\perp(Y) \Big\|_1 &\leq \frac{2\varepsilon}{1-\varepsilon} + \frac{4\frac{\varepsilon^2}{(1-\varepsilon)^2} + \frac{2\varepsilon}{1-\varepsilon}}{1/D - 4\frac{\varepsilon^2}{(1-\varepsilon)^2}} \\
  &= \frac{2\varepsilon}{1-\varepsilon} + \frac{4D\varepsilon^2 + 2D\varepsilon(1-\varepsilon)}{(1-\varepsilon)^2 - 4D\varepsilon^2} \\
  &= \frac{2\varepsilon}{1-\varepsilon} + \frac{2D\varepsilon + 2D\varepsilon^2}{1 - 2\varepsilon + (1-4D)\varepsilon^2}.
\end{align*}

Finally, expanding this expression to first order in $\varepsilon$, we use $\frac{1}{1-\varepsilon} = 1 + \varepsilon + \order{\varepsilon^2}$ and $\frac{1}{1 - 2\varepsilon + \order{\varepsilon^2}} = 1 + 2\varepsilon + \order{\varepsilon^2}$:
\begin{align*}
  \Big\| \Pi_{\mu_{\star,\vec{e}},\mu_{\star,\overleftarrow{e}}}^\perp(Y)-\Pi_{I_1,I_1}^\perp(Y) \Big\|_1 &\leq 2\varepsilon(1 + \varepsilon + \order{\varepsilon^2}) + (2D\varepsilon + \order{\varepsilon^2})(1 + 2\varepsilon + \order{\varepsilon^2}) \\
  &= 2\varepsilon + 2D\varepsilon + \order{\varepsilon^2} \\
  &= 2(D+1)\varepsilon + \order{\varepsilon^2}.
\end{align*}
The condition on 
$$ 
\varepsilon<\frac{1}{1+2\sqrt{D}}
$$
ensures that the normalization is positive.
\end{proof}

Uptil now, to require meaningful normalization with BP fixed point and closeness to the trace-less projector, we require, 
\begin{equation}
    \varepsilon< \min\{ \frac{1}{D}, \frac{1}{1+2\sqrt{D}}\}
\end{equation}
With this, we will establish that loop excitations decay with the size of the loop. To do that, we will first bound the `building block' of a loop excitation in the following. 

\begin{prop}[Excitation decay bound]
\label{prop:excitation_decay} 
For any $v\in V$, the normalized building block with one edge excited satisfies, 

\begin{center}
  $\Big\|$
\begin{tikzpicture}[
  scale=0.6,
  baseline={([yshift=-0.65ex] current bounding box.center)},
  tensor/.style={draw, line width=0.9pt, fill=tensorcolor, rounded corners=2pt, minimum size=8mm},
  redbond/.style={line width=1.6pt, draw=red},
  blackbond/.style={line width=1.6pt, draw=black}
]
  \def\dx{2.5}
  \def\dy{2.5}

  \node[tensor, fill=tensorcolor!50] (D) at (0,-\dy) {};
  \draw[blackbond] (D.east) -- ++(0.6,0);
  \draw[redbond] (D.west) -- ++(-0.6,0);
  \draw[blackbond] (D.south) -- ++(0,-0.6);
  \draw[blackbond] (D.north) -- ++(0,0.6);
\end{tikzpicture} $\Big\|_2 \leq \eta(\varepsilon, D, \Delta) := 2D^{2-\Delta/2} (D+2)\varepsilon + \order{\varepsilon^2}.$
\end{center}

where the red-line on the edge $\vec{v}=(v,n)$ denotes the excitation projector $\Pi_{\mu_{\star,\vec{e}},\mu_{\star,\overleftarrow{e}}}^\perp$ on that edge, and the vertex is normalized by the square root of BP expectation value [c.f.~\autoref{def:BPnormalized}].
\end{prop}



\begin{proof}
Fix a pair of conjugate edges $\vec{e}, \overleftarrow{e} \in \vec{E}$. Now, $\Pi_{\mu_{\star,\vec{e}},\mu_{\star,\overleftarrow{e}}}^\perp \circ f_{\vec{e}}$ is a superoperator from $\BB(\HH)^{\otimes{(\Delta-1)}} \to \BB(\HH)$. We will first aim to bound the numerator $\|\Pi_{\mu_{\star,\vec{e}},\mu_{\star,\overleftarrow{e}}}^\perp \circ f_{\vec{e}}(\bm Y)\|_1$ for arbitrary bounded $\bm{Y}$, and then add back in the bound on the BP normalization from \autoref{prop:bpapproximationbound}.

{For any matrix $A \in \C^{m\times n}$, we have $\|A\|_2 \leq \sqrt{\text{rank}(A)} \max_{v:\|v\|_2=1}\|Av\|_2$ and $\|A\|_2 \leq \|A\|_1$. Now, the superoperator from $\BB(\HH)^{\otimes{(\Delta-1)}} \to \BB(\HH)$ viewed has a matrix has dimensions $D^2$ and $D^{2(\Delta-1)}$. Thus, the rank of this matrix is $\leq D^2$. This gives us, 
}

$$\|\Pi_{\mu_{\star,\vec{e}},\mu_{\star,\overleftarrow{e}}}^\perp \circ f_{\vec{e}}\|_2 \leq D \max_{\bm{Y}: \|\bm{Y}\|_2 = 1}\|\Pi_{\mu_{\star,\vec{e}},\mu_{\star,\overleftarrow{e}}}^\perp \circ f_{\vec{e}}(\bm Y)\|_2 \leq D \max_{\bm{Y}: \|\bm{Y}\|_1 = 1}\|\Pi_{\mu_{\star,\vec{e}},\mu_{\star,\overleftarrow{e}}}^\perp \circ f_{\vec{e}}(\bm Y)\|_2$$ 
Since $\Pi_{I_1,I_1}^\perp\circ \Phi_0 = 0$ (as $\Phi_0(\cdot)$ is proportional to $\mathbbm 1$), we have
\begin{align*}
\|\Pi_{\mu_{\star,\vec{e}},\mu_{\star,\overleftarrow{e}}}^\perp \circ f_{\vec{e}}(\bm{Y})\|_2 
&=  \|(\Pi_{\mu_{\star,\vec{e}},\mu_{\star,\overleftarrow{e}}}^\perp - \Pi_{I_1,I_1}^\perp + \Pi_{I_1,I_1}^\perp) \circ f_{\vec{e}}(\bm{Y})\|_2 \\ 
&\le \|(\Pi_{\mu_{\star,\vec{e}},\mu_{\star,\overleftarrow{e}}}^\perp - \Pi_{I_1,I_1}^\perp) \circ f_{\vec{e}}(\bm{Y})\|_2 + \| \Pi_{I_1,I_1}^\perp \circ f_{\vec{e}}(\bm{Y})\|_2  \\ 
&= \|(\Pi_{\mu_{\star,\vec{e}},\mu_{\star,\overleftarrow{e}}}^\perp - \Pi_{I_1,I_1}^\perp) \circ f_{\vec{e}}(\bm{Y})\|_2 + \| \Pi_{I_1,I_1}^\perp \circ \Delta\Phi(\bm{Y})\|_2.
\end{align*}

We upper bound the first term as, 

\begin{align}
  \|(\Pi_{\mu_{\star,\vec{e}},\mu_{\star,\overleftarrow{e}}}^\perp - \Pi_{I_1,I_1}^\perp) \circ f_{\vec{e}}(\bm{Y})\|_2 &\leq \|(\Pi_{\mu_{\star,\vec{e}},\mu_{\star,\overleftarrow{e}}}^\perp - \Pi_{I_1,I_1}^\perp) \circ f_{\vec{e}}(\bm{Y})\|_1 \\ 
  &\le g(\varepsilon, D, \Delta) \|f_{\vec{e}}(\bm{Y})\|_1  \\ 
  &\leq D g(\varepsilon, D, \Delta)
\end{align}

where we used $\|AB\|_1 \leq \|A\|_1 \|B\|_\infty \leq \|A\|_1 \|B\|_1$ and \autoref{lem:projection_difference}. For the upper bound on $\|f_{\vec{e}}(\bm{Y})\|_1$, we use $\|f_{\vec{e}}(\bm{Y})\|_1 = \Tr f_{\vec{e}}(\bm{Y}) \leq D \Tr\bm{Y} \leq D $.

For the second term, from \autoref{lem:depolperturbation} by $\big\|\bm{Y}\big\|_1 =1$, we get, 
$$
\big\|\Delta\Phi(\bm{Y})\big\|_1 \le D(1-\delta^2).
$$

and thus (note $\Pi_{I_1,I_1}^\perp[Y] = Y - \Tr Y / D$), 
$$\| \Pi_{I_1,I_1}^\perp \circ \Delta\Phi(\bm{Y})\|_2. \le \| \Pi_{I_1,I_1}^\perp \circ \Delta\Phi(\bm{Y})\|_1 \le 2D(1-\delta^2)$$
Combining these estimates,
\begin{align*}
    \|\Pi_{\mu_{\star,\vec{e}},\mu_{\star,\overleftarrow{e}}}^\perp \circ f_{\vec{e}}(\bm{Y})\|_2
    &\le D\left(D\, g(\varepsilon, D, \Delta) + 2D(1-\delta^2)\right) \\
    &= D^{2}\left(2(1-\delta^2) + g(\varepsilon, D, \Delta)\right) \\
    &= D^{2}\left(2\varepsilon + g(\varepsilon, D, \Delta)\right).
\end{align*}

Now, we add back the BP normalization. To upper bound the entire expression, we require a lower bound on the denominator $Z^{(v)}$. From our previous bound in \autoref{prop:bpapproximationbound}, $Z^{(v)} \ge D^{\Delta/2}(1 - D\varepsilon + \order{\varepsilon^2})$, which gives the upper bound:
\begin{equation}
    \frac{1}{Z^{(v)}} \leq \frac{1}{D^{\Delta/2}\left(1 - D\varepsilon\right)} + \order{\varepsilon^2}.
\end{equation}
We define our combined upper bound $\eta(\varepsilon, D, \Delta)$ as:
\begin{equation}
  \eta(\varepsilon, D, \Delta)
  := \frac{ D^2\left[2\varepsilon + g(\varepsilon, D, \Delta)\right]}{D^{\Delta/2}\left(1 - D\varepsilon\right)},
\end{equation}
where, as derived previously,
\begin{equation}
  g(\varepsilon, D, \Delta) = 2(D+1)\varepsilon + \order{\varepsilon^2}.
\end{equation}

Substituting this directly into the numerator bound:
\begin{align*}
  D^2\left[2\varepsilon + g(\varepsilon, D, \Delta)\right] 
  &= D^2\left[2\varepsilon + 2(D+1)\varepsilon + \order{\varepsilon^2}\right] \\
  &= 2D^2(D + 2)\varepsilon + \order{\varepsilon^2}.
\end{align*}

Next, we expand the geometric series from our inverse normalization bound for small $\varepsilon$:
\begin{equation}
    \frac{1}{1 - D\varepsilon} = 1 + D\varepsilon + \order{\varepsilon^2}.
\end{equation}

Finally, assembling $\eta(\varepsilon, D, \Delta)$ and keeping only up to the linear term in $\varepsilon$:
\begin{align*}
  \eta(\varepsilon, D, \Delta)
  &= \frac{D^2}{D^{\Delta/2}} \Big( 2(D+2)\varepsilon + \order{\varepsilon^2} \Big) \Big( 1 + D\varepsilon + \order{\varepsilon^2} \Big) \\
  &= D^{2-\Delta/2} 2(D+2)\varepsilon + \order{\varepsilon^2}.
\end{align*}

\end{proof}

We now extend this to show that if the single-excitation vertex is $\order{\varepsilon}$, so is the multi-excitation vertex. This is in general non-trivial, but can be guaranteed by strong-injectivity as we show below.

\begin{lemma}\label{lem:excitationprojectorbound}
  We have that an excitation block with any number of excited edges satisfies,
  \begin{center}
  $\Big\|$
\begin{tikzpicture}[
  scale=0.6,
  baseline={([yshift=-0.65ex] current bounding box.center)},
  tensor/.style={draw, line width=0.9pt, fill=tensorcolor, rounded corners=2pt, minimum size=8mm},
  redbond/.style={line width=1.6pt, draw=red},
  blackbond/.style={line width=1.6pt, draw=black}
]
  \def\dx{2.5}
  \def\dy{2.5}

  \node[tensor, fill=tensorcolor!50] (D) at (0,-\dy) {};
  \draw[blackbond] (D.east) -- ++(0.6,0);
  \draw[redbond] (D.west) -- ++(-0.6,0);
  \draw[blackbond] (D.south) -- ++(0,-0.6);
  \draw[blackbond] (D.north) -- ++(0,0.6);
\end{tikzpicture} $\Big\|_2 \leq \eta(\varepsilon, D, \Delta) \implies ~~ \Big\|$ \begin{tikzpicture}[
  scale=0.6,
  baseline={([yshift=-0.65ex] current bounding box.center)},
  tensor/.style={draw, line width=0.9pt, fill=tensorcolor, rounded corners=2pt, minimum size=8mm},
  redbond/.style={line width=1.6pt, draw=red},
  blackbond/.style={line width=1.6pt, draw=black}
]
  \def\dx{2.5}
  \def\dy{2.5}

  \node[tensor, fill=tensorcolor!50] (D) at (0,-\dy) {};
  \draw[redbond] (D.east) -- ++(0.6,0);
  \draw[redbond] (D.west) -- ++(-0.6,0);
  \draw[blackbond] (D.south) -- ++(0,-0.6);
  \draw[blackbond] (D.north) -- ++(0,0.6);
\end{tikzpicture}$\Big\|_2 \leq \eta(\varepsilon, D, \Delta) \cdot (1+ \order{\varepsilon})$
\end{center}
\end{lemma}

\begin{proof}
  Note that inserting a projector $B$ (satisfying $BB^\dagger \preceq \mathbbm{1}$) does not increase the Frobenius norm:
$$
\|AB\|_2^2 = \Tr(B^\dagger A^\dagger A B) = \Tr(BB^\dagger A^\dagger A) \leq \Tr(A^\dagger A) = \|A\|_2^2
$$

{However, note that $\Pi_{\mu_{\star,\vec{e}},\mu_{\star,\overleftarrow{e}}}^\perp(Y)$ is not necessarily a projector. But, by \autoref{lem:projection_difference} we have for any bounded $Y$, 
}
$$\|\Pi_{\mu_{\star,\vec{e}},\mu_{\star,\overleftarrow{e}}}^\perp(Y)-\Pi_{I_1,I_1}^\perp(Y)\|_1 \le \order{\varepsilon}$$

and thus, by a triangle inequality on the second excited leg,
\tikzset{
  diffbond/.style={
    line width=1.6pt,
    draw=red!70!black,
    dashed
  }
}
\begin{center}
    $\Big\|$
\begin{tikzpicture}[
  scale=0.6,
  baseline={([yshift=-0.65ex] current bounding box.center)},
  tensor/.style={draw, line width=0.9pt, fill=tensorcolor!50, rounded corners=2pt, minimum size=8mm},
  redbond/.style={line width=1.6pt, draw=red},
  blackbond/.style={line width=1.6pt, draw=black},
  diffbond/.style={line width=1.6pt, draw=red!70!black, dashed}
]
  \node[tensor] (D) at (0,0) {};
  \draw[redbond]   (D.west)  -- ++(-0.6,0);
  \draw[redbond]   (D.east)  -- ++(0.6,0);
  \draw[blackbond] (D.north) -- ++(0,0.6);
  \draw[blackbond] (D.south) -- ++(0,-0.6);
\end{tikzpicture}
$\Big\|_2
\le
\Big\|$
\begin{tikzpicture}[
  scale=0.6,
  baseline={([yshift=-0.65ex] current bounding box.center)},
  tensor/.style={draw, line width=0.9pt, fill=tensorcolor!50, rounded corners=2pt, minimum size=8mm},
  redbond/.style={line width=1.6pt, draw=red},
  blackbond/.style={line width=1.6pt, draw=black},
  diffbond/.style={line width=1.6pt, draw=red!70!black, dashed}
]
  \node[tensor] (D) at (0,0) {};
  \draw[redbond]   (D.west)  -- ++(-0.6,0);
  \draw[blackbond] (D.east)  -- ++(0.6,0);
  \draw[blackbond] (D.north) -- ++(0,0.6);
  \draw[blackbond] (D.south) -- ++(0,-0.6);
\end{tikzpicture}
$\Big\|_2
+
\Big\|$
\begin{tikzpicture}[
  scale=0.6,
  baseline={([yshift=-0.65ex] current bounding box.center)},
  tensor/.style={draw, line width=0.9pt, fill=tensorcolor!50, rounded corners=2pt, minimum size=8mm},
  redbond/.style={line width=1.6pt, draw=red},
  blackbond/.style={line width=1.6pt, draw=black},
  diffbond/.style={line width=1.6pt, draw=red!70!black, dashed}
]
  \node[tensor] (D) at (0,0) {};
  \draw[redbond]   (D.west)  -- ++(-0.6,0);
  \draw[diffbond]  (D.east)  -- ++(0.6,0);
  \draw[blackbond] (D.north) -- ++(0,0.6);
  \draw[blackbond] (D.south) -- ++(0,-0.6);
\end{tikzpicture}
$\Big\|_2$.
\end{center}
where we used the fact that $\Pi_{I_1,I_1}^\perp$ is a projector to upper bound the first term already, and, the dashed line denotes the difference $\Pi_{\mu_{\star,\vec{e}},\mu_{\star,\overleftarrow{e}}}^\perp(Y)-\Pi_{I_1,I_1}^\perp(Y)$.
For the first term by \autoref{prop:excitation_decay},
\[
\Big\|
\begin{tikzpicture}[
  scale=0.6,
  baseline={([yshift=-0.65ex] current bounding box.center)},
  tensor/.style={draw, line width=0.9pt, fill=tensorcolor!50, rounded corners=2pt, minimum size=8mm},
  redbond/.style={line width=1.6pt, draw=red},
  blackbond/.style={line width=1.6pt, draw=black}
]
  \node[tensor] (D) at (0,0) {};
  \draw[redbond]   (D.west)  -- ++(-0.6,0);
  \draw[blackbond] (D.east)  -- ++(0.6,0);
  \draw[blackbond] (D.north) -- ++(0,0.6);
  \draw[blackbond] (D.south) -- ++(0,-0.6);
\end{tikzpicture}
\Big\|_2
\le \eta(\varepsilon,D,\Delta),
\]
and for the second term,
\[
\Big\|
\begin{tikzpicture}[
  scale=0.6,
  baseline={([yshift=-0.65ex] current bounding box.center)},
  tensor/.style={draw, line width=0.9pt, fill=tensorcolor!50, rounded corners=2pt, minimum size=8mm},
  redbond/.style={line width=1.6pt, draw=red},
  blackbond/.style={line width=1.6pt, draw=black},
  diffbond/.style={line width=1.6pt, draw=red!70!black, dashed}
]
  \node[tensor] (D) at (0,0) {};
  \draw[redbond]   (D.west)  -- ++(-0.6,0);
  \draw[diffbond]  (D.east)  -- ++(0.6,0);
  \draw[blackbond] (D.north) -- ++(0,0.6);
  \draw[blackbond] (D.south) -- ++(0,-0.6);
\end{tikzpicture}
\Big\|_2
\le \order{\varepsilon}\,
\Big\|
\begin{tikzpicture}[
  scale=0.6,
  baseline={([yshift=-0.65ex] current bounding box.center)},
  tensor/.style={draw, line width=0.9pt, fill=tensorcolor!50, rounded corners=2pt, minimum size=8mm},
  redbond/.style={line width=1.6pt, draw=red},
  blackbond/.style={line width=1.6pt, draw=black}
]
  \node[tensor] (D) at (0,0) {};
  \draw[redbond]   (D.west)  -- ++(-0.6,0);
  \draw[blackbond] (D.east)  -- ++(0.6,0);
  \draw[blackbond] (D.north) -- ++(0,0.6);
  \draw[blackbond] (D.south) -- ++(0,-0.6);
\end{tikzpicture}
\Big\|_2.
\]

by \autoref{lem:projection_difference}. The constant hides factors of $D$ in converting the $1-$norm of \autoref{lem:projection_difference} into the $2-$norm required here, which we do not keep track of as this correction is second order in $\varepsilon$. Hence since $\eta(\varepsilon, D, \Delta) = \order{\varepsilon}$, we establish that for a local tensor with any (atmost $\Delta = \order{1}$) number of excited legs, we have that 
\begin{center}
  $\Big\|$
\begin{tikzpicture}[
  scale=0.6,
  baseline={([yshift=-0.65ex] current bounding box.center)},
  tensor/.style={draw, line width=0.9pt, fill=tensorcolor, rounded corners=2pt, minimum size=8mm},
  redbond/.style={line width=1.6pt, draw=red},
  blackbond/.style={line width=1.6pt, draw=black}
]
  \def\dx{2.5}
  \def\dy{2.5}

  \node[tensor, fill=tensorcolor!50] (D) at (0,-\dy) {};
  \draw[redbond] (D.east) -- ++(0.6,0);
  \draw[redbond] (D.west) -- ++(-0.6,0);
  \draw[blackbond] (D.south) -- ++(0,-0.6);
  \draw[blackbond] (D.north) -- ++(0,0.6);
\end{tikzpicture} $\Big\|_2 \leq \eta(\varepsilon, D, \Delta) + \order{\varepsilon^2}$ \end{center}

\begin{center}
  $\Big\|$
\begin{tikzpicture}[
  scale=0.6,
  baseline={([yshift=-0.65ex] current bounding box.center)},
  tensor/.style={draw, line width=0.9pt, fill=tensorcolor, rounded corners=2pt, minimum size=8mm},
  redbond/.style={line width=1.6pt, draw=red},
  blackbond/.style={line width=1.6pt, draw=black}
]
  \def\dx{2.5}
  \def\dy{2.5}

  \node[tensor, fill=tensorcolor!50] (D) at (0,-\dy) {};
  \draw[redbond] (D.east) -- ++(0.6,0);
  \draw[redbond] (D.west) -- ++(-0.6,0);
  \draw[blackbond] (D.south) -- ++(0,-0.6);
  \draw[blackbond] (D.north) -- ++(0,0.6);
\end{tikzpicture} $\Big\|_2 \leq (1 + \order{\varepsilon})\cdot \Big\|$
\begin{tikzpicture}[
  scale=0.6,
  baseline={([yshift=-0.65ex] current bounding box.center)},
  tensor/.style={draw, line width=0.9pt, fill=tensorcolor, rounded corners=2pt, minimum size=8mm},
  redbond/.style={line width=1.6pt, draw=red},
  blackbond/.style={line width=1.6pt, draw=black}
]
  \def\dx{2.5}
  \def\dy{2.5}

  \node[tensor, fill=tensorcolor!50] (D) at (0,-\dy) {};
  \draw[blackbond] (D.east) -- ++(0.6,0);
  \draw[redbond] (D.west) -- ++(-0.6,0);
  \draw[blackbond] (D.south) -- ++(0,-0.6);
  \draw[blackbond] (D.north) -- ++(0,0.6);
\end{tikzpicture} $\Big\|\leq \eta(\varepsilon, D, \Delta) \cdot (1+ \order{\varepsilon})$ 
\end{center}
By induction on the number of excited legs, we conclude the proof.
\end{proof}

We emphasize that this Lemma gives a rather loose bound, since we would expect that inserting an excitation lifts the loop contribution to an higher-order. Here, we only show that inserting an excitation keeps the loop contribution to the same order. This will only affect the value of our constant threshold, and will not have any physical consequences.

\smsubsection{Decay of loop activities}

We now establish exponential decay of loop activities using the excitation decay bounds from \autoref{prop:excitation_decay}.

\begin{lemma}[Vertex-edge relation]
\label{lem:vertex_edge_bound}
For a loop $\ell = (W,F)$ with $|W|=n$, $|F|=m$ on a graph of maximum degree $\Delta$, we have
$$
n \ge \frac{2m}{\Delta}.
$$
\end{lemma}

\begin{proof}
Since each vertex has degree at least $2$ within the loop,
$$
2m = \sum_{e \in F} 2 = \sum_{w \in W} \deg_\ell(w) \le \sum_{w \in W} \Delta = \Delta n,
$$
where $\deg_\ell(w)$ denotes the degree of vertex $w$ within the subgraph $\ell$. Thus $n \ge 2m/\Delta$.
\end{proof}

The loop activity $Z_\ell$ measures correlations that deviate from the fixed point configuration. Our goal is to show that $|Z_\ell|$ decays exponentially in the weight $m$ of the loop.

\begin{prop}[Exponential decay of loop activities]
\label{prop:loop_decay}
Let $\bm{\mu}_\star$ be the PEPS fixed point from \autoref{thm:nonTI_uniqueness}. Consider a loop $\ell = (W,F)$ with weight $m = |F|$ on a graph of maximum degree $\Delta$. With the excitation bound $
\eta \equiv \eta(\varepsilon, D, \Delta) $ from \autoref{prop:excitation_decay}, the loop activity satisfies
$$
|Z_\ell| \le \order{\eta^{2m/\Delta}}
$$

\end{prop}

\begin{proof}
We decompose the loop by successively cutting vertices and applying Cauchy--Schwarz and submultiplicativity of the trace norm.

Order the vertices $W = \{w_1, \dots, w_n\}$ such that the subgraph remains connected after removing vertices in order (this is always possible for a connected loop). For each vertex $w_i$, let $T_{w_i}$ denote the local tensor at that vertex, decorated with:
\begin{itemize}
\item Antiprojectors $\Pi_{\mu_{\star,\vec{e}}}^\perp$ on edges $e \in F$ incident to $w_i$;
\item Fixed point insertions $\mu_{\star,\vec{e}}$ on edges $e \notin F$ incident to $w_i$.
\end{itemize}

Let $\ell_{\setminus S}$ denote the subgraph obtained by removing vertices in $S \subset W$.

We illustrate the decomposition on a simple square loop. The initial configuration has all edges excited (antiprojectors inserted):
\begin{center}
  $Z = $
\begin{tikzpicture}[
  scale=0.6,
  baseline={([yshift=-0.65ex] current bounding box.center)},
  tensor/.style={draw, line width=0.9pt, fill=tensorcolor, rounded corners=2pt, minimum size=8mm},
  redbond/.style={line width=1.6pt, draw=red}
]
  \def\dx{2.5}
  \def\dy{2.5}

  \node[tensor] (A) at (0,0) {};
  \node[tensor] (B) at (\dx,0) {};
  \node[tensor] (C) at (\dx,-\dy) {};
  \node[tensor] (D) at (0,-\dy) {};

  \draw[redbond] (A.east) -- (B.west);
  \draw[redbond] (B.south) -- (C.north);
  \draw[redbond] (D.east) -- (C.west);
  \draw[redbond] (A.south) -- (D.north);
\end{tikzpicture}
\end{center}

We apply Cauchy--Schwarz to separate the bottom-left vertex from the rest. Viewing $Z_\ell$ as an inner product between the tensor at $w_1$ to the rest of the network, we apply Cauchy--Schwarz on this, to get the bound,
$$
|Z_\ell| \le \|T_{w_1}\|_2 \cdot \|Z_{\ell_{\setminus\{w_1\}}}\|_2.
$$
For the single loop network, we get, 
\begin{center}
  $|Z| \leq \Big\|$
\begin{tikzpicture}[
  scale=0.6,
  baseline={([yshift=-0.65ex] current bounding box.center)},
  tensor/.style={draw, line width=0.9pt, fill=tensorcolor, rounded corners=2pt, minimum size=8mm},
  redbond/.style={line width=1.6pt, draw=red},
  blackbond/.style={line width=1.6pt, draw=black}
]
  \def\dx{2.5}
  \def\dy{2.5}

  \node[tensor, fill=tensorcolor!50] (D) at (0,-\dy) {};
  \draw[blackbond] (D.east) -- ++(0.6,0);
  \draw[redbond] (D.north) -- ++(0,0.6);
\end{tikzpicture} $\Big\|_2$
\quad $\times$ \quad $\Big\|$
\begin{tikzpicture}[
  scale=0.6,
  baseline={([yshift=-0.65ex] current bounding box.center)},
  tensor/.style={draw, line width=0.9pt, fill=tensorcolor, rounded corners=2pt, minimum size=8mm},
  redbond/.style={line width=1.6pt, draw=red},
  blackbond/.style={line width=1.6pt, draw=black}
]
  \def\dx{2.5}
  \def\dy{2.5}

  \node[tensor] (A) at (0,0) {};
  \node[tensor] (B) at (\dx,0) {};
  \node[tensor] (C) at (\dx,-\dy) {};

  \draw[redbond] (A.east) -- (B.west);
  \draw[redbond] (B.south) -- (C.north);
  \draw[redbond] (C.west) -- ++(-0.6,0);
  \draw[blackbond] (A.south) -- ++(0,-0.6);
\end{tikzpicture}$\Big\|_2$ 
\end{center}

Next, we decompose $Z_{\ell_{\setminus\{w_1\}}}$ as a matrix product: the tensor at $w_2$ (with input from $w_1$ and outputs to other neighbors) times the remaining contracted tensor. By submultiplicativity of the Frobenius norm, 
$$
\|Z_{\ell_{\setminus\{w_1\}}}\|_2 \le \|T_{w_2}\|_2 \cdot \|Z_{\ell_{\setminus\{w_1,w_2\}}}\|_2.
$$

\begin{center}
  $\Big\|$
\begin{tikzpicture}[
  scale=0.6,
  baseline={([yshift=-0.65ex] current bounding box.center)},
  tensor/.style={draw, line width=0.9pt, fill=tensorcolor, rounded corners=2pt, minimum size=8mm},
  redbond/.style={line width=1.6pt, draw=red},
  blackbond/.style={line width=1.6pt, draw=black}
]
  \def\dx{2.5}
  \def\dy{2.5}

  \node[tensor] (A) at (0,0) {};
  \node[tensor] (B) at (\dx,0) {};
  \node[tensor] (C) at (\dx,-\dy) {};

  \draw[redbond] (A.east) -- (B.west);
  \draw[redbond] (B.south) -- (C.north);
  \draw[redbond] (C.west) -- ++(-0.6,0);
  \draw[blackbond] (A.south) -- ++(0,-0.6);
\end{tikzpicture}$\Big\|_2\leq~~$\quad  $\Big\|$
\begin{tikzpicture}[
  scale=0.6,
  baseline={([yshift=-0.65ex] current bounding box.center)},
  tensor/.style={draw, line width=0.9pt, fill=tensorcolor, rounded corners=2pt, minimum size=8mm},
  redbond/.style={line width=1.6pt, draw=red},
  blackbond/.style={line width=1.6pt, draw=black}
]
  \def\dx{2.5}
  \def\dy{2.5}

  \node[tensor, fill=tensorcolor!50] (A) at (0,0) {};
  \draw[redbond] (A.east) -- ++(0.6,0);
  \draw[blackbond] (A.south) -- ++(0,-0.6);
\end{tikzpicture}$\Big\|_2$
\quad $\times$ \quad $\Big\|$
\begin{tikzpicture}[
  scale=0.6,
  baseline={([yshift=-0.65ex] current bounding box.center)},
  tensor/.style={draw, line width=0.9pt, fill=tensorcolor, rounded corners=2pt, minimum size=8mm},
  redbond/.style={line width=1.6pt, draw=red},
  blackbond/.style={line width=1.6pt, draw=black}
]
  \def\dx{2.5}
  \def\dy{2.5}

  \node[tensor] (B) at (\dx,0) {};
  \node[tensor] (C) at (\dx,-\dy) {};

  \draw[redbond] (B.south) -- (C.north);
  \draw[redbond] (C.west) -- ++(-0.6,0);
  \draw[blackbond] (B.west) -- ++(-0.6,0);
\end{tikzpicture}$\Big\|_2$
\end{center}
This gives us, 
\begin{center}
  $|Z| \leq \Big\|$
\begin{tikzpicture}[
  scale=0.6,
  baseline={([yshift=-0.65ex] current bounding box.center)},
  tensor/.style={draw, line width=0.9pt, fill=tensorcolor, rounded corners=2pt, minimum size=8mm},
  redbond/.style={line width=1.6pt, draw=red},
  blackbond/.style={line width=1.6pt, draw=black}
]
  \def\dx{2.5}
  \def\dy{2.5}

  \node[tensor, fill=tensorcolor!50] (D) at (0,-\dy) {};
  \draw[blackbond] (D.east) -- ++(0.6,0);
  \draw[redbond] (D.north) -- ++(0,0.6);
\end{tikzpicture} $\Big\|_2$
\quad $\times$ \quad $\Big\|$
\begin{tikzpicture}[
  scale=0.6,
  baseline={([yshift=-0.65ex] current bounding box.center)},
  tensor/.style={draw, line width=0.9pt, fill=tensorcolor, rounded corners=2pt, minimum size=8mm},
  redbond/.style={line width=1.6pt, draw=red},
  blackbond/.style={line width=1.6pt, draw=black}
]
  \def\dx{2.5}
  \def\dy{2.5}

  \node[tensor, fill=tensorcolor!50] (A) at (0,0) {};
  \draw[redbond] (A.east) -- ++(0.6,0);
  \draw[blackbond] (A.south) -- ++(0,-0.6);
\end{tikzpicture}$\Big\|_2$
\quad $\times$ \quad $\Big\|$
\begin{tikzpicture}[
  scale=0.6,
  baseline={([yshift=-0.65ex] current bounding box.center)},
  tensor/.style={draw, line width=0.9pt, fill=tensorcolor, rounded corners=2pt, minimum size=8mm},
  redbond/.style={line width=1.6pt, draw=red},
  blackbond/.style={line width=1.6pt, draw=black}
]
  \def\dx{2.5}
  \def\dy{2.5}

  \node[tensor] (B) at (\dx,0) {};
  \node[tensor] (C) at (\dx,-\dy) {};

  \draw[redbond] (B.south) -- (C.north);
  \draw[redbond] (C.west) -- ++(-0.6,0);
  \draw[blackbond] (B.west) -- ++(-0.6,0);
\end{tikzpicture}$\Big\|_2$
\end{center}

Proceeding iteratively for all $n$ vertices,
$$
|Z_\ell| \le \prod_{i=1}^n \|T_{w_i}\|_2.
$$

This gives us, 
\begin{center}
  $|Z| \leq \Big\|$
\begin{tikzpicture}[
  scale=0.6,
  baseline={([yshift=-0.65ex] current bounding box.center)},
  tensor/.style={draw, line width=0.9pt, fill=tensorcolor, rounded corners=2pt, minimum size=8mm},
  redbond/.style={line width=1.6pt, draw=red},
  blackbond/.style={line width=1.6pt, draw=black}
]
  \def\dx{2.5}
  \def\dy{2.5}

  \node[tensor, fill=tensorcolor!50] (D) at (0,-\dy) {};
  \draw[blackbond] (D.east) -- ++(0.6,0);
  \draw[redbond] (D.north) -- ++(0,0.6);
\end{tikzpicture} $\Big\|_2$
\quad $\times$ \quad $\Big\|$
\begin{tikzpicture}[
  scale=0.6,
  baseline={([yshift=-0.65ex] current bounding box.center)},
  tensor/.style={draw, line width=0.9pt, fill=tensorcolor, rounded corners=2pt, minimum size=8mm},
  redbond/.style={line width=1.6pt, draw=red},
  blackbond/.style={line width=1.6pt, draw=black}
]
  \def\dx{2.5}
  \def\dy{2.5}

  \node[tensor, fill=tensorcolor!50] (A) at (0,0) {};
  \draw[redbond] (A.east) -- ++(0.6,0);
  \draw[blackbond] (A.south) -- ++(0,-0.6);
\end{tikzpicture}$\Big\|$
\quad $\times$ \quad $\Big\|$
\begin{tikzpicture}[
  scale=0.6,
  baseline={([yshift=-0.65ex] current bounding box.center)},
  tensor/.style={draw, line width=0.9pt, fill=tensorcolor, rounded corners=2pt, minimum size=8mm},
  redbond/.style={line width=1.6pt, draw=red},
  blackbond/.style={line width=1.6pt, draw=black}
]
  \def\dx{2.5}
  \def\dy{2.5}

  \node[tensor, fill=tensorcolor!50] (B) at (\dx,0) {};
  \draw[redbond] (B.south) -- ++(0,-0.6);
  \draw[blackbond] (B.west) -- ++(-0.6,0);
\end{tikzpicture}$~\Big\|$
\quad $\times$ \quad $\Big\|$
\begin{tikzpicture}[
  scale=0.6,
  baseline={([yshift=-0.65ex] current bounding box.center)},
  tensor/.style={draw, line width=0.9pt, fill=tensorcolor, rounded corners=2pt, minimum size=8mm},
  redbond/.style={line width=1.6pt, draw=red},
  blackbond/.style={line width=1.6pt, draw=black}
]
  \def\dx{2.5}
  \def\dy{2.5}

  \node[tensor, fill=tensorcolor!50] (C) at (\dx,-\dy) {};

  \draw[redbond] (C.west) -- ++(-0.6,0);
  \draw[blackbond] (C.north) -- ++(0,0.6);
\end{tikzpicture}$~\Big\|$
\end{center}

For each vertex $w_i$, after absorbing at least one antiprojector from an incident edge in $F$, the decorated tensor $T_{w_i}$ contains an excitation. By \autoref{prop:excitation_decay}, the contribution from the excitation satisfies
$$
\|T_{w_i}\|_2 \le \eta(\varepsilon, D, \Delta) + \order{\varepsilon^2}
$$
Further, inserting fixed-points cannot increase the Frobenius norm. By \autoref{lem:excitationprojectorbound}, the presence of more than one excitation can only increase the norm to a subleading value in $\varepsilon$. 


Hence, finally using \autoref{lem:vertex_edge_bound}, we obtain
$$
|Z_\ell| \le \order{\eta^{2m/\Delta}}
$$
\end{proof}

\begin{prop}[Beating exponential growth] \label{prop:loop_decay_threshold}
There exists a threshold $\varepsilon_{**}$ that is $\Theta(1)$ when $\Delta$ and $D$ are $\order{1}$, and with the scaling behavior bounded by
\begin{equation}\label{eq:loop_decay_threshold}
{
 \varepsilon_{**}= \order{
\min\left\{
\frac{D^{\Delta/2-2}}{2(D+2)}\,\frac{e^{-\Delta/4}}{(2e\Delta)^{\Delta/2}},
\frac{1}{1+2\sqrt D},
\frac{1}{D}
\right\}}.
}
\end{equation}
such that $\varepsilon < \varepsilon_{**}$ ensures loop decay with $|Z_\ell| \le e^{-c|\ell|}$ for some
\[
c>c_0=\log(2e\Delta)+\frac12,
\]
\end{prop}

\begin{proof}
To guarantee decay as $e^{-cm}$ for some $c>0$, \autoref{prop:loop_decay}
requires
\begin{equation}
\eta \le e^{-c\Delta/2}.
\end{equation}
We have the leading-order estimate from Proposition \ref{prop:excitation_decay} and Lemma \ref{lem:excitationprojectorbound},
\begin{equation}
\eta(\varepsilon,D,\Delta)=2D^{2-\Delta/2}(D+2)\varepsilon+\mathcal O(\varepsilon^2),
\end{equation}
The $O(\varepsilon^2)$ term has coefficient controlled by $\Delta$ and $D$. We did not keep track of it in order to simplify the calculations. However, since $D$ and $\Delta$ are both $\order{1}$, the coefficient is $\order{1}$. Therefore, we can set $\varepsilon$ to be smaller than a constant threshold to have
\begin{equation}
\eta(\varepsilon,D,\Delta) \le c' 2D^{2-\Delta/2}(D+2)\varepsilon ,
\end{equation}
Where $c'$ is some $\Theta(1)$ number that depends on $D$ and $\Delta$ in general and encapsulates the higher-order coefficients we do not keep track of. It is therefore sufficient, for small enough $\varepsilon$, that
\begin{equation}\label{eq:remove_epsilon_squared}
c' 2D^{2-\Delta/2}(D+2)\varepsilon \le e^{-c\Delta/2}.
\end{equation}
Solving for $\varepsilon$ gives
\begin{equation}
\varepsilon \le \frac{e^{-c\Delta/2}}{c' 2D^{2-\Delta/2}(D+2)}
=\order{ \frac{D^{\Delta/2-2}}{2(D+2)}\,e^{-c\Delta/2}}.
\end{equation}
Hence, to ensure loop decay for some $c>c_0$, it is enough to impose the
strict inequality obtained by evaluating the right-hand side at $c_0$:
\begin{equation}
\varepsilon <
\order{\frac{D^{\Delta/2-2}}{2(D+2)}\,e^{-c_0\Delta/2}}.
\end{equation}
Substituting
\[
c_0=\log(2e\Delta)+\frac12
\]
yields
\begin{align}
e^{-c_0\Delta/2}
&= \exp\!\left[-\frac{\Delta}{2}\left(\log(2e\Delta)+\frac12\right)\right] \\
&= e^{-\Delta/4}(2e\Delta)^{-\Delta/2} \\ 
&= e^{-3\Delta/4} (2\Delta)^{-\Delta/2}.
\end{align}
Therefore,
\begin{equation}
\varepsilon <
\frac{D^{\Delta/2-2}}{2(D+2)}\,\frac{e^{-3\Delta/4}}{(2\Delta)^{\Delta/2}} =\order{ \frac{1}{2D^2(D+2)}\left(\frac{D}{(2\Delta e^{3/2})}\right)^{\Delta/2}}
\end{equation}
Combining this with the auxiliary constraints from \autoref{prop:bpapproximationbound} and \autoref{lem:projection_difference},
\[
\varepsilon<\frac{1}{1+2\sqrt D},
\qquad
\varepsilon<\frac1D,
\]
gives the stated bound. We note that Eq.~\ref{eq:loop_decay_threshold} only gives a bound on the scaling behavior of the threshold because we already take a threshold in Eq.~\ref{eq:remove_epsilon_squared}. This threshold could have a different scaling behavior that overwhelms the scaling in Eq.~\ref{eq:loop_decay_threshold}. However, we are always guaranteed that the threshold is $\Theta(1)$ as long as $D$ and $\Delta$ are $\order{1}$.
\end{proof}

\emph{Implication of thresholds.} We have, 
\begin{enumerate}
    \item[(i)] The threshold for uniqueness, 
    $$
    \varepsilon_* = \frac{1}{2\Delta-1}
    $$
    \item[(ii)] The threshold for loop decay, 
    $$
    \varepsilon_{**} = \min\left\{
\order{\left(\frac{D}{2e^{3/2}\Delta }\right)^{\Delta/2}},
\order{\frac{1}{D}}
\right\}.
    $$
\end{enumerate}
Now, for $\varepsilon_{**}$ we note that whenever $D < 2e^{3/2}\Delta$ the factor
$$\order{\left(\frac{D}{2e^{3/2}\Delta }\right)^{\Delta/2}}$$
dominates owing to the exponential in $\Delta/2$. Thus, given $D < 2e^{3/2}\Delta$ we then have $\varepsilon_{**} < \varepsilon_*$ generically owing to the combinatorial dependence of $\Delta$. On the other hand, $D > 2e^{3/2}\Delta$, then this bound becomes vacuous and the $\order{1/D}$ factor dominates. Even still, we then have  $\varepsilon_{**} \sim 1/D < \varepsilon_* \sim 1/\Delta $ since $D > 2e^{3/2}\Delta$. In practice, we expect the cluster expansion to converge for a wider range than $\varepsilon_{**}$ owing to additional structure in the loops (e.g., see~\cite{midha2026}), as well as the fact that the proof methods are rather conservative.
\newpage 
\smsection{Algorithmic Locality}

We first establish that the message-passing dynamics exhibits finite propagation speed: perturbations in a local region affect only nearby messages after finite time evolution.

\smsubsection{Message-passing dynamics}

Consider the global synchronous message-passing dynamics. Starting from an initial configuration $\bm{\mu}^{(0)} = (\mu_{\vec{e}}^{(0)})_{\vec{e}\in\vec{E}}$, the state at time $t$ is given by
\begin{equation}
\label{eq:message_dynamics}
\bm{\mu}^{(t)} := \bm{F}^t[\bm{\mu}^{(0)}],
\end{equation}
where $\bm{F}: \KK_G \to \KK_G$ is the global normalized message-passing map defined component-wise: for each directed edge $\vec{e}=(v,n)\in\vec{E}$,
\begin{equation}\label{eq:normalized_update_rule}
[\bm{F}(\bm{\mu})]_{\vec{e}} := F_{\vec{e}}(\bm{\mu}) := \frac{\Phi_{(v,n)}\!\left(
\bigotimes_{m \in \NN(v)\setminus\{n\}} \mu_{(m,v)}
\right)}{\Tr\left[\Phi_{(v,n)}\!\left(
\bigotimes_{m \in \NN(v)\setminus\{n\}} \mu_{(m,v)}
\right)\right]},
\end{equation}
where the normalization ensures each message remains in $\KK(\HH_{\vec{e}})$ (trace-normalized positive operators).

\smsubsection{Finite speed of propagation}
We first highlight a useful fact about message-passing dynamics, in that it has a \emph{lightcone}. That is, local information on the graph propagates at most ballistically (with graph distance). 

\begin{defn}[Support of perturbation]
Let $\bm{\mu}^{(0)}$ and $\tilde{\bm{\mu}}^{(0)}$ be two initial configurations. The \emph{support} of their difference is the set of vertices
\[
A := \{v \in V: \exists\, n\in\NN(v) \text{ such that } \mu_{(v,n)}^{(0)} \neq \tilde{\mu}_{(v,n)}^{(0)}\}.
\]
We say the initial perturbation is \emph{supported on $A$}.
\end{defn}

\begin{lemma}[Finite propagation speed]
\label{lem:finite_propagation}
Let $\bm{\mu}^{(0)}$ and $\tilde{\bm{\mu}}^{(0)}$ be two initial configurations such that their difference is supported on $A \subseteq V$. Let $\bm{\mu}^{(t)} = \bm{F}^t[\bm{\mu}^{(0)}]$ and $\tilde{\bm{\mu}}^{(t)} = \bm{F}^t[\tilde{\bm{\mu}}^{(0)}]$ be the corresponding time-evolved states.

Then for any region $B \subseteq V$ satisfying
\[
d(A,B) > t,
\]
we have
\[
\mu_e^{(t)} = \tilde{\mu}_e^{(t)}
\qquad \text{for all } {\vec{e}}=(v,n) \text{ with } v \in B.
\]
\end{lemma}

\begin{proof}
We prove by induction on $t$ that perturbations propagate at most one edge per time step. 
\[
\begin{tikzpicture}[
  scale=0.8,
  baseline={([yshift=-0.65ex] current bounding box.center)},
  tensor/.style={
    draw, line width=0.9pt, fill=black!15, rounded corners=2pt, minimum size=6mm
  },
  tensorA/.style={
    draw, line width=0.9pt, fill=red!40, rounded corners=2pt, minimum size=6mm
  },
  bond/.style={line width=1.5pt, draw=black!70},
  arr/.style={line width=1.8pt, draw=red, ->}
]

\def\dx{1.4}
\def\dy{1.4}
\def\shear{0.0}

\foreach \t/\xoff in {0/0, 1/6.5, 2/13} {
  \begin{scope}[xshift=\xoff cm, xslant=\shear, yscale=0.98]

    \node[font=\large\bfseries] at (2*\dx, 0.8*\dy) {$t=\t$};

    \foreach \r in {0,...,4}{
      \foreach \c in {0,...,4}{
        \node[tensor] (T-\t-\c-\r) at (\c*\dx, -\r*\dy) {};
      }
    }
    \node[tensorA] (T-\t-2-2) at (2*\dx, -2*\dy) {$A$};

    \foreach \r in {0,...,4}{
      \foreach \c in {0,...,3}{
        \pgfmathtruncatemacro{\cp}{\c+1}
        \draw[bond] (T-\t-\c-\r) -- (T-\t-\cp-\r);
      }
    }
    \foreach \c in {0,...,4}{
      \foreach \r in {0,...,3}{
        \pgfmathtruncatemacro{\rp}{\r+1}
        \draw[bond] (T-\t-\c-\r) -- (T-\t-\c-\rp);
      }
    }

    \foreach \r in {0,...,4}{
      \foreach \c in {0,...,4}{
        \pgfmathtruncatemacro{\du}{abs(\c-2) + abs(\r-2)}
        
        \pgfmathtruncatemacro{\isactive}{\du <= \t ? 1 : 0}
        
        \ifnum\isactive=1
          \ifnum\r>0
            \pgfmathtruncatemacro{\rp}{\r-1}
            \pgfmathtruncatemacro{\dv}{abs(\c-2) + abs(\rp-2)}
            \ifnum\du<\dv 
              \draw[arr] (T-\t-\c-\r) -- (T-\t-\c-\rp);
            \fi
          \fi
          \ifnum\r<4
            \pgfmathtruncatemacro{\rp}{\r+1}
            \pgfmathtruncatemacro{\dv}{abs(\c-2) + abs(\rp-2)}
            \ifnum\du<\dv
              \draw[arr] (T-\t-\c-\r) -- (T-\t-\c-\rp);
            \fi
          \fi
          \ifnum\c>0
            \pgfmathtruncatemacro{\cp}{\c-1}
            \pgfmathtruncatemacro{\dv}{abs(\cp-2) + abs(\r-2)}
            \ifnum\du<\dv
              \draw[arr] (T-\t-\c-\r) -- (T-\t-\cp-\r);
            \fi
          \fi
          \ifnum\c<4
            \pgfmathtruncatemacro{\cp}{\c+1}
            \pgfmathtruncatemacro{\dv}{abs(\cp-2) + abs(\r-2)}
            \ifnum\du<\dv
              \draw[arr] (T-\t-\c-\r) -- (T-\t-\cp-\r);
            \fi
          \fi
        \fi
      }
    }
  \end{scope}
}
\end{tikzpicture}
\]
For the base case: by assumption, $\mu_{\vec{e}}^{(0)} = \tilde{\mu}_{\vec{e}}^{(0)}$ for all edges ${\vec{e}}$ with source vertex outside $A$. Assume now that the claim holds for time $t$: if $d(v,A) > t$, then $\mu_{\vec{e}}^{(t)} = \tilde{\mu}_{\vec{e}}^{(t)}$ for all edges $e$ with source vertex $v$.

Now consider time $t+1$. For any directed edge ${\vec{e}}=(v,n)$, the message at time $t+1$ is
\[
\mu_{\vec{e}}^{(t+1)} = F_{\vec{e}}(\bm{\mu}^{(t)}) = \frac{1}{\Tr(\cdots)}\Phi_{(v,n)}\!\left(
\bigotimes_{m \in \NN(v)\setminus\{n\}} \mu_{(m,v)}^{(t)}
\right).
\]

If $d(v,A) > t+1$, then all neighbors $m \in \NN(v)$ satisfy $d(m,A) \ge d(v,A) - 1 > t$. By the inductive hypothesis, $\mu_{(m,v)}^{(t)} = \tilde{\mu}_{(m,v)}^{(t)}$ for all $m \in \NN(v)$.

Therefore,
\[
\mu_e^{(t+1)} =\frac{1}{\Tr(\cdots)} \Phi_{(v,n)}\!\left(
\bigotimes_{m \in \NN(v)\setminus\{n\}} \mu_{(m,v)}^{(t)}
\right)
= \frac{1}{\Tr(\cdots)}\Phi_{(v,n)}\!\left(
\bigotimes_{m \in \NN(v)\setminus\{n\}} \tilde{\mu}_{(m,v)}^{(t)}
\right)
= \tilde{\mu}_e^{(t+1)}.
\]

This completes the induction.
\end{proof}

\begin{corollary}[Light cone structure]
The message-passing dynamics $\bm{F}$ exhibits a causal light cone: information propagating from a local region $A$ at time $0$ can only affect messages at edges ${\vec{e}}=(v,n)$ with $d(v,A) \le t$ at time $t$.
\end{corollary}

\smsubsection{Weak perturbations preserve injectivity}

We will review a technical argument that shows weak perturbations on tensors can only change their injectivity parameters in a controlled manner. This follows from standard results in perturbation theory of singular values, following Weyl.

\begin{theorem}[Weyl's perturbation \cite{weyl1912asymptotische}] \label{thm:weylperturbation}
    Let $A \in \C^{m \times n}$. Let $\tilde{A} = A+ E$ be a perturbation of $A$. Then the ordered set of singular values of $A$ and $\tilde{A}$ satisfy,
    \begin{equation}
        |\tilde{\sigma}_i - \sigma_i| \leq \|E\|_\infty
    \end{equation}
\end{theorem}

Using Weyl's stability of singular values under weak perturbations, we show that for an injective tensor with $\varepsilon < \varepsilon_\text{th}$ for some threshold $\varepsilon_\text{th} > 0$, all perturbations of the tensor with operator norm $\order{\varepsilon_\text{th} - \varepsilon}$ preserve the injectivity threshold. 

\begin{lemma}[Stability of the injectivity parameter under perturbations]\label{lem:injectivity_stability}
Let $A \in \C^{m\times n}$ have ordered singular values $\{1 \geq \dots \geq \delta > 0\}$. Let $\varepsilon := 1-\delta^2$ and assume $\varepsilon < \varepsilon_{\mathrm{th}}$. Now let $\widetilde A = A+E$ and let,
\begin{equation}
    \widetilde \sigma_1 \ge \widetilde \sigma_2 \ge \cdots \ge \widetilde \sigma_r
\end{equation}
be the singular values of $\widetilde A$. Define the perturbed ratio
\begin{equation}
    \widetilde\delta := \frac{\widetilde \sigma_r}{\widetilde \sigma_1},
\qquad
\widetilde\varepsilon := 1-\widetilde\delta^2.
\end{equation}
If
\begin{equation}
\|E\|_\infty
<
\frac{\varepsilon_{\mathrm{th}}-\varepsilon}
{2\sqrt{1-\varepsilon_{\mathrm{th}}}\,(1+\sqrt{1-\varepsilon_{\mathrm{th}}})}
+ \order{(\varepsilon_{\mathrm{th}}-\varepsilon)^2}
\end{equation}then the perturbed matrix satisfies $\widetilde{\varepsilon} < \varepsilon_{\mathrm{th}} .
$.
\end{lemma}

\begin{proof}
By Weyl's perturbation bound~\autoref{thm:weylperturbation},

$$
|\widetilde \sigma_1-\sigma_1| \le \|E\|_\infty,
\qquad
|\widetilde \sigma_r-\sigma_r| \le \|E\|_\infty.
$$
Since $\sigma_1=1$ and $\sigma_r=\delta$, this gives
$$\widetilde \sigma_1 \le 1+\|E\|_\infty,
\qquad
\widetilde \sigma_r \ge \delta-\|E\|_\infty.
$$
Therefore

$$
\widetilde\delta
=
\frac{\widetilde \sigma_r}{\widetilde \sigma_1}
\ge
\frac{\delta-\|E\|_\infty}{1+\|E\|_\infty}.
$$
To guarantee $\widetilde\delta > \delta_{\mathrm{th}}$, it is enough to impose

$$
\frac{\delta-\|E\|_\infty}{1+\|E\|_\infty} > \delta_{\mathrm{th}}.
$$
Since $1+\|E\|_\infty > 0$, this is equivalent to
$$
\delta-\|E\|_\infty > \delta_{\mathrm{th}}(1+\|E\|_\infty),
$$
that is,
$$\delta-\delta_{\mathrm{th}} > (1+\delta_{\mathrm{th}})\|E\|_\infty.$$
Hence it suffices that
\begin{align*}
\|E\|_\infty
&<
\frac{\delta-\delta_{\mathrm{th}}}{1+\delta_{\mathrm{th}}} \\
&=
\frac{\sqrt{1-\varepsilon}-\sqrt{1-\varepsilon_{\mathrm{th}}}}
{1+\sqrt{1-\varepsilon_{\mathrm{th}}}} \\
&=
\frac{
\frac{\varepsilon_{\mathrm{th}}-\varepsilon}{2\sqrt{1-\varepsilon_{\mathrm{th}}}}
+ O\!\left((\varepsilon_{\mathrm{th}}-\varepsilon)^2\right)
}
{1+\sqrt{1-\varepsilon_{\mathrm{th}}}} \\
&=
\frac{\varepsilon_{\mathrm{th}}-\varepsilon}
{2\sqrt{1-\varepsilon_{\mathrm{th}}}\bigl(1+\sqrt{1-\varepsilon_{\mathrm{th}}}\bigr)}
+ \order{(\varepsilon_{\mathrm{th}}-\varepsilon)^2}.
\end{align*}
Under this condition, \(\widetilde\delta > \delta_{\mathrm{th}}\), and therefore
$$
\widetilde\varepsilon
=
1-\widetilde\delta^2
<
1-\delta_{\mathrm{th}}^2
=
\varepsilon_{\mathrm{th}}.
$$

\end{proof}

\smsubsection{Locality of fixed points}

We now establish that local perturbations to a strongly injective PEPS lead to exponentially decaying modifications of the fixed point messages. This result combines the finite propagation speed with the Banach contraction property to show that the effect of a local change diminishes exponentially with distance.

Consider a PEPS on graph $G=(V,E)$ that is $\varepsilon$-injective with $\varepsilon < \varepsilon_{*}$ satisfying the contraction condition from \autoref{prop:banach_contr}. Let $\bm{\mu}_\star = (\mu_{\star,\vec{e}})_{e\in\vec{E}}$ be the unique fixed point from \autoref{thm:nonTI_uniqueness}.

Now suppose we make a \emph{weak local perturbation} to the PEPS by modifying the tensors at vertices in a region $A \subseteq V$, obtaining a new PEPS where the tensors $T_A$, are replaced by perturbed tensors $T'_A$. The ``weakness" of the perturbation is used to simplify the proof, however we expect the principle of algorithmic locality to hold for generic injectivity-preserving perturbations as well.

Let $\bm{F}'$ denote the message-passing map for the perturbed PEPS, and let $\bm{\mu}'_\star$ be its unique fixed point. Our goal is to quantify how $\|\mu'_{\star,\vec{e}} - \mu_{\star,\vec{e}}\|_1$ depends on $d(e,A)$.

\begin{theorem}[Exponential decay of perturbations]
\label{thm:exponential_decay_perturbation}
Let $\bm{\mu}_\star$ and $\bm{\mu}'_\star$ be the fixed points of the original and perturbed PEPS as described above. Let $q < 1$ be the contraction constant from \autoref{prop:banach_contr} for the original PEPS. For weak perturbations satisfying, 
\begin{equation}
  \|T_A - T'_A\|_\infty = \order{(\varepsilon_* - \varepsilon)}
\end{equation}

where the operator norm is taken relative to the convention that virtual legs are input and the physical leg is output; then for any directed edge $\vec{e}=(v,n)$ with $d(v,A) = r$, the difference in fixed point messages satisfies
\[
\|\mu'_{\star,\vec{e}} - \mu_{\star,\vec{e}}\|_1 \le \order{q^{r-1}} \le \order{e^{-r/\xi_*}}
\]
with $1/\xi_* = \order{\log{\varepsilon_*/\varepsilon}}$.
\end{theorem}

\begin{proof}




We aim to bound $\|\mu_{\star,\vec{e}} - \mu'_{\star,e}\|_1$, for a given $e$. We will use the following triangle inequality:
\begin{equation} \label{eq:triangleLPPLfp}
    \|\mu_{\star,\vec{e}} - \mu'_{\star,\vec{e}}\|_1 \leq  \underbrace{\|\mu_{\vec{e}}^{(t)} - \mu'_{\star,\vec{e}}\|_1}_{\text{Banach}} + \underbrace{\|\mu_{\star,\vec{e}} - \mu_{\vec{e}}^{(t)}\|_1}_{\text{Lightcone}}
\end{equation}
where $\mu_{\vec{e}}^{(t)}$ is the evolution under the perturbed update rule, starting from $\mu_{\star}$, i.e.,
\begin{equation}
    \bm{\mu}^{(t+1)} = \bm{F}'( \bm{\mu}^{(t)})
\end{equation}
with $\mu_{\vec{e}}^{(0)} = \mu_{\star,\vec{e}}$.

By the contractive property of $\bm{F}'$, we know $\bm{\mu}^{(t)}$ is exponentially close to the fixed point
\begin{equation}
    \underbrace{\|\mu_{\vec{e}}^{(t)} - \mu'_{\star,\vec{e}}\|_1}_{\text{Banach}}  \le \epsilon q^t
\end{equation}
Where $\epsilon = \max_e \|\mu'_{\star,\vec{e}} - \mu_{\star,\vec{e}} \|_1$. Next, by Lemma \ref{lem:finite_propagation}, $\mu_e^{(t)}$ is different from $\mu_{\star,\vec{e}}$ only when $e$ is within the lightcone at time $t$. That is, for $t < r$, 

\begin{equation}
    \underbrace{\|\mu_{\star,\vec{e}} - \mu_{\vec{e}}^{(t)}\|_1}_{\text{Lightcone}} = 0
\end{equation}

Therefore, we optimize the upper-bound in Eq.~\ref{eq:triangleLPPLfp} as a function of time and obtain the tightest bound at $t=r-1$, viz.,
\begin{equation}
    \|\mu'_{\star,\vec{e}} - \mu_{\vec{e}}^{(t)} \|_1  \le \epsilon q^{r-1}
\end{equation}
Therefore, we get,
\begin{equation}
    \|\mu'_{\star,\vec{e}} - \mu_{\vec{e}}^{(t)} \|_1  \le \order{q^{r-1}}
\end{equation}
Now note that $q = 2(\Delta-1)\varepsilon/(1-\varepsilon) = \order{\varepsilon/\varepsilon_*}$, leading to 
\begin{equation}
    \|\mu'_{\star,\vec{e}} - \mu_{\vec{e}}^{(t)} \|_1  \le \order{e^{-r/\xi_*}}
\end{equation}
with $1/\xi_* = \order{\log{\varepsilon_*/\varepsilon}}$.
\end{proof}

\smsubsection{Locality of expectation values}
We now establish that local perturbations to the PEPS lead to exponentially suppressed changes in observable measurements at distant regions. The key insight is that observables are computed via a cluster expansion built on the fixed point messages, combining:
\begin{enumerate}
\item \emph{Locality of fixed points} (\autoref{thm:exponential_decay_perturbation}): Perturbations in region $A$ affect fixed point messages exponentially in distance with a finite length scale $\xi_*$ [given $\varepsilon < \varepsilon_*$].
\item \emph{Loop activity decay}: Loop activities decay as $|Z_\ell| \le \order{\eta^{2m/\Delta}}$ where $m$ is the number of edges in the loop [given $\varepsilon < \varepsilon_{**}$].
\item \emph{Cluster expansion technology}: Observables are expressed as corrections built from loop activities in terms of connected clusters intersecting the observable region.
\end{enumerate}

\begin{theorem}[Algorithmic Locality]
\label{thm:observable_locality}
Consider a PEPS on graph $G=(V,E)$ with injectivity parameter $\varepsilon$ satisfying both:
\begin{enumerate}
\item The contraction condition for unique fixed points ensured by $\varepsilon < \varepsilon_*$ (\autoref{thm:nonTI_uniqueness}), and
\item Decay of loops ensured by $\varepsilon < \varepsilon_{**}$  (from \autoref{prop:loop_decay_threshold})
\end{enumerate}

Let $|\psi\rangle$ and $|\psi'\rangle$ be PEPS states differing only by local tensors in region $A \subseteq V$. Then, 
for weak perturbations satisfying, 
\begin{equation}
  \|T_A - T'_A\|_\infty = \order{\varepsilon_* - \varepsilon}
\end{equation}
For any local observable $O_B$ with support on region $B \subseteq V$ with $d(A,B) = R$ (graph distance), the change in expectation value satisfies
\[
\left|\langle O_B \rangle_{\psi'} - \langle O_B \rangle_{\psi}\right| \le \order{\|O_B\|_\infty\cdot e^{-R/\xi_{**}}},
\]
where $1/\xi_{**} = \order{\min\{1/\xi_*, c-c_0\}} > 0$, with $\xi_* = \order{\log{\varepsilon_*/\varepsilon}}$ from \autoref{thm:exponential_decay_perturbation}.
\end{theorem}
We provide two proofs of this theorem. The first proof of this theorem is constructive, mirroring how tensor network BP algorithm works in practice, and uses the locality of fixed-point. This requires running the message-passing equations again on the modified state, and then bounding the remaining difference via the cluster expansion. The second proof is more direct and less illustrative, and however not useful in practical tensor network algorithms.

\begin{proof}
We use the observable cluster expansion from \autoref{prop:localobsexpansionformal} for both states.

By the cluster expansion, we have the formal expressions,
\begin{align}
\langle O_B \rangle_{\psi} &= \langle O_B \rangle_{\text{BP}} \cdot \exp\left(\sum_{\substack{\text{connected} \, \mathbf{W} \\ \supp(\mathbf{W}) \cap B \neq\emptyset }} \phi(\mathbf{W}) Z^{O_B}_{\mathbf{W}}\right), \\
\langle O_B \rangle_{\psi'} &= \langle O_B \rangle_{\text{BP}}' \cdot \exp\left(\sum_{\substack{\text{connected} \, \mathbf{W} \\ \supp(\mathbf{W}) \cap B \neq\emptyset }} \phi(\mathbf{W}) Z^{O_B'}_{\mathbf{W}}\right).
\end{align}

We can write the difference in observable measurement before and after perturbation as,
\begin{align}
\langle O_B \rangle_{\psi'} - \langle O_B \rangle_{\psi} 
&= \left(\langle O_B \rangle_{\text{BP}}' - \langle O_B \rangle_{\text{BP}}\right) \exp\left(\sum \phi(\mathbf{W}) Z^{O_B'}_{\mathbf{W}}\right) \\
&\quad + \langle O_B \rangle_{\text{BP}} \left(\exp\left(\sum \phi(\mathbf{W}) Z^{O_B'}_{\mathbf{W}}\right) - \exp\left(\sum \phi(\mathbf{W}) Z^{O_B}_{\mathbf{W}}\right)\right).
\end{align}

By triangle inequality,
\begin{equation}
    \left|\langle O_B \rangle_{\psi'} - \langle O_B \rangle_{\psi}\right| 
\le \underbrace{\left|\langle O_B \rangle_{\text{BP}}' - \langle O_B \rangle_{\text{BP}}\right| \cdot \left|\exp\left(\sum \phi(\mathbf{W}) Z^{O_B'}_{\mathbf{W}}\right)\right|}_{\text{BP difference}}
+ \underbrace{\left|\langle O_B \rangle_{\text{BP}}\right| \cdot \left|\exp\left(\sum \phi(\mathbf{W}) Z^{O_B'}_{\mathbf{W}}\right) - \exp\left(\sum \phi(\mathbf{W}) Z^{O_B}_{\mathbf{W}}\right)\right|}_{\text{cluster difference}}.
\end{equation}

To bound the exponential factors, note that by the cluster expansion convergence (which requires $c > c_0$), the cluster sum is absolutely convergent,
\begin{equation}
    \left|\sum_{\substack{\text{connected} \, \mathbf{W} \\ \supp(\mathbf{W}) \cap B \neq\emptyset }} \phi(\mathbf{W}) Z^{O_B}_{\mathbf{W}}\right| \le \order{|B|e^{-(c-c_0)}}.
\end{equation}
Therefore, $\left|\exp\left(\sum \phi(\mathbf{W}) Z^{O_B}_{\mathbf{W}}\right)\right| \le e^{\order{|B|  \cdot e^{-(c-c_0)}}}$ and similarly $\left|\exp\left(\sum \phi(\mathbf{W}) Z^{O_B'}_{\mathbf{W}}\right)\right| \le e^{\order{|B|  \cdot e^{-(c-c_0)}}}$, both constants.

By \autoref{thm:exponential_decay_perturbation}, the fixed point messages at distance $d \ge R$ from $A$ satisfy
\begin{equation}
    \|\mu'_{\star,\vec{e}} - \mu_{\star,\vec{e}}\|_2 \le e^{-R/\xi_*}.
\end{equation}
Since $\langle O_B \rangle_{\text{BP}}$ depends only on fixed point messages in a on the local tensors at the boundary of $B$, and $d(B,A) = R$, we obtain
\begin{equation}
    \left|\langle O_B \rangle_{\text{BP}}' - \langle O_B \rangle_{\text{BP}}\right| \le \order{ \|O_B\|_\infty\cdot e^{-R/\xi_*}}
\end{equation}

We now decompose clusters into two types based on the distance from the measurement region. Choose cutoff radius $R_\text{th}$. Partition clusters contributing to $O_B$ into:
\begin{itemize}
\item \textbf{Near clusters}: $\mathcal{W}_{\text{near}} = \{\mathbf{W}: \supp(\mathbf{W}) \cap B \neq \emptyset, \, |\mathbf{W}| \le R_\text{th}\}$
\item \textbf{Far clusters}: $\mathcal{W}_{\text{far}} = \{\mathbf{W}: \supp(\mathbf{W}) \cap B \neq \emptyset, \, |\mathbf{W}| > R_\text{th}\}$
\end{itemize}

For clusters with $|\mathbf{W}| \le R_\text{th}$ loops, each loop activity $Z_\ell$ depends on fixed point messages along the loop. Since the cluster has support overlapping $B$ and $|\mathbf{W}| \le R_\text{th} $, the cluster extends at most distance $R_\text{th}$ from $B$, so it is at distance at least $(R - R_\text{th})$ from $A$.

Each loop $\ell$ with weight $|\ell|$ (number of excited edges) involves message tensors along the loop. The loop activity depends on the fixed point messages at each vertex along the loop, with at most $\Delta$ (the maximum degree) insertions per vertex. Also, we have the baseline loop decay $|Z_\ell| \le e^{-c|\ell|}$. Now, we can bound the change in loop activity upon making the perturbation through a sequence of triangle inequalities, one upon each excited edge and each fixed-point insertion. There are $\order{|\ell|}$ number of excitations and fixed point insertions. We illustrate this for a simple loop,
\newcommand{\drawloop}[4]{%
  \begin{tikzpicture}[
    scale=0.6,
    baseline={([yshift=-0.5ex] current bounding box.center)},
    tensor/.style={draw, line width=0.8pt, fill=black!15, rounded corners=1.5pt, minimum size=4mm},
    b/.style={line width=1.5pt, draw=blue!60},
    r/.style={line width=1.5pt, draw=red!60},
    k/.style={line width=1.5pt, draw=black}
  ]
    \node[tensor] (T1) at (0, 1) {};
    \node[tensor] (T2) at (1, 1) {};
    \node[tensor] (T3) at (1, 0) {};
    \node[tensor] (T4) at (0, 0) {};
    
    \draw[#1] (T1) -- (T2); 
    \draw[#2] (T2) -- (T3); 
    \draw[#3] (T3) -- (T4); 
    \draw[#4] (T4) -- (T1); 
  \end{tikzpicture}%
}

\begin{align}
&\left| \, \drawloop{b}{b}{b}{b} - \drawloop{r}{r}{r}{r} \, \right| \\[1em]
&\quad = \left| 
\begin{aligned}
  &\quad\, \left( \drawloop{b}{b}{b}{b} - \drawloop{r}{b}{b}{b} \right) \\
  & + \left( \drawloop{r}{b}{b}{b} - \drawloop{r}{r}{b}{b} \right) \\
  & + \left( \drawloop{r}{r}{b}{b} - \drawloop{r}{r}{r}{b} \right) \\
  & + \left( \drawloop{r}{r}{r}{b} - \drawloop{r}{r}{r}{r} \right)
\end{aligned}
\right| \\[1em]
&\quad \le \max_{e \in \ell} \|\mu'_{\star,\vec{e}} - \mu_{\star,\vec{e}}\|_2\cdot\quad  \left(\left| \, \drawloop{k}{b}{b}{b} \, \right| 
+ \left| \, \drawloop{r}{k}{b}{b} \, \right| 
+ \left| \, \drawloop{r}{r}{k}{b} \, \right| 
+ \left| \, \drawloop{r}{r}{r}{k} \, \right|\right) 
\\[1em]
&\quad \le \max_{e \in \ell} \|\mu'_{\star,\vec{e}} - \mu_{\star,\vec{e}}\|_2\cdot \order{\Delta |\ell| e^{-c|\ell|}}
\end{align}
where the last step follows from (i) decay of loops, and (ii) number of total edges/fixed-point insertions in a loop being $\order{\Delta|\ell|}$. Thus we have,
\begin{equation}
    |Z'_\ell - Z_\ell| \le   \order{\Delta |\ell| \cdot \max_{e \in \ell} \|\mu'_{\star,\vec{e}} - \mu_{\star,\vec{e}}\|_2} \cdot e^{-c|\ell|} \le \order{\Delta \cdot q^{(R - R_\text{th})} \cdot|\ell| \cdot  e^{-c|\ell|}}.
\end{equation}

Now, for a cluster $\mathbf{W} = \{(\ell_i, \eta_i)\}$ with $|\mathbf{W}| = \sum_i \eta_i |\ell_i|$ loops, by a triangle inequality on each loop and the bound on the individual loop difference, we get,
\begin{align}
    |Z^{O_B'}_{\mathbf{W}} - Z^{O_B}_{\mathbf{W}}| &= \left| \prod_{\ell_i} Z'_{\ell_i} - \prod_{\ell_i} Z_{\ell_i} \right| \nonumber \\
    &= \left| \sum_{\ell_i} \left( \prod_{\ell_1 < \ell_i} Z'_{\ell_1} \right) (Z'_{\ell_i} - Z_{\ell_i}) \left( \prod_{\ell_2 > \ell_i} Z_{\ell_2} \right) \right| \nonumber \\
    &\le \sum_{\ell_i} \left( \prod_{\ell_1 < \ell_i} |Z'_{\ell_1}| \right) |Z'_{\ell_i} - Z_{\ell_i}| \left( \prod_{\ell_2 > \ell_i} |Z_{\ell_2}| \right) \nonumber \\
    &\le \sum_{\ell_i} \left( \prod_{\ell_1 < \ell_i} e^{-c|\ell_1|} \right) |Z'_{\ell_i} - Z_{\ell_i}| \left( \prod_{\ell_2 > \ell_i} e^{-c|\ell_2|} \right) \nonumber \\
    &= \sum_{\ell_i} \exp\left( -c \sum_{\ell_j \neq \ell_i} |\ell_j| \right) |Z'_{\ell_i} - Z_{\ell_i}| \nonumber \\
    &\le \sum_{\ell_i} \exp\left( -c \sum_{\ell_j \neq \ell_i} |\ell_j| \right) \order{\Delta |\ell_i| \cdot \max_{e \in \ell_i} \|\mu'_{\star,\vec{e}} - \mu_{\star,\vec{e}}\|_2} e^{-c|\ell_i|} \nonumber \\
    &\le \sum_{\ell_i} \order{\Delta \cdot q^{(R - R_\text{th})} \cdot|\ell_i|} \exp\left( -c \sum_{\ell_j} |\ell_j| \right).
\end{align}

Summing over all the near clusters:
\begin{align}
\left|\sum_{\mathbf{W} \in \mathcal{W}_{\text{near}}} \phi(\mathbf{W}) (Z^{O_B'}_{\mathbf{W}} - Z^{O_B}_{\mathbf{W}})\right| 
&\le \sum_{\mathbf{W} \in \mathcal{W}_{\text{near}}} |\phi(\mathbf{W})| \cdot |Z^{O_B'}_{\mathbf{W}} - Z^{O_B}_{\mathbf{W}}| \\
&\le  \Delta \cdot q^{(R - R_\text{th})} \sum_{\mathbf{W} \in \mathcal{W}_{\text{near}}} |\phi(\mathbf{W})| e^{-c|\mathbf{W}|}\sum_i \eta_i|\ell_i|  \\
&\le  \order{\Delta q^{(R - R_\text{th})}} \sum_{n=1}^{R_\text{th}} \sum_{\substack{\mathbf{W}: |\mathbf{W}|=n \\ \supp(\mathbf{W}) \cap B \neq \emptyset}} |\phi(\mathbf{W})| n e^{-cn}  \\ 
&\le   \order{\Delta q^{(R - R_\text{th})}} \sum_{n=1}^{R_\text{th}} n\sum_{\substack{\mathbf{W}: |\mathbf{W}|=n \\ \supp(\mathbf{W}) \cap B \neq \emptyset}} |\phi(\mathbf{W})| e^{-cn}  \\ 
&\le   \order{\Delta q^{(R - R_\text{th})}} \sum_{n=1}^{R_\text{th}} n e^{-(c-c_0)n} \\ 
\end{align}

Thus, 
\begin{equation}
    \left|\sum_{\mathbf{W} \in \mathcal{W}_{\text{near}}} \phi(\mathbf{W}) (Z^{O_B'}_{\mathbf{W}} - Z^{O_B}_{\mathbf{W}})\right|  \leq \order{\Delta q^{(R - R_\text{th})}}
\end{equation}

For clusters with $|\mathbf{W}| > R_\text{th}$, we will not bound the difference $|Z^{O_B'}_{\mathbf{W}} - Z^{O_B}_{\mathbf{W}}|$. Instead, we bound the total contribution from large clusters directly using \autoref{lem:cluster_tail} (and choose a threshold $R_\text{th}$ accordingly):
\begin{equation}
    \left|\sum_{\mathbf{W} \in \mathcal{W}_{\text{far}}} \phi(\mathbf{W}) Z^{O_B}_{\mathbf{W}}\right| \le \order{|B| e^{-(c-c_0)(R_\text{th}+1)} }= \order{|B| e^{-(c-c_0)(R_\text{th}+1)} }.
\end{equation}

Similarly for the perturbed state, applying \autoref{lem:cluster_tail} to region $B$ in the perturbed PEPS:
\begin{equation}
    \left|\sum_{\mathbf{W} \in \mathcal{W}_{\text{far}}} \phi(\mathbf{W}) Z^{O_B'}_{\mathbf{W}}\right| \le \order{|B|e^{-(c-c_0)(R_\text{th} +1)}}.
\end{equation}

Therefore,
\begin{equation}
    \left|\sum_{\mathbf{W} \in \mathcal{W}_{\text{far}}} \phi(\mathbf{W}) (Z^{O_B'}_{\mathbf{W}} - Z^{O_B}_{\mathbf{W}})\right| \le \order{|B| \cdot e^{-(c-c_0)R_\text{th}}}.
\end{equation}

Since $|\exp(x) - \exp(y)| \le e^{\max(|x|,|y|)} |x-y|$ for bounded $x,y$, and both cluster sums are bounded by $\order{|B|  \cdot e^{-(c-c_0)}}$, we have
\begin{equation}
    \left|\exp\left(\sum \phi(\mathbf{W}) Z^{O_B'}_{\mathbf{W}}\right) - \exp\left(\sum \phi(\mathbf{W}) Z^{O_B}_{\mathbf{W}}\right)\right| 
\le \order{\left|\sum_{\mathbf{W}} \phi(\mathbf{W}) (Z^{O_B'}_{\mathbf{W}} - Z^{O_B}_{\mathbf{W}})\right|}.
\end{equation}
Combining the BP and cluster contributions, we have:
\begin{equation}
    \left|\langle O_B \rangle_{\psi'} - \langle O_B \rangle_{\psi}\right| 
\le   \|O_B\|_\infty e^{-R/\xi_*} +  |\langle O_B \rangle_{\text{BP}}| \left( C_1 e^{-(R - R_\text{th})/(\xi_*)} + C_2 |B| e^{-(c-c_0)R_\text{th}}\right).
\end{equation}
where $C_1, C_2$ are constants which depend on $R_\text{th}$. The bound can be optimized over choice of the cutoff radius, here we only care about the scaling behavior, so we simply set $R_{\text{th}} = \Theta(R)$. Since $|\langle O_B \rangle_{\text{BP}}| \le \|O_B\|_\infty$ and taking $1/\xi_{**} = \min\{\xi_*, c-c_0\}/2$, we obtain
\begin{equation}
    \left|\langle O_B \rangle_{\psi'} - \langle O_B \rangle_{\psi}\right| \le \order{\|O_B\|_\infty\cdot e^{-R / \xi_{**}}}
\end{equation}

\end{proof}

We provide an alternate proof of the theorem. This is accomplished by using the \emph{same} message background of the unperturbed state on the perturbed state, as done in \cite{midha2026}. The proof follows solely through the convergence of cluster expansion, and does not explicitly invoke the algorithmic locality of the fixed points. 

\begin{proof}[Alternative Proof]
  By the cluster expansion, we have the formal expressions,
\begin{align}
\langle O_B \rangle_{\psi} &= \langle O_B \rangle_{\text{BP}} \cdot \exp\left(\sum_{\substack{\text{conn.} \, \mathbf{W} \leftarrow \LL_{AB} \\ \supp(\mathbf{W}) \cap B \neq\emptyset }} \phi(\mathbf{W}) Z^{O_B}_{\mathbf{W}}\right), \\
\langle O_B \rangle_{\psi'} &= \langle O_B \rangle_{\text{BP}}  \cdot \exp\left(\sum_{\substack{\text{conn.} \, \mathbf{W} \leftarrow \LL_{AB}  \\ \supp(\mathbf{W}) \cap B \neq\emptyset }} \phi(\mathbf{W}) Z^{O_B'}_{\mathbf{W}}\right).
\end{align}

where we use the \emph{same} set of messages for the expansion of the perturbed network as well, resulting in modifying the set of excitations to $\LL_A \to \LL_{AB}$. This requires the perturbation to preserve $\varepsilon < 1/D$ to ensure the positivity of the BP normalization, which is implied by the fact that the messages are full rank for $\varepsilon < 2/D$ [by \autoref{lem:fullrank}].

Now, since the messages are identical in both the networks, all BP estimates and clusters not intersecting the region of perturbation $A$ evaluate identically in both the expansions. Hence, we have the difference involves clusters intersecting both $A$ and $B$, that is,

\begin{align}
  \langle O_B \rangle_{\psi'} - \langle O_B \rangle_{\psi} &= \langle O_B \rangle_{\psi} \left(\exp\left(\sum_{\substack{\text{conn.} \, \mathbf{W} \leftarrow \LL_{AB}  \\ \supp(\mathbf{W}) \cap B \neq\emptyset \\ \supp(\mathbf{W}) \cap A \neq\emptyset }} \phi(\mathbf{W}) (Z^{O_B'}_{\mathbf{W}} - Z^{O_B}_{\mathbf{W}})\right) - 1\right)
\end{align}
This difference involves \emph{connected} clusters intersecting both $A$ and $B$. The weight of any such cluster must satisfy $|\mathbf{W}| \ge R$. Hence, we bound this difference as follows:
\begin{align}
  |\langle O_B \rangle_{\psi'} - \langle O_B \rangle_{\psi}| &\le |\langle O_B \rangle_{\psi}| \left(\exp\left(\sum_{\substack{\text{conn.} \, \mathbf{W} \leftarrow \LL_{AB}  \\ \supp(\mathbf{W}) \cap B \neq\emptyset \\ \supp(\mathbf{W}) \cap A \neq\emptyset }} |\phi(\mathbf{W}) (Z^{O_B'}_{\mathbf{W}} - Z^{O_B}_{\mathbf{W}})|\right) - 1\right) \\  
  &\le |\langle O_B \rangle_{\psi}|  \prod_{X\in \{\emptyset, \text{'A'}\}} \exp\left(\sum_{\substack{\text{conn.} \, \mathbf{W} \leftarrow \LL_{AB} \\ \supp(\mathbf{W})\ni A \\ \supp(\mathbf{W})\ni B }} |\phi_{\mathbf{W}} Z^{O_B^X}_{\mathbf{W}}|\right) - 1 \\ 
  &\le |\langle O_B \rangle_{\psi}|  \exp\left(C|AB| e^{-(c-c_0)(R+1)}\right) - 1 \\ 
  &\le |\langle O_B \rangle_{\psi}| C|AB| e^{-(c-c_0)(R+1)} \exp{C|AB| e^{-(c-c_0)(R+1)}} 
\end{align}

where $C=\order{1}$ and we used \autoref{lem:cluster_tail} to bound the tail. We also used the fact that $e^x - 1 \leq xe^x$. Since $|\langle O_B \rangle| \le \order{\|O_B\|}$ and taking $1/\xi = (c-c_0)$ we obtain
\begin{equation}
    \left|\langle O_B \rangle_{\psi'} - \langle O_B \rangle_{\psi}\right| \le \order{\|O_B\|_\infty\cdot e^{-R / \xi}}
\end{equation}
\end{proof}

\smsection{Injective PEPS in the wild}

\smsubsection{Random PEPS}
Given bond dimension $D$, physical dimension $d$, and a bounded degree graph $G=(V,E)$ with maximum degree $\Delta$, for each vertex $v\in V$, sample entries of 
\[
G_v : \bigotimes_{e\ni v}\C^D \to \C^d
\]
as complex standard Gaussians $z\sim \mathcal N_{\C}(0,1)$, which denotes $z = x+\iota y$, where $x,y\sim \mathcal N (0,1/2)$.

\begin{theorem}\label{thm:gaussiansm}
    Fix $\varepsilon_0,\eta \in (0,1)$. Then, if
   \begin{equation}\label{eq:gaussianlocaldim}
        d \geq \frac{64}{\varepsilon_0^2} \left(2 D^\Delta \log 9 + \log \frac{2N}{\eta}\right)
   \end{equation}
    the resulting PEPS is $\varepsilon_0-$injective with probability $\geq 1-\eta$.
\end{theorem}

\begin{lemma}\label{lem:gaussianconc}
    Let $A \in \C^{m \times n}$ be a matrix with each entry sampled independently as $A_{ij} \sim \mathcal N_{\C}(0,1/m)$. For every fixed $x \in \C^n$ with $\|x\|_2 = 1$ and every $s\in(0,1)$ we have 
    \[
    \Pr \left(\left|\| Ax\|^2_2 -1\right| > s\right) \leq 2\exp{\left(-m s^2/4\right)} 
    \] 
\end{lemma}

\begin{proof}
    Each component of $\sqrt m A x$ is drawn from $\mathcal N_{\C}(0,1)$, and distinct components are independent. As a result, 
    \[
    m\cdot \|A x\|_2^2 = \sum_{i=1}^m Z_i
    \]
    with $Z_i \sim \text{Exp}(1)$ are independent variables, with $\E[e^{tZ_i}]=1/(1-t)$ for $t <1 $. For any $0 < t < 1/2$ Markov's inequality gives, 
    \begin{align*}
        \Pr \left(\|Ax\|_2^2 - 1 > s\right) &= \Pr\left(e^{tm \|Ax\|_2^2} > e^{tm(1+s)}\right)\\ 
        &\leq e^{-tm(1+s)}\E[e^{tm \|Ax\|_2^2}] \\ 
        &\leq e^{-tm(1+s)} \left(\frac{1}{1-t}\right)^m \\ 
        &\leq e^{m(-t(1+s) - \log(1-t))} \\
        &\leq e^{-ms^2/4}
    \end{align*}
where we took $t=s/2$ at the last step. With a similar argument for the $  \Pr \left(\|Ax\|_2^2 - 1 < -s\right)$, we combine and obtain:
 \[
    \Pr \left(\left|\| Ax\|^2_2 -1\right| > s\right) \leq 2\exp{\left(-m s^2/4\right)}.
    \] 
\end{proof}

\begin{prop}\label{prop:opnormwitness}
    There is a finite set $\mathcal X \subset \C^n$ with $|\mathcal X| \leq 9^{2n}$ such that for every $H\in \C^{n\times n}$ Hermitian matrix, 
    \[
    \|H\|_{\infty} \leq 2\max_{x\in\mathcal X} |x^\dagger H x|
    \]
\end{prop}
\begin{proof}
    Consider an isomorphism between $\C^n$ and $\R^{2n}$. Now, take a maximal collection $\mathcal X$ of unit vectors whose pairwise distances are greater than $1/4$. Balls of radius $1/8$ centered at these points are disjoint, and contained in the ball of radius $9/8$. Comparing volumes gives:
    \[
    |\mathcal X| \cdot \left(\frac{1}{8}\right)^{2n} \leq \left(\frac{9}{8}\right)^{2n} 
    \]
    as a result, $|\mathcal X|\leq 9^{2n}$. Since the set of vectors is maximal, for every unit vector $v\in \C^n$ there exists $x\in \mathcal X$,
    \[
    \|v - x\|_2 \leq 1/4 
    \]
    Given $H$, choose a unit vector $v$ such that $|v^\dagger H v| = \|H\|_\infty$. Now,
    \begin{align*}
        \|H\|_\infty &= |v^\dagger H v| \\ 
            &\leq |x^\dagger H x| + |v^\dagger H v - x^\dagger H x| \\ 
            &\leq |x^\dagger H x| + 2\|H\|_{\infty} \|v-x\|_2 \\ 
            &\leq |x^\dagger H x| + 2\|H\|_{\infty} \frac{1}{4}
    \end{align*}
    Thus, 
    \[
        \|H\|_{\infty} \leq 2\max_{x\in\mathcal X} |x^\dagger H x|
    \]
\end{proof}
\begin{proof}[Proof of \autoref{thm:gaussiansm}]
    Let $A_v := G_v/\sqrt d$. Denoting $n=D^\Delta$, this is a $d \times n$ matrix with independent entries sampled from $\mathcal N_{\C}(0,1/d)$. Let $H_v := A_v^\dagger A_v - I_n $. If $\|H_v\|_{\infty} > t$, then from \autoref{prop:opnormwitness} we conclude that $\left|\|A_v x\|_2^2 -1\right| > t/2$ for $x\in \mathcal X$ where $\mathcal X$ has $\leq 9^{2n}$ unit vectors. \autoref{lem:gaussianconc} implies along with a union bound over the unit vectors, 
    \[
    \Pr\left(\|A_v^\dagger A_v -I_n\|_{\infty} > t\right) \leq 2 \cdot 9^{2n} e^{-d t^2/16}
    \]
    Now apply this to all the vertices $v\in V$ and take a union bound over them, 
    \[
    \Pr\left(\exists v\in V ~:~\|A_v^\dagger A_v -I_n\|_{\infty} > t\right) \leq 2N \cdot e^{2 D^\Delta \log 9 -d t^2/16}
    \]
    Take $t = \varepsilon_0/2$, then the condition on $d$ of \autoref{eq:gaussianlocaldim} implies that except with a probability of $\eta$ we have simultaneously for all vertices, 
    \[
    (1-t) I \preceq A_v^\dagger A_v \preceq (1+t) I
    \]
    As a result, for each $v$ we have $\delta(T_v)^2 \geq (1-t)/(1+t)$, and $\varepsilon(T_v)\leq 2t/(1+t) < 2t =\varepsilon_0$.
\end{proof}

\smsubsection{High temperature Gibbs states}
We now consider two-local Hamiltonians defined for a graph $G=(V,E)$, 
\[
H = \sum_{\{v,w\}\in E}h_{vw}
\]
with $\|h_{vw}\|_{\infty} \leq J$, $\supp(h_{vw}) \subseteq \{v,w\}$, and all terms pairwise commuting.  
\begin{theorem}\label{thm:gibbssm}
    Given injectivity threshold $\varepsilon_0$, there exists a constant inverse temperature $\beta_c > 0$ depending only on $d, \Delta, J, \varepsilon_0$ such that for $\beta < \beta_c$, a purification $\Psi_\beta$ has a PEPS representation with bond dimension $D=2d^2$ which is $\varepsilon_0-$injective.
\end{theorem}

\begin{lemma}\label{lem:gibbsfactor}
    For sufficiently small $\beta J$, the edge Gibbs factors have the decomposition,
    \[
    e^{-\beta h_{uv}} = \frac{c_{uv}}{D} \sum_{i=1}^D F_{u,i} \otimes F_{v,i}
    \]
    where $c_{uv}>0$, $F_{u,i}$ are positive matrices with trace $d$, satisfying 
    \[
    \|F - I\|_\infty \leq O_d(\sqrt{\beta J}).
    \]
\end{lemma}
\begin{proof}
    Let $u=1,v=2$ for concreteness, and let $W = e^{-\beta h_{12}}$ and denote $c=\Tr(W)/d^2$. Then, define the traceless matrix $S := W/c - I_{12}$, and the partial traces $A = \Tr_2 S / d$, $B=\Tr_1 S / d$. Since $\|h\|_\infty \leq J$ we have that $\|S\|_\infty \leq e^{2\beta J}-1 \leq O(\beta J)$ for sufficiently small $\beta J$. Similarly, $\|A\|_\infty, \|B\|_\infty \leq O(\beta J)$.
    
    Now define the residual matrix
    \[
    R := W/c - (I_1+A)\otimes (I_2+B).
    \]
    This has zero partial traces, $\Tr_1 R = \Tr_2 R = 0$, and $\|R\|_{\HS} = O_d(\beta J)$, where HS is the Hilbert-Schmidt norm. Applying a real singular value decomposition in the space of traceless Hermitian matrices, 
    \[
    R = \sum_{i=1}^r \lambda_i X_i \otimes Y_i
    \]
    with $r \leq d^2 -1 $, $\lambda_i \geq 0$, and $X_i, Y_i$ are traceless Hermitian matrices with $\|X_i\|_{\HS}=1$ and $\|Y_i\|_{\HS}=1$. Moreover, $\lambda_i\leq O_d(\beta J)$.

    Set $K=d^2$, and pad the SVD with zeros so it has $K$ terms. Then, for $i\in [K]$ and $\sigma \in \{\pm 1\}$ define, 
    \begin{align*}
        F_{1,(i,\sigma)} &:= I + A + \sigma \sqrt{K\lambda_i} X_i \\ 
        F_{2,(i,\sigma)} &:= I + B + \sigma \sqrt{K\lambda_i} Y_i
    \end{align*}
    Note that these have trace $d$ and $\|F_{1,(i,\sigma)} - I\|_\infty, \|F_{2,(i,\sigma)} - I\|_\infty \leq O_d(\sqrt{ \beta J})$. For sufficiently small $\beta J$ these matrices are thus positive. Now, averaging over the $\sigma$ labels gives,
    \[
    \frac{1}{2}\sum_{\sigma\in\{\pm 1\}} F_{1,(i,\sigma)} \otimes F_{2,(i,\sigma)} = (I+A)\otimes (I+B) + K\lambda_i X_i \otimes Y_i
    \]
    Now averaging over $i\in [K]$ and denoting $D=2K = 2d^2$, 
\begin{align*}
     \frac{1}{D}\sum_{i\in[K]}\sum_{\sigma\in\{\pm 1\}} F_{1,(i,\sigma)} \otimes F_{2,(i,\sigma)} &= (I+A)\otimes (I+B) + \sum_{i\in[K]}\lambda_i X_i \otimes Y_i \\
     &= (I+A)\otimes (I+B) + R \\
     &= W /c 
\end{align*}
Multiply by $c$ and conclude.
\end{proof}

\begin{proof}[Proof of \autoref{thm:gibbssm}]
    By \autoref{lem:gibbsfactor}, we have for each edge $e=\{v,w\}$,
    \[
    e^{-\beta h_{e}} = \frac{c_{e}}{D} \sum_{i=1}^D F_{u,e,i} \otimes F_{v,e,i}
    \]
    For each vertex $v\in V$, let $\bm{i}_v := (i_e)_{e\ni v}$ denote the incident virtual leg labels, and define, 
    \[
    M_v(\bm i) := \prod_{e\ni v} F_{u,e,i} .
    \]
    Since the factors $F$ on distinct edges commute and are positive, $M_v$ is positive. Define the local PEPS tensor by its action on the virtual legs by 
    \[
    T_v \ket{\bm i} := \left(\prod_{e\ni v}(c_e/D)^{1/4}\right)|\sqrt{M_v(\bm i)}\rrangle_{S_v, E_v} \otimes \ket{\bm i}_{R_v}
    \]
    where $E_v$ is a $d-$dimensional purification register, $R_v$ is an orthogonal pointer register and the vectorization satisfies, 
    $\Tr_{E_v}\left[|\sqrt{M}\rrangle \llangle \sqrt{M}|\right]  = M$.

    We can now check that this is a valid purification. Tracing out the $E,R$ registers gives, 
    \begin{align*}
        \Tr_{E,R} \ketbra{\Psi_\beta} &= \frac{\prod_{e\in E}c_e}{D^{|E|}} \sum_{(i_e)_{e\in E}} \bigotimes_{v\in V} M_v(\bm i_v) \\ 
        &= \prod_{\{v,w\}\in E}\left(\frac{c_e}{D}\sum_{i=1}^D F_{u,e,i}\otimes F_{v,e,i}\right) \\ 
        &= \prod_{\{v,w\}\in E} e^{-\beta h_{vw}} = e^{-\beta H}
    \end{align*}

    We now show strong injectivity. First recall that by~\autoref{lem:gibbsfactor} one can write 
    \[
    F_{u,e,i}=I + K_{u,e,i}
    \]
    with $\Tr K_{u,e,i}=0$ and $\|K_{u,e,i}\|_\infty \leq \eta =  O_d(\sqrt{\beta J})$. Expand $M_v(\bm i)$:
    \[
    M_v(\bm i) = \prod_{e\ni v}\left(I + K_{u,e,i}\right)
    \]
    We will first show that 
    \begin{equation}\label{eq:Mcontrol}
             d(1-O_{d,\Delta}(\beta J)) \leq \Tr M_v(\bm i) \leq d(1+O_{d,\Delta}(\beta J)).
    \end{equation}
    To that end, we can expand in terms of the number of $\eta$ contributions from the  $K_{u,e,i}$ as:
    \begin{align*}
        \left|\frac{\Tr M_v(\bm i)}{d} - 1\right| &\leq \sum_{k=2}^\Delta \binom{\Delta}{k}\eta^k \\ 
        &\leq (1+\eta)^\Delta - 1 - \Delta \eta \\ 
        &\leq O_{d,\Delta}(\beta J).
    \end{align*}
    Now, note that by orthonormality of the $R_v$ register, 
    \[
    \|T_v\ket{\bm i}\|_2 ^2 = \left(\prod_{e\ni v}\sqrt{c_e/D}\right) \Tr\left[M_v(\bm i)\right]
    \]
    As a result, 
    \[
    \frac{\sigma_{\min}(T_v)^2}{\sigma_{max}(T_v)^2} = \frac{\min_{\bm i}\Tr M_v(\bm i)}{\max_{\bm i}\Tr M_v(\bm i)}
    \]
    which is $\geq 1-O_{d,\Delta}(\beta J)$ by \autoref{eq:Mcontrol}. As a result, for any $\varepsilon_0>0$ we can pick an inverse temperature threshold $\beta_c$ depending only on $d,\Delta, J,\varepsilon_0$ such that for all $\beta < \beta_c$ the PEPS is $\varepsilon_0-$injective. 
\end{proof}

\smsubsection{Noisy quantum circuits}

\begin{theorem}
    For every fixed depth $t$, injectivity threshold $\varepsilon_0 \in (0,1)$ there exists a constant threshold $p_c< 1$ depending only on $t,\varepsilon_0$ such that for $p>p_c$ the output state admits an exact PEPS purification $\Psi_\rho$ on a subgraph of $G^2$ with a bond dimension $D \leq 4^t$ which is $\varepsilon_0-$injective.
\end{theorem}
\begin{proof}
    The depolarizing channel admits Kraus operators, 
    \[
    K_0 = \sqrt{1- \frac{3p}{4}}I, \quad K_j =\sqrt{\frac{p}{4}}\sigma_j
    \]
    for $j \in \{1,2,3\}$ where $\sigma_j$ denote the Pauli operators. Reshape the channel into a tensor $M_p: \C^2 \otimes \C^2 \to \C^4$, defined as 
    \[
    (M_p)_{\alpha;(i,o)} := (K_{\alpha;oi}).
    \]
    for $i,o\in [4]$ denoting the input and output indices, and $\alpha \in [4]$ denoting the Kraus index. Since the Pauli matrices are orthogonal, the squared singular values of this tensor are $2-3p/2$ and $p/2$ each for each nontrivial Pauli. Thus, up tonormalization, 
    \[
    r(p) I \preceq M_p^\dagger M_p \preceq I,\qquad r(p) := \frac{p}{4-3p}.
    \]
    A two qubit gate $U$ followed by the depolarizing channel on each qubit provides a tensor of the form 
    \[
    (M_p \otimes M_p)(U \otimes I)
    \]
    up toindex permutations. Here $U$ acts on the two input indices, and $I$ is on the two output indices. The unitary part does not change the singular value ratio.
    We now construct the PEPS. Consider an assignment of each noisy two qubit gate to one of its qubits in some fixed order, call this qubit the \emph{owner} of the gate. Single qubit operations are assigned to their own qubit. Further, retain all Kraus labels as physical purification registers. Absorb the initial product state tensor into the first layer, and retain the final output legs as physical indices.

    For each vertex $v$, contract all the temporal legs between tensors assigned to $v$. The wires connecting tensors with different owners remain virtual bonds. The bond dimension is at most $4^t$.

    The resulting tensors $T_v$ satisfy: 
    \[
    r(p)^{2t} I \preceq T_v^\dagger T_v \preceq I 
    \]
    Given a target $\varepsilon_0$, one can choose $p_c < 1$ such that 
    \[
    \left(\frac{p}{4-3p}\right)^{2t} \geq 1-\varepsilon_0
    \]
    for all $p > p_c$, and the resulting PEPS is $\varepsilon_0-$injective.
\end{proof}

\end{document}